\documentclass[a4paper]{article}
\usepackage{amsmath,amssymb,amsthm,mathtools}
\usepackage{tikz,graphicx}
\usepackage{tkz-euclide}
\usepackage{pict2e}
\usepackage{todonotes}
\usetkzobj{all}
\theoremstyle{plain}
\usepackage[bottom]{footmisc}

\newtheorem*{theorem*}{Theorem}
\newtheorem{theorem}{Theorem}[section] 
\newtheorem{lemma}[theorem]{Lemma}
\newtheorem{proposition}[theorem]{Proposition}
\newtheorem{corollary}[theorem]{Corollary}

\theoremstyle{definition}
\newtheorem{definition}[theorem]{Definition}
\newtheorem{remark}[theorem]{Remark}
\newtheorem{rhp}[theorem]{Riemann--Hilbert Problem}

\usepackage{tikz}
\usepackage{pict2e}
\usetikzlibrary{arrows,calc,chains, positioning, shapes.geometric,shapes.symbols,decorations.markings,arrows.meta}
\usetikzlibrary{patterns,angles,quotes}
\usetikzlibrary{decorations.markings}
\tikzset{middlearrow/.style={
			decoration={markings,
				mark= at position 0.6 with {\arrow{#1}} ,
			},
			postaction={decorate}
		}
	}
\tikzset{->-/.style={decoration={
				markings,
				mark=at position #1 with {\arrow{latex}}},postaction={decorate}}}
	
	\tikzset{-<-/.style={decoration={
				markings,
				mark=at position #1 with {\arrowreversed{latex}}},postaction={decorate}}}

\newcommand{\ds}{\displaystyle}

\numberwithin{equation}{section}

\mathtoolsset{showonlyrefs}

\DeclareMathOperator*{\Res}{Res}
\DeclareMathOperator{\supp}{supp}
\def\bigO{{\mathcal O}}

\newcommand{\lozr}{
	--++(1,1)--++(0,1)--++(-1,-1)
	--++(0,-1)
}

\newcommand{\lozd}{--++(1,1)--++(1,0)--++(-1,-1)--++(-1,0)
}

\newcommand{\lozu}{--++(1,0)--++(0,1)--++(-1,0)--++(0,-1)
}

\tikzset{
	master/.style={
		execute at end picture={
			\coordinate (lower right) at (current bounding box.south east);
			\coordinate (upper left) at (current bounding box.north west);
		}
	},
	slave/.style={
		execute at end picture={
			\pgfresetboundingbox
			\path (upper left) rectangle (lower right);
		}
	}
}

\tikzset{middlearrow/.style={
		decoration={markings,
			mark= at position 0.6 with {\arrow{#1}} ,
		},
		postaction={decorate}
	}
}

\renewcommand{\Im}{\mathop{\mathrm{Im}}}
\renewcommand{\Re}{\mathop{\mathrm{Re}}}

\newcommand{\C}{{\mathbb C}}

\usepackage[colorlinks=true]{hyperref}
\hypersetup{urlcolor=blue, citecolor=red, linkcolor=blue}

\title{\vspace{-1cm}A periodic hexagon tiling model and non-Hermitian orthogonal
polynomials}
\date{}
\author{
		C. Charlier 
				  \footnote{Department of Mathematics, 
				Royal Institute of Technology (KTH),
				Stockholm, Sweden.  Email: cchar@kth.se. Supported by the Swedish Research Council, Grant No. 2015-05430 and the European
				Research Council, Grant Agreement No. 682537.} 
		 \and M. Duits 	 \footnote{Department of Mathematics, 
		 	Royal Institute of Technology (KTH),
		 	Stockholm, Sweden.  Email: duits@kth.se.   Supported by the Swedish Research Council grant (VR) Grant no. 2016-05450 and the G\"oran Gustafsson Foundation.}
		 \and A.B.J. Kuijlaars 
		 		 \footnote{Department of Mathematics, Katholieke Universiteit Leuven, Belgium, Email: arno.kuijlaars@kuleuven.be. Supported by long term structural funding-Methusalem grant of the Flemish
		 		 	Government, and by FWO Flanders projects G.0864.16 and G.0910.20, and EOS 30889451.}	 
		 \and J. Lenells  \footnote{Department of Mathematics, 
		 	Royal Institute of Technology (KTH),
		 	Stockholm, Sweden.  Email: 	jlenells@kth.se. Supported by the European Research Council, Grant Agreement No. 682537, the Swedish Research Council, Grant No. 2015-05430, the G\"oran Gustafsson Foundation, and the Ruth and Nils-Erik Stenb\"ack Foundation.}
}

\begin{document}

	\maketitle
		%\vspace*{-25pt}
	\begin{abstract}
	
	 We study a one-parameter family of probability measures on lozenge tilings of large regular hexagons that interpolates between  the uniform measure on all possible tilings and  a particular fully frozen tiling.   The description of the asymptotic behavior  can be separated into two regimes: the low and the high temperature regime.  Our main results are  the computations of  the disordered regions in both regimes and the limiting densities of the different lozenges there.   For low temperatures, the disordered region consists of two disjoint ellipses. In the high temperature regime the two ellipses merge into a single simply connected region. At the transition from the low to the high temperature a tacnode appears. The key to our asymptotic study is a recent approach introduced by Duits and Kuijlaars providing a double integral representation for the correlation kernel. One of the factors in the integrand  is the Christoffel-Darboux kernel associated to polynomials that satisfy non-Hermitian orthogonality relations with respect to a complex-valued weight on a contour in the complex plane. We compute the asymptotic behavior of these orthogonal polynomials and the Christoffel-Darboux kernel by means of a Riemann-Hilbert analysis. After substituting the resulting asymptotic formulas into the double integral we prove our main results by  classical steepest descent arguments.

		\end{abstract}

\tableofcontents

\section{Introduction}

 We study random lozenge tilings of large regular hexagons.  We place the regular hexagon so that it has corners at $(0,0)$, $(0,N)$, $(N,2N)$, $(2N,2N)$, $(2N,N)$ and $(N,0)$ and consider tilings of the hexagon with the following three types of lozenges
	\begin{center}
		Type  $I$ \tikz[scale=.3] \draw (0,0)  \lozr; \quad  Type $II$  \tikz[scale=.3]  \draw (0,0) \lozu;\quad and  Type $III$   \tikz[scale=.3] { \draw (0,0) \lozd;},
	\end{center}
see also Figure \ref{fig:hexagon}. The vertices of the lozenges are on the integer lattice and the vertical and horizontal edges have unit length.  There are numerous ways of defining a probability measures on all possible tilings of the hexagon. In this paper, we will be interested in  the case in which the probability of  a tiling $\mathcal T $ is given by 
	$$
		\mathbb P(\mathcal T)= \frac{W(\mathcal T)}{\sum_{\widetilde{\mathcal T}} W(\widetilde{\mathcal T})},
	$$ 
	where $W$ is a weight function on all possible tilings defined by 
	$$
		W(\mathcal T)= \prod_{\tikz[scale=.2] \draw (0,0) \lozu; \  \in \mathcal T} w(\tikz[scale=.3] \draw (0,0) \lozu;)
	$$
	with 
	\begin{equation}\label{eq:definition_weight}
		w\Big(	\tikz[scale=.3,baseline=(current bounding box.center)] { \draw (0,0) \lozu; \filldraw circle(5pt);  \node[below] (i) at (0,-.2) {\tiny{$(i,j)$}};} 	\Big)=
		\begin{cases}
			\alpha, & i \text{ even,}\\
			1, &  i \text{ odd,} 
		\end{cases}
	\end{equation}
	for some fixed $\alpha \in (0,1].$ Note that if $\alpha=1$ all tilings occur with the same probability and the probability measure reduces to the uniform measure on all possible tilings. We exclude $\alpha=0$. In the limit $\alpha\downarrow 0$, there is only one possible tiling, 
	see e.g.\ 	Figure \ref{fig:hexagonExtr} below,
	 and there is no randomness. The main results in this paper concern the asymptotic behavior of the random tilings as the size of the hexagon grows large, i.e., as $N\to \infty$, and how this asymptotic behavior depends on the parameter~$\alpha$.

	\begin{figure}[t]
	\begin{center}
		\begin{tikzpicture}[scale=2.1]
		\draw[very thick] (0,0)--++(0,1)--++(1,1)--++(1,0)--++(0,-1)--++(-1,-1)--++(-1,0);
		\draw  (-.1,0.5)node[left] {N};
		\draw  (0.5,-.1)  node[below]{N};
		\draw (1.6,.5) node[right] {N};
		\end{tikzpicture}
		\includegraphics[scale=.4]{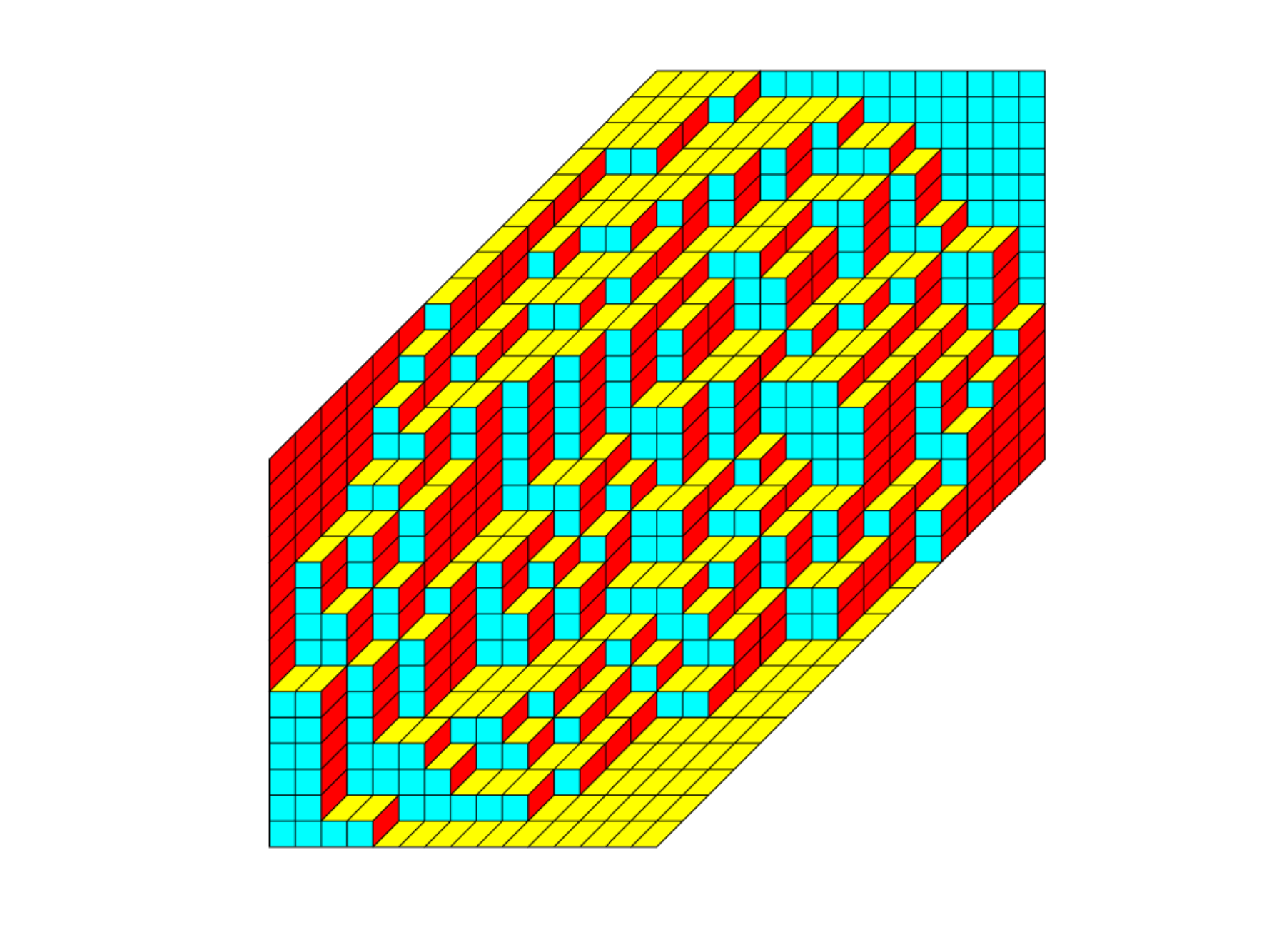}
		\caption{The hexagon (left) and an example of a tiling (right) of the hexagon by lozenges.}
		\label{fig:hexagon}
	\end{center}
\end{figure}
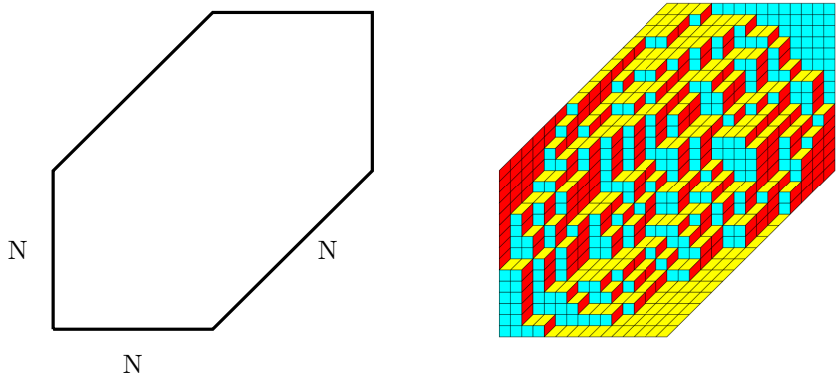

	Random tilings of  planar domains have been extensively studied in the past decades and we refer to \cite{BKMM,CEP,CLP,CKP,J05,K,KO,KOS} for important early references,  and  to  \cite{BG,Jdet,J17} for excellent introductions to the topic.  When the domains are large, the statistical properties of the tilings are expected to be described by universal  limiting processes.  In various special classes, and especially in case the random measure is a determinantal point process, tools have been developed to compute the asymptotic behavior and verify the appearance of these universal processes. For instance,   if the random measure is in the Schur class \cite{Ok1,OR1}, then we have a double integral representation for the correlation kernel  at our disposal to analyze the fine properties of the model.  Random lozenge tilings of the hexagon are however typically not in the Schur class and asymptotic studies are often more complicated. 
	
	Although not being in the Schur class, the large $N$ behavior  of  random lozenge tilings of the hexagon with the uniform measure (corresponding to $\alpha=1$ in our setup) has also been intensively studied by various authors.  Based on a representation in terms of Hahn polynomials as found in \cite{Jptrf} (see also \cite{Gorin}), the authors of \cite{BKMM} managed to perform a steepest descent analysis of the discrete Riemann--Hilbert (RH) problem for the Hahn polynomials and, consequently, describe the limiting disordered regions and the local universality laws. In \cite{Gorin}   the local universality was obtained using methods developed in  \cite{BO}. In a more general context,   uniform lozenge tilings of more complicated domains were studied by means of double integral formulas \cite{AvMJ,DM1,DM2,DM3, Petrov1,Petrov2}. 
	
	 An important part of the recent literature  on random tilings is concerned with proving the universality of the global fluctuations and the emergence of the Gaussian Free Field. For the uniform measure on all possible tilings of the hexagon there are now various techniques in the literature that prove this claim. In \cite{Petrov2}  the convergence of the global height fluctuations to the Gaussian Free Field was established using double integral formulas for the kernel. An alternative proof based on the recurrence coefficients of the Hahn polynomials was given in \cite{Duits2} extending the results on the fluctuations along vertical sections in \cite{BreuerD}. Discrete loop equations can also be used \cite{BGG} to compute the  fluctuations along vertical sections. In \cite{BuGo1,BuGo2}, another approach is introduced using the notion of a Schur generating function.  Each of these methods apply to their own general class of models and contain the uniform measure as a special case.

Measures  on  tilings of the (finite) hexagon  that are not uniform are known to be difficult to analyze asymptotically and much less results are known.  For instance, in  \cite{BGR} the authors introduced elliptic weights on the lozenge tilings, but a full asymptotic study of these models is still open. The situation $0<\alpha< 1$, which is the topic of this paper,  is a rather  gentle way to break the uniform measure. Still, the above mentioned techniques do not apply.  To study our model we will use a recently developed  new approach \cite{DK} for studying determinantal point processes that are defined via products of minors of (scalar or block) Toeplitz minors. Although the original motivation of \cite{DK} was to analyze the so-called 2-periodic Aztec diamond (see also \cite{BCJ,CJ}), the methods apply to a much wider range of  (tiling) models. The approach mainly consists of combining two important methods for asymptotic analysis: the classical steepest descent method for integrals and the Deift/Zhou steepest descent method for RH problems \cite{Deift,DZ}. This opens up new possibilities for analyzing models that were thus far out of reach and the model studied in this paper is one such example. 
	
	It is possible to take the limit of our model in which the vertical sides of the hexagon tend to infinity (see, for example, \cite{BD} for an explanation that starts from the same setting as in the present paper). In that limit, our model is the same as a 2-periodic weighting of plane partitions against a linearly shaped back wall, as studied in \cite{M1} (see also \cite{Ahn} for a generalization to the setting of Macdonald processes).  This model is then in the Schur class and thus double integral representations are available for asymptotic studies. It is important to  note that the case of a finite hexagon does not only lead to technical challenges, but also more 
	complicated phenomena occur. For instance, in our model 
	a tacnode appears for $\alpha = 1/9$.

	\begin{figure}[t]
	\begin{center}
		\includegraphics[scale=.5]{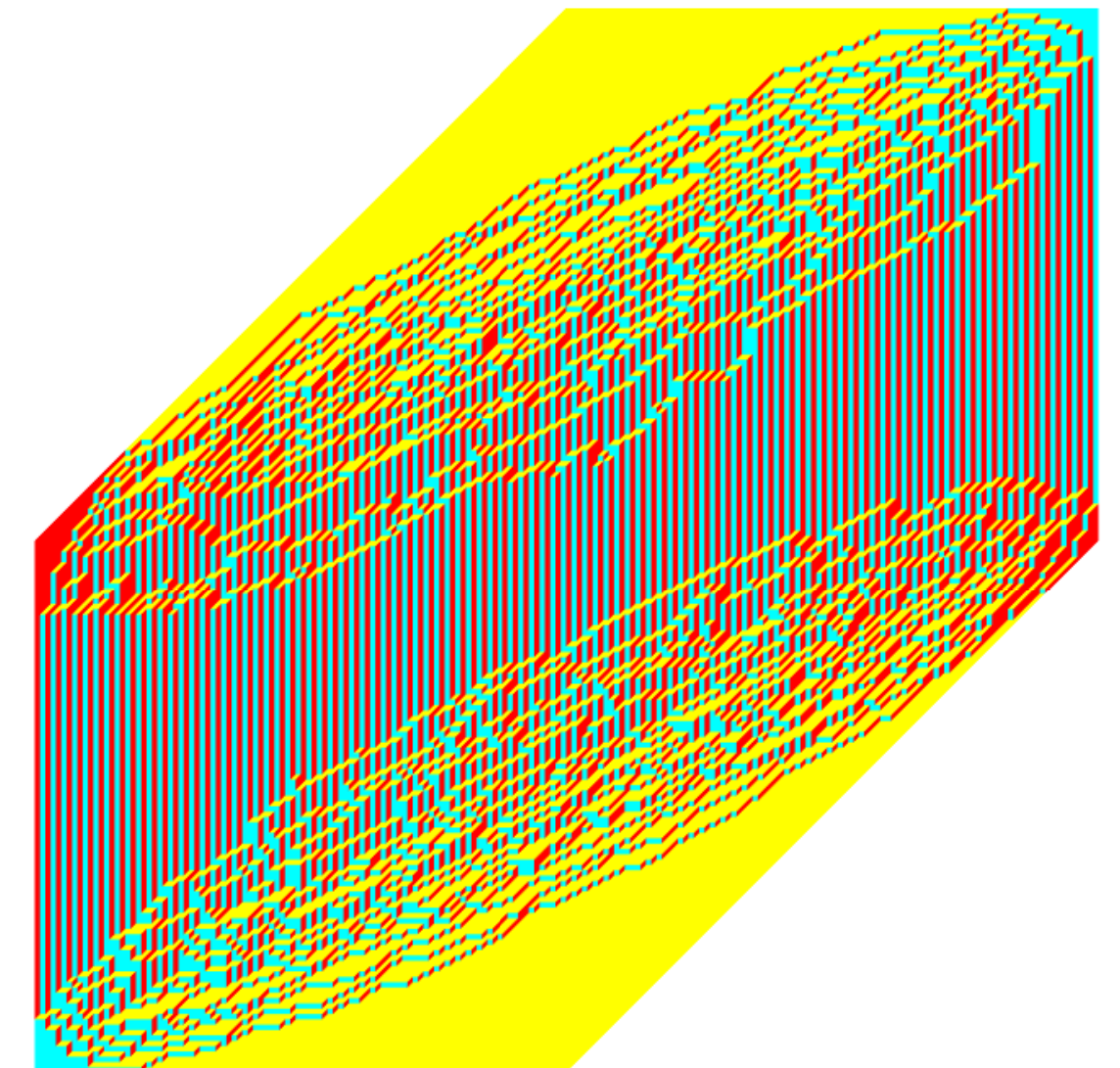}\quad 
		\includegraphics[scale=.5]{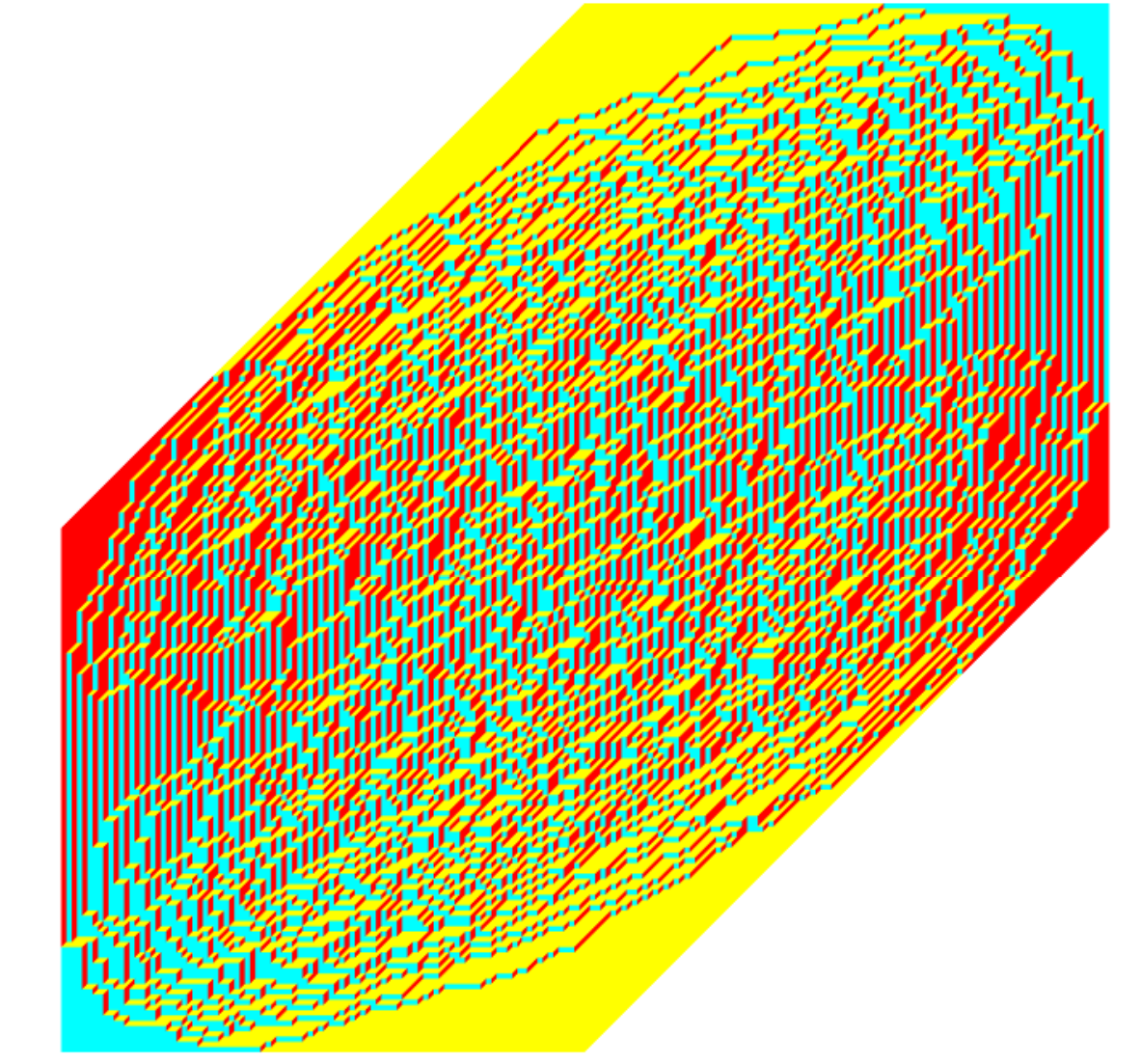}
		\caption{Two sample tilings corresponding to the low temperature (left) and high temperature (right) regimes, respectively.}		
		\label{fig:hexagonsample}
	\end{center}
\end{figure}

	In Figure \ref{fig:hexagonsample} we have plotted two sample tilings for large hexagons,  one with $0 < \alpha<\frac19$ and the other with $\frac19 < \alpha <1$. We see that for $0 < \alpha<\frac19$ there appear two clouds in which the tiling shows randomness, while it is frozen outside.  In the figure with $\frac19 < \alpha <1$, these two clouds seem to have merged.  To understand why this phenomenon is happening, it is useful to view $\alpha$ as a temperature parameter. Indeed, after defining the energy of a tiling as 
	$$
	\mathcal E(\mathcal  T) = \#\left \{ 	\tikz[scale=.3,baseline=(current bounding box.center)] { \draw (0,0) \lozu; \filldraw circle(5pt);  \node[below] (i) at (0,-.2) {\tiny{$(i,j)$}};}   \ \mid \ i \text{ even } \right \},
	$$
	we can write the weight of a tiling $\mathcal T$ as
	$W(\mathcal T)= e^{(\log \alpha) \mathcal E(\mathcal T)}$,
	and its probability as
	\[ \mathbb P(\mathcal  T) = \frac{1}{Z} e^{-\beta \mathcal E(\mathcal T)}, \qquad \beta = - \log \alpha \]
	which is a Gibbs measure with inverse temperature $\beta$.
	Thus,  $T= -\frac{1}{\log \alpha}$ may (and  will) be viewed as the temperature parameter. The low temperature limit $T \downarrow 0$ corresponds to $\alpha \downarrow 0$ and the high temperature limit $T \to \infty$ to $\alpha \uparrow 1$. 
	
	For low temperatures, the number $\mathcal E(\mathcal T)$ is expected to be small. In fact, for $T\downarrow 0$ the randomness disappears and the lozenge configurations freeze to the unique tiling with  $\mathcal E(\mathcal T)=0$. This is the tiling that is shown in the left half of Figure \ref{fig:hexagonExtr}. It can be thought of as a staircase shaped wall where the floor and the ceiling only have tiles of type III. As the temperature increases, randomness starts appearing near the interfaces where the wall meets the ceiling and the floor. For $T$ positive but small, we expect to observe two separate clouds that are far away from each other. When $T$ increases further, the clouds meet and form one cloud. Eventually, as $T\to \infty$,  the 
	model becomes the uniform measure on tilings and the cloud
	becomes the ellipse that is inscribed in the hexagon,
	as in the right part of Figure \ref{fig:hexagonExtr}.
	
	In other words, we expect that there is a critical point in the low to high temperature transition at which the topology of the disordered regime changes from being disconnected to being connected. As we will see, this transition indeed happens at $\alpha=\frac 19$. We will therefore speak of $0<\alpha<\frac 19$  as the \emph{low temperature regime} and of $\frac19 < \alpha \leq 1$ as the \emph{high temperature regime.}

	\begin{figure}[t]
	\begin{center}
		\includegraphics[scale=.5]{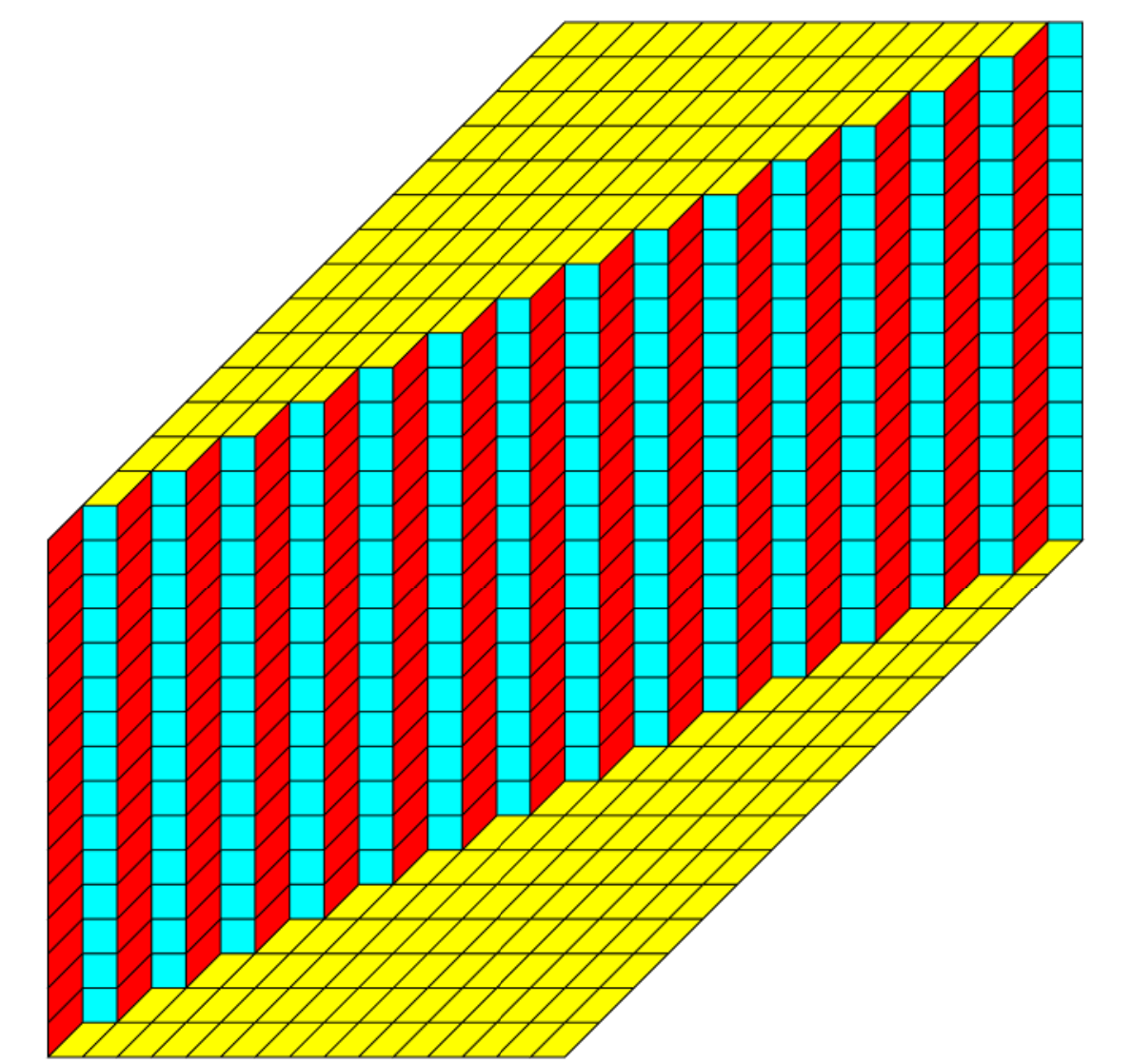} \quad  \includegraphics[scale=.5]{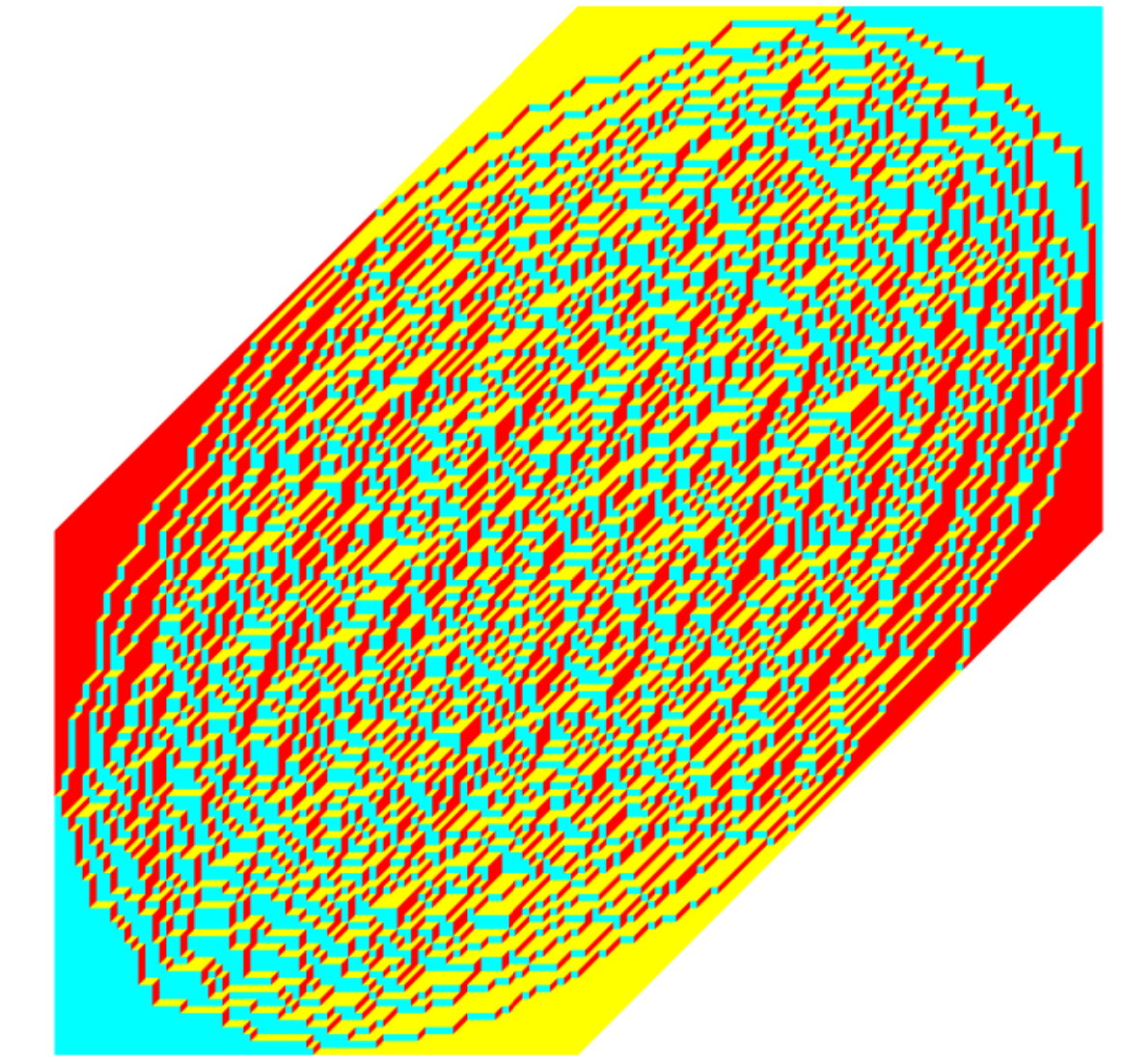}	
		\caption{The two extreme cases: $\alpha=1$ leading to the uniform measure (right) and $\alpha=0$ for which there is only one possible tiling (left).}		
		\label{fig:hexagonExtr}
	\end{center}
\end{figure}

	Our analysis follows a recent work \cite{DK}. The backbone of the approach in \cite{DK} is a connection to polynomials that satisfy an orthogonality relation
	(that could be matrix valued) on a contour 
	in the complex plane. In the present paper we will be dealing with scalar orthogonality on a closed contour $\gamma$
	going once around the origin with counterclockwise orientation.
	Let $p_n$ be the monic polynomial of degree $n$ such that 
	\begin{equation}
		\label{eq:orthogonality1}
		\frac{1}{2\pi i}
		\oint_\gamma p_n(z) z^j \frac{(z+1)^N (z+\alpha)^N}{z^{2N}}\ dz= 0, \qquad j=0,1,\ldots, n-1. 
	\end{equation}
	 It is important to note that \eqref{eq:orthogonality1} is an orthogonality condition with respect to a non-Hermitian bilinear form. It is therefore not evident that the polynomials $p_n$ are well-defined. We will prove that they are, provided that
	 $n \leq 2N$, see Proposition \ref{prop:prop51}.
	 The orthogonality \eqref{eq:orthogonality1} also changes with $N$, the size of
	 the hexagon.

	It turns out that the random tilings naturally define a determinantal point process with a  correlation kernel that can be expressed in terms of the polynomials $p_n$. For the exact statement, we need to introduce a well-known correspondence between tilings of the hexagon and non-intersecting paths.  For more background on determinantal point processes, random tilings and non-intersecting paths, we refer to \cite{Jdet}. 
	
	We draw lines on two of the three types of lozenges as follows:
	\begin{center}
		\tikz[scale=.5]{ \draw (0,0)  \lozr; \draw[very thick] (0,.5)--(1,1.5); \filldraw (0,0.5) circle(3pt);  \filldraw (1,1.5) circle(3pt); }  \quad \tikz[scale=.5] { \draw (0,0) \lozu; \draw[very thick] (0,.5)--(1,.5); \filldraw (0,0.5) circle(3pt);  \filldraw (1,0.5) circle(3pt); }\quad and \quad   \tikz[scale=.5] { \draw (0,0) \lozd;}.
	\end{center}
	The paths form a collection of non-intersecting paths $\pi_j: \{0,\ldots, 2N\} \to \mathbb Z+\tfrac{1}{2}$ with initial points $\pi_j(0)=j+\tfrac 12$ and endpoints $\pi_j(2N)=N+\tfrac 12+j$  for $j=0, \ldots,N-1$. It is well-known and easy to see that there is a one-to-one correspondence between tilings of the hexagon and non-intersecting up-right paths with these initial and end configurations. The probability measure on the tilings defined in \eqref{eq:definition_weight} induces a probability measure on such collections of non-intersecting paths. The 
	Lindstr\"om--Gessel--Viennot lemma \cite{GV,L} tells us that the probability measure is proportional to
	\begin{equation}  \label{eq:probdensity}
			\prod_{m=0}^{2N-1} \det \left[T_m\left(\pi_j(m)-\tfrac12,\pi_k(m+1)-\tfrac12\right)\right]_{j,k=1}^{N},
	\end{equation}
	where the $T_m$ are $\mathbb Z \times \mathbb Z$ matrices given by 
		\begin{equation} \label{eq:Tmeven}
		T_m(x,y)=
			\begin{cases} 
				\alpha,& \text{ if } y=x,\\
				1, & \text{ if }  y=x+1, \\
				0, & \text{ otherwise},
			\end{cases}
		\end{equation}
	if $m$ is even, and 
	\begin{equation} \label{eq:Tmodd}
		T_m(x,y)=
			\begin{cases} 
				1,& \text{ if } y=x \text{ or } y = x+1,\\
				0, & \text{ otherwise},
			\end{cases}
		\end{equation}
	if $m$  is odd. 
	The probability \eqref{eq:probdensity} is a determinantal point process 
	with a correlation kernel given by the Eynard--Metha formula \cite{EM}.
	
	In case the $\mathbb Z \times \mathbb Z$ matrices $T_m$ in \eqref{eq:probdensity}
	are (scalar or block) Toeplitz matrices, the paper \cite{DK} 
	gives a double contour integral formula for the correlation kernel, which involves
	the (scalar or block) symbols of the Toeplitz matrices as well as a
	reproducing kernel for (scalar or matrix-valued) orthogonal polynomials,
	see also \cite{BeD}.
	
	The matrices \eqref{eq:Tmeven} and \eqref{eq:Tmodd} are infinite Toeplitz
	matrices with only two non-zero diagonals. Their respective symbols are
	$z+\alpha$ and $z+1$. Both Toeplitz matrices appear $N$ times in the
	product \eqref{eq:probdensity} and this accounts for the orthogonality
	measure in \eqref{eq:orthogonality1}. Then the general formula
	in \cite{DK} reduces to the following in the special
	situation of this paper.
	  
	\begin{proposition} \label{proposition1.1}
	Let $\alpha \in (0,1]$ and let $k \geq 1$ be an integer. Then for integers
	$x_1, \ldots, x_k$, $y_1, \ldots, y_k$, with $(x_i,y_i) \neq (x_j,y_j)$
	if $i \neq j$,  we have
	\begin{equation} \label{prob point process}
		\mathbb P\left[ \begin{array}{l} \text{paths go through each of the points } \\ (x_1,y_1+\tfrac 12), \ldots,  (x_k,y_k+\tfrac 12)
		\end{array} \right]
			=\det\left[ K_N(x_i,y_i,x_j,y_j)\right]_{i,j=1}^k,
	\end{equation}
	where the kernel $K_N$ is given by
	\begin{multline}\label{eq:kernel}
		K_N(x_1,y_1,x_2,y_2)=-\frac{\chi_{x_1>x_2}}{2 \pi i} \oint_\gamma (z+1)^{\lfloor \frac{x_1}{2} \rfloor -\lfloor \frac{x_2}{2} \rfloor} (z+\alpha)^{\lfloor \frac{x_1+1}{2} \rfloor -\lfloor \frac{x _2+1}{2}\rfloor} \frac{dz}{z^{y_1-y_2+1}}  \\
	+	\frac{1}{(2\pi i)^2}  \oint_\gamma \oint_\gamma   R_N(w,z) \frac{(w+1)^N(w+\alpha)^N }{w^{2N}}  
		 \frac{(z+1)^{\lfloor \frac{x_1}{2}\rfloor}(z+\alpha)^{\lfloor \frac{x_1+1}{2}\rfloor }}{(w+1)^{\lfloor \frac{x_2}{2} \rfloor}(w+\alpha)^{\lfloor \frac{x_2+1}{2}\rfloor}} \frac{w^{y_2}}{z^{y_1+1}} dz dw,
	\end{multline}
	for $y_1,y_2 \in \mathbb Z$ and $x_1,x_2 \in \{1,\ldots,2N-1\}$. 
	Here $\lfloor x\rfloor$ denotes the largest integer $\leq x$ as usual, $\chi_{x_1 > x_2} = 1$ if $x_1>x_2$ and $0$ otherwise, $\gamma$ is a closed 
	contour that goes once around $0$ in counterclockwise direction,
	  and $R_N(w,z)$ is the $N$th  Christoffel-Darboux kernel 
	  for the orthogonal polynomials $p_n$ defined by
		\begin{align} \nonumber 
			R_N(w,z) & = \sum_{n=0}^{N-1} \frac{p_n(w)p_n(z)}{\kappa_n} \\
	  		& = \kappa_{N-1}^{-1}\frac{p_N(z) p_{N-1}(w)-p_N(w)p_{N-1}(z)}{z-w}\label{eq:CDkernel}
		\end{align}
	and 
	\begin{equation} \label{eq:kappan}
		\kappa_n = \frac 1 {2 \pi i} \oint_\gamma( p_n(z))^2 \frac{(z+1)^N(z+\alpha)^N}{z^{2N}} \ dz,
		\end{equation}
	is the squared `norm' of $p_n$. 
	\end{proposition}
\begin{proof}
	This is a special case of \cite[Theorem 4.7]{DK}, but for convenience of the reader
	we give more details on how to make the identification in the Appendix.
\end{proof}

The above proposition is the starting point of our analysis. Clearly, to analyze the limiting behavior of the probabilities \eqref{prob point process} it suffices to compute the asymptotic behavior of the kernel $K_{N}$ in \eqref{eq:kernel} as $N\to \infty$. To this end, we first compute the asymptotic behavior of the Christoffel-Daroux kernel $R_N$ corresponding to the orthogonal polynomials using Riemann-Hilbert techniques. After inserting the resulting asymptotics of $R_N$ into \eqref{eq:kernel}, we compute the asymptotic behavior of $K_{N}$ by a saddle point analysis. It should not come as a surprise to the experienced reader that there many possible fallpits and  one may view  the  fact that this approach can indeed be carried out  as the main result of our paper. With this approach one can, in principle, compute all fine asymptotic properties of the model.  In an effort to  limit the length of the paper, we restrict our main results to the description of the disordered region and the densities of the different types of lozenge there. We will though briefly comment on possible other limiting results  that are within reach.

\section{Statement of results} 
In this section we state our main results. The proofs are postponed to later sections. 
\subsection{Preliminaries}
Our main result concerns the limiting densities of the lozenges as the size of the hexagon goes to infinity. 
We introduce the scaled variables $(\xi,\eta)$ in the large $N$ limit by
\begin{equation} \label{eq:scaled_variables}
\begin{cases}
\frac{x}{N}\to  1+\xi,\\
\frac{y}{N}\to  1+ \eta,
\end{cases}
\end{equation}
where the point $(\xi,\eta)$ belongs to the hexagon
\begin{equation} \label{eq:Hhexagon}
	\mathcal H= \left\{(\xi,\eta) \mid  -1\leq \xi \leq 1, \ -1 \leq  \eta \leq 1,\ -1\leq  \eta-\xi \leq 1 \right\}.
\end{equation}

We will study the following probabilities
\begin{equation}  \label{eq:tileprobs}
	\mathbb P\left(\tikz[scale=.3,baseline=(current bounding box.center)] {\draw (0,-1) \lozr; \filldraw (0,-1) circle(5pt); \draw (0,-1) node[below] {$(x,y)$}} \right), 
	\quad
	\mathbb P\left(\tikz[scale=.3,baseline=(current bounding box.center)] {\draw (0,0) \lozu; \filldraw (0,0) circle(5pt); \draw (0,0) node[below] {$(x,y)$} }\right), 
	\quad \textrm{ and }  \quad 
	\mathbb P\left(\tikz[scale=.3,baseline=(current bounding box.center)] {\draw (0,0) \lozd; \filldraw (1,0) circle(5pt); \draw (1,0) node[below] {$(x,y)$}} \right).
\end{equation}
Here $(x,y)$ is the coordinate for the black dot. From simple geometric considerations, 
we note that these probabilities add up to $1$.  
Our main result, Theorem \ref{thm:main} below, gives the limits
of the probabilities \eqref{eq:tileprobs} under the scaling \eqref{eq:scaled_variables}
provided that $(\xi,\eta)$ belongs to the liquid region. The result is
stated in terms of a saddle point for
the double contour integral in \eqref{eq:kernel}. 
The saddle points turn out to be solutions of an algebraic equation
\begin{equation} \label{eq:saddlepointeq} 
	\left( \frac{\xi}{2} \left(\frac{1}{z+1} + \frac{1}{z+\alpha} \right)	
	- \frac{\eta}{z} \right)^2 = Q_{\alpha}(z) \end{equation}
with a rational function $Q_{\alpha}$ that we describe
next. 	The liquid region $\mathcal L_{\alpha}$ is 
characterized by the property that \eqref{eq:saddlepointeq} has 
a solution $z= s(\xi,\eta;\alpha)$ in the upper half plane.

\subsection{The rational function $Q_{\alpha}$} \label{sec:zeta}

The rational function $Q_{\alpha}$ will arise from 
the equilibrium problem associated with the varying
weight $\frac{(z+1)^N(z+\alpha)^N}{z^{2N}}$ that we will analyze
in  Section \ref{gfunctionsec} below. Here we state the
formulas that come out of this analysis and we refer to Section \ref{gfunctionsec}
for motivation why  indeed $Q_{\alpha}$ is relevant to our problem.
The definition of $Q_{\alpha}$ is different for the two
cases $\alpha \leq \frac{1}{9}$ and $\alpha  \geq \frac{1}{9}$
and this reflects the phase transition at $\alpha = \frac{1}{9}$.

\begin{definition} \label{def:Qa}
For each $0 \leq \alpha \leq 1$, we define two complex numbers $z_\pm(\alpha)$ 
and a rational function $Q_\alpha$ as follows:
	\begin{enumerate}
		\item[(a)] For $\frac{1}{9} \leq \alpha \leq 1$, we let
		\begin{equation} \label{eq:zpmhigh} 
		z_{\pm}(\alpha) =  
		-\frac{3 - 2 \sqrt{\alpha} + 3\alpha}{8}
		\pm \frac{3i \left(1+\sqrt{\alpha}\right)}{8} \sqrt{\left(1- \tfrac{\sqrt{\alpha}}{3}\right)
			\left(3 \sqrt{\alpha}-1\right)} \end{equation}
			and
		\begin{equation} \label{eq:Qalphahigh} 
		Q_{\alpha}(z) = \frac{\left(z+\sqrt{\alpha}\right)^2 (z-z_+(\alpha))(z-z_-(\alpha))}{z^2(z+1)^2(z+\alpha)^2}. \end{equation}
		
		\item[(b)] For $0 \leq \alpha \leq \frac{1}{9}$, we let
		\begin{equation} \label{eq:zpmlow} 
		z_{\pm}(\alpha) = - \frac{1+3\alpha}{4} \pm 
		\frac{1}{4} \sqrt{(1-\alpha)(1-9\alpha)}
		\end{equation}
		and
		\begin{equation} \label{eq:Qalphalow} 
		Q_{\alpha}(z) = \frac{(z-z_+(\alpha))^2 (z-z_-(\alpha))^2}{z^2(z+1)^2(z+\alpha)^2}.\end{equation}
	\end{enumerate}  
\end{definition}

Let us comment on how $Q_\alpha$ depends on $\alpha$ and the transition at $\alpha = \frac{1}{9}$. For $\frac{1}{9} \leq \alpha \leq 1$, 
it can be checked from \eqref{eq:zpmhigh} that 
$|z_{\pm}(\alpha)| = \sqrt{\alpha}$ and
	\begin{equation} \label{eq:zpmhigh2} 
		z_{\pm}(\alpha) = \sqrt{\alpha} e^{\pm i \theta_{\alpha}} 
	\end{equation}
for some angle $\theta_{\alpha}$ which increases from $\frac{2\pi}{3}$
to $\pi$ as $\alpha$ decreases from $1$ to $\frac{1}{9}$. 
For $0 \leq \alpha \leq \frac{1}{9}$, the numbers $z_{\pm}(\alpha)$ are real and satisfy
\[ -\frac{1}{2} < z_-(\alpha) < -\sqrt{\alpha} < z_+(\alpha) < -\alpha
\quad \text{for } 0 < \alpha < \frac{1}{9} \]
with $z_-(\alpha) z_+(\alpha) = \alpha$. 

For $\frac{1}{9} < \alpha < 1$,  the function $Q_{\alpha}$ in \eqref{eq:Qalphahigh} has one double zero and two simple zeros, whereas for $0 < \alpha < \frac{1}{9}$
it has two double zeros on the real line by \eqref{eq:Qalphalow}. For
$\alpha = \frac{1}{9}$ both \eqref{eq:zpmhigh} and \eqref{eq:zpmlow} yield $z_+(\alpha) = z_-(\alpha) = - \frac{1}{3}$, and both  \eqref{eq:Qalphahigh} 
and \eqref{eq:Qalphalow} yield
\begin{equation} \label{eq:Qalphacrit}
Q_{\alpha}(z) = 
\frac{(z+ \frac{1}{3})^4}{z^2(z+1)^2(z+\frac{1}{9})^2}  
\qquad \text{ for } \alpha = \frac{1}{9},
\end{equation}
which has a fourth order zero at $-\frac{1}{3}$.
For $\alpha = 1$, the formulas \eqref{eq:zpmhigh} and \eqref{eq:Qalphahigh} reduce to  
\begin{equation} \label{eq:Qalphais1} 
Q_{\alpha}(z) = \frac{z^2 + z+1}{z^2(z+1)^2}
	\qquad \text{ for } \alpha = 1, 
\end{equation}
and 
$z_{\pm}(1) = -\frac{1}{2} \pm \frac{\sqrt{3}}{2} i = e^{\pm \frac{2\pi i}{3}}$.
\medskip

The function $Q_\alpha$ plays an important role in the asymptotic study of the orthogonal polynomials. The $g$-function that is used 
in the normalization of the RH problem for the orthogonal polynomials will be constructed in terms of $Q_\alpha$ as
\begin{equation} \label{eq:gz}
	g(z) = \frac{1}{\pi i} \int_{\Sigma_0} 
	\log(z-s) Q_{\alpha}^{1/2}(s) ds \end{equation} 
with 
$\Sigma_0 = \{ \sqrt{\alpha} e^{it} \mid -\theta_{\alpha} \leq t
\leq \theta_{\alpha} \}$ and $\theta_{\alpha} = \arg z_+(\alpha) \in [\frac{2\pi}{3},\pi]$. See Definition~\ref{mu0gdef} below for the
precise definition of the branches of the logarithm and the 
square root in \eqref{eq:gz}.

  The following definition is central for the saddle point analysis of the double integral in \eqref{eq:kernel}.

\begin{definition} \label{def:Xi} For each $0< \alpha \leq 1$ and $ (\xi, \eta) \in \mathcal H$, 
	we define  $\Xi_{\alpha}(z) = \Xi_{\alpha}(z;\xi,\eta)$ as any solution of the equation	
	\begin{equation} \label{eq:algebraiccurve}
\left(\Xi_{\alpha}(z)- \frac \xi 2 \left(\frac{1}{z+1}+ \frac{1 }{z+\alpha}\right)+\frac{\eta}{z} \right)^2= Q_\alpha(z).
\end{equation}
\end{definition}

\begin{figure}[t]
	\begin{center}
		\begin{tikzpicture}[xscale=.7,yscale=1](15,10)(0,0)
		
		%\filldraw[gray!20!white]  (8,1) --++(6,0) --++(1,1) --++(-6,0) --++(-1,-1);
		%\draw[very thick,gray] (8,1)--++(6,0); 
		\draw[gray!40!white] (7,0) --++(6,0) --++(2,2) --++(-6,0) --++(-2,-2);
		%\draw (11,1.5) node {$\mathcal R_2$};
		
		\draw[gray!40!white] (7,3) --++(6,0) --++(2,2) --++(-6,0) --++(-2,-2);
		%\filldraw[gray!20!white]  (8,4) --++(6,0) --++(1,1) --++(-6,0) --++(-1,-1);
		%\draw[very thick,gray] (8,4)--++(6,0); 
		%\draw (11,4.5) node {$\mathcal R_1$};		
		
		\filldraw (8.5,1)  circle (2pt);	 
		\filldraw (9.5,1)  circle (2pt);	 
		\filldraw (11,1)  circle (2pt);	 
		\filldraw   (14,1) circle (2pt);
		\draw   (8.5,1) node[below] {$-1$};
		\draw   (9.5,1) node[below] {$-\alpha$};
		\draw   (11,1) node[below] {$0$};
		\draw   (14,1) node[below] {$\infty$};
		
		\filldraw (8.5,4)  circle (2pt);	 
		\filldraw (9.5,4)  circle (2pt);	 
		\filldraw (11,4)  circle (2pt);	 
		\filldraw   (14,4) circle (2pt);
		\draw   (8.5,4) node[below] {$-1$};
		\draw   (9.5,4) node[below] {$-\alpha$};
		\draw   (11,4) node[below] {$0$};
		\draw   (14,4) node[below] {$\infty$};
		
		\end{tikzpicture}
		\begin{tikzpicture}[xscale=.7,yscale=.95](15,10)(0,0)
		
		%	\draw[very thick,gray] (8,1)--++(6,0); 
		\draw[gray!40!white] (7,0) --++(6,0) --++(2,2) --++(-6,0) --++(-2,-2);
		\draw (13,1.5) node {$ \mathcal R_{\alpha,-}$};
		\draw[very thick](9.5,.5) .. controls (8.5,1) and (8.5,1) .. (10.5,1.5);
		
		\draw[gray!40!white] (7,3) --++(6,0) --++(2,2) --++(-6,0) --++(-2,-2);
		%	\filldraw[gray!20!white]  (8,4) --++(6,0) --++(1,1) --++(-6,0) --++(-1,-1);
		%	\draw[very thick,gray] (8,4)--++(6,0); 
		\draw (13,4.5) node {$ \mathcal R_{\alpha,+}$};
		\draw[very thick](9.5,3.5) .. controls (8.5,4) and (8.5,4) ..(10.5,4.5);
		
		\draw[dashed,help lines] (9.5,.5)--(9.5,3.5);	
		\draw[dashed,help lines] (10.5,1.5)--(10.5,4.5);
		\draw (9.5,0.5) node[below] {$z_-(\alpha)$};		
		\draw (9.5,3.5) node[below] {$z_-(\alpha)$};		
		\draw (10.5,1.5) node[above] {$z_+(\alpha)$};		
		\draw (10.5,4.5) node[above] {$z_+(\alpha)$};		
		
		\filldraw (9.5,0.5) circle(1.5pt);	
		\filldraw  (9.5,3.5) circle(1.5pt);	
		\filldraw  (10.5,1.5) circle(1.5pt);		
		\filldraw (10.5,4.5) circle(1.5pt);	
		
		\filldraw (8.5,1)  circle (2pt);	 
		\filldraw (9.5,1)  circle (2pt);	 
		\filldraw (11,1)  circle (2pt);	 
		\draw   (8.5,1) node[below] {$-1$};
		\draw   (9.5,1) node[below] {$-\alpha$};
		\draw   (11,1) node[below] {$0$};
		\draw   (14,1) node[below] {$\infty$};
		\filldraw   (14,1) circle (2pt);
		
		\filldraw (8.5,4)  circle (2pt);	 
		\filldraw (9.5,4)  circle (2pt);	 
		\filldraw (11,4)  circle (2pt);	 
		\draw   (8.5,4) node[below] {$-1$};
		\draw   (9.5,4) node[below] {$-\alpha$};
		\draw   (11,4) node[below] {$0$};
		\draw   (14,4) node[below] {$\infty$};
		\filldraw   (14,4) circle (2pt);
		
		\end{tikzpicture}
		\caption{On the right, the two-sheeted Riemann surface for the high temperature
		case $\frac19 < \alpha \leq 1$ is displayed. The function $\Xi_{\alpha}$ is meromorphic
	on the Riemann surface with simple poles at the indicated points
$-1$, $-\alpha$, $0$ on both sheets 
	and a simple zero at both points at $\infty$. In the low temperature case $0< \alpha< \frac19$, the cuts from $z_+(\alpha)$ to $z_-(\alpha)$ disappear and  the surface decouples, resulting in the picture that is displayed at the left. }
		\label{fig:riemannsurfaces}
	\end{center}
\end{figure}
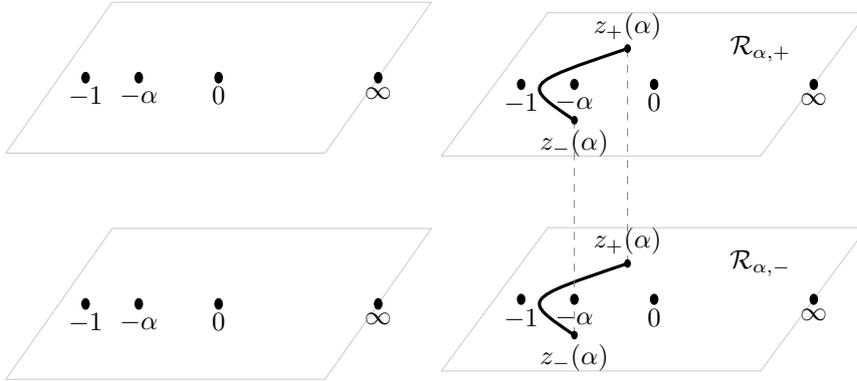
In the low temperature regime $0 < \alpha < \frac 19$, we see from
\eqref{eq:Qalphalow} that 
 $Q_{\alpha}$ is the square of a rational function. This means that \eqref{eq:algebraiccurve}  factorizes and $\Xi_{\alpha}$  decouples into two rational functions with poles at $-1,-\alpha,0$ and a zero at $\infty$. This in turn implies that we obtain two well-defined rational functions $\Xi_{\alpha,\pm}$ from \eqref{eq:algebraiccurve}: 
\begin{equation} \label{eq:defz1z2low} \begin{aligned}
\Xi_{\alpha,\pm}(z) & = \pm \left(Q_\alpha(z)\right)^{\frac12}+\frac{\xi}2\left(\frac{1 }{z+1}+\frac{1 }{z+\alpha}\right) -\frac{\eta}{z} \\
	& = \pm \frac{(z-z_+(\alpha))(z-z_-(\alpha))}{z(z+1)(z+\alpha)}
		+ \frac{\xi}2\left(\frac{1 }{z+1}+\frac{1 }{z+\alpha}\right) -\frac{\eta}{z}. 
\end{aligned}
\end{equation}

 $\Xi_{\alpha}$ then is a meromorphic function defined on the Riemann surface $\mathcal R_{\alpha}$ associated with the equation $w^2=(z-z_+(\alpha))(z-z_-(\alpha))$.   It has two sheets 
 $\mathcal R_{\alpha,\pm}$,
 that are connected by a cut from $z_+(\alpha)$ to $z_-(\alpha)$
 that we choose as
 \[ \mathcal C = \{(w,z) \in \mathcal R_{\alpha} \mid |z| = \sqrt{\alpha}, \,
 \theta_{\alpha} \leq |\arg z| \leq \pi \}, \]
 where we recall from \eqref{eq:zpmhigh2} that $\theta_{\alpha} = \arg z_+(\alpha) =  - \arg z_-(\alpha)$.
 We take $w = ((z-z_+)(z-z_-))^{1/2}$ with the branch of the
 square root that behaves like $z$ as $z \to \infty$
 on the first sheet $\mathcal R_{\alpha,+}$ and that behaves like $-z$
 as $z \to \infty$ on the second sheet.

Accordingly we have two branches of $\Xi_{\alpha}$,
 \begin{align} \label{eq:defz1z2high}
 \Xi_{\alpha,\pm}(z)&= \pm  Q_\alpha (z)^{1/2} + \frac\xi2 \left(\frac{1}{z+1}+\frac{1}{z+\alpha}\right)-\frac{\eta}{z},  \\
 	& = \frac{(z+\sqrt{\alpha})w}{z(z+1)(z+\alpha)} + \frac\xi2 \left(\frac{1}{z+1}+\frac{1}{z+\alpha}\right)-\frac{\eta}{z},
 	\quad (w, z) \in \mathcal R_{\alpha, \pm}, 
 \end{align}  
 see also Figure \ref{fig:riemannsurfaces}. The function $\Xi_{\alpha}$ is meromorphic
 on the Riemann surface with simple poles at 
 $-1$, $-\alpha$, $0$ on both sheets 
 and a simple zero at both points at $\infty$. The four remaining
 zeros will be the saddle points for the double contour
 integral.

\subsection{Saddle points and the liquid region}
We next describe the liquid region for general $0< \alpha\leq  1$. 
A reader acquainted with the asymptotic analysis of similar models for which the kernel can be represented in terms of double integral formulas, will recall that the liquid region  in such cases is defined in terms of the saddle points of a phase function occurring in the integrand (see for example \cite{BF,Duits1,Ok2,Petrov1}). In the present situation, the function $\Xi_{\alpha}$ from \eqref{eq:defz1z2low}, \eqref{eq:defz1z2high} plays the role of the derivative of the phase
function, which  now turns out to be multivalued. The saddle points are the zeros of $\Xi_{\alpha}$. As was the case in previous works, we are interested in the particular saddle with strictly positive imaginary part (if it exists). 

\begin{proposition}\label{prop:saddle}
	Let $0< \alpha \leq 1$ and $(\xi, \eta) \in \mathcal H $. Then there exists at most one solution  $z=s(\xi, \eta; \alpha)$ to $\Xi_{\alpha}(z;\xi,\eta)=0$ in $\mathbb C^+=\{z \in \mathbb C \mid \Im z>0\}$.
\end{proposition}
The proof of Proposition \ref{prop:saddle} will be given in Section~\ref{propositionproofsec}.
With this result at hand, we define the map $(\xi,\eta) \mapsto s(\xi,\eta;\alpha)$. 

	\begin{definition}\label{def:liquid} Let $0 < \alpha \leq 1$.
We define the {\bf liquid} region $\mathcal L_\alpha \subset \mathcal H$ 
by
\begin{equation} \label{eq:liquidLalpha}
	\mathcal L_\alpha= \left\{ (\xi,\eta) \in \mathcal{H} \mid 
	\exists z = s(\xi, \eta; \alpha) \in \mathbb{C}^+ :  \Xi_{\alpha}(z;\xi,\eta)=0 \right\}
	\end{equation}
and the map $s: \mathcal L_\alpha \to\mathbb{C}^+$ by $(\xi,\eta) \mapsto s(\xi,\eta; \alpha)$.
\end{definition}
	
\subsection{Main result}	
For a given $(\xi,\eta) \in \mathcal L_{\alpha}$ with $s = s(\xi,\eta;\alpha)$, 
let $T_1$ and $T_\alpha$ denote the triangles in $\mathbb{C}$ with vertex sets $\{-1,0, s\}$ and $\{-\alpha, 0,s\}$, respectively. As indicated in Figure \ref{fig:triangles}, the angles of $T_1$ and $T_\alpha$ are denoted by $\{\phi_1,\phi_2, \phi_3\}$ and $\{\psi_1,\psi_2, \psi_3\}$, respectively. Note that $\phi_3=\psi_3$ for any $\alpha$, but $\phi_j=\psi_j$ for $j = 1,2$ if and only if $\alpha=1$. The following is the main result of the paper.
	
	\begin{figure}[t]
		\begin{center}
				\begin{tikzpicture}[scale=.9]
				\coordinate (O) at (-4,2);
				\coordinate (A) at (-6,0);
				\coordinate (B) at (0,0);
				\draw (O)--(A)--(B)--cycle;
				
				\tkzMarkAngle[fill= gray,size=0.65cm,%
				opacity=.4](A,O,B)
				\tkzLabelAngle[pos = 0.35](A,O,B){$\phi_2$}
					\tkzLabelAngle[pos = -0.35](A,O,B){$s(\xi,\eta;\alpha)$}
				
				\tkzMarkAngle[fill= gray,size=1cm,%
				opacity=.4](B,A,O)
				\tkzLabelAngle[pos = 0.6](B,A,O){$\phi_1$}
				\tkzLabelAngle[pos = -0.35](B,A,O){$-1$}
				
				\tkzMarkAngle[fill= gray,size=1.5cm,%
				opacity=.4](O,B,A)
				\tkzLabelAngle[pos = 0.9](O,B,A){$\phi_3$}
				\tkzLabelAngle[pos = -0.35](O,B,A){$0$}
				\draw (-3,-.5) node {$T_1$};
			\end{tikzpicture}			
			\begin{tikzpicture}[scale=.9]
			\coordinate (O) at (-4,2);
			\coordinate (A) at (-5,0);
			\coordinate (B) at (0,0);
			\draw (O)--(A)--(B)--cycle;
			
			\tkzMarkAngle[fill= gray,size=0.65cm,%
			opacity=.4](A,O,B)
			\tkzLabelAngle[pos = 0.35](A,O,B){$\psi_2$}
			\tkzLabelAngle[pos = -0.35](A,O,B){$s(\xi,\eta;\alpha)$}
			
			\tkzMarkAngle[fill= gray,size=1cm,%
			opacity=.4](B,A,O)
			\tkzLabelAngle[pos = 0.6](B,A,O){$\psi_1$}
			\tkzLabelAngle[pos = -0.35](B,A,O){$-\alpha$}
			
			\tkzMarkAngle[fill= gray,size=1.5cm,%
			opacity=.4](O,B,A)
			\tkzLabelAngle[pos = 0.9](O,B,A){$\psi_3$}
			\tkzLabelAngle[pos = -0.35](O,B,A){$0$}
			\draw (-3,-.5) node {$T_\alpha$};
			\end{tikzpicture}
		\end{center}
	\caption{The triangles $T_1$ and $T_\alpha$.}
		\label{fig:triangles}
	\end{figure}
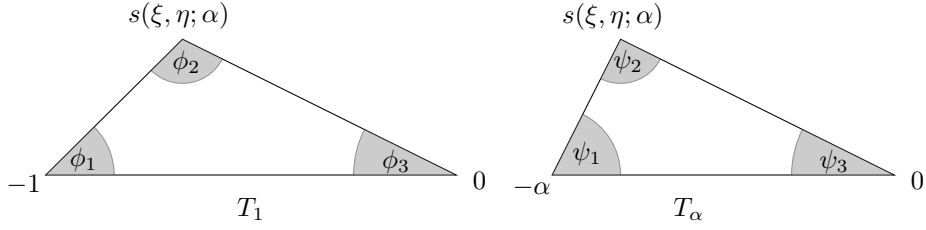
	\begin{theorem} \label{thm:main}
	Let $\alpha \in (0, 1]$. Let  $x,y\in \mathbb N$ be varying with $N$ 
	such that \eqref{eq:scaled_variables} holds
 with $(\xi,\eta) \in \mathcal L_\alpha$.
 Let $\phi_j = \phi_j(\xi,\eta;\alpha)$, $\psi_j = \psi_{j}(\xi,\eta;\alpha)$ for $j=1,2,3$ denote the angles of the
 triangles as shown in Figure \ref{fig:triangles}. Then
		\begin{equation}\label{prob paths up main thm}
		\lim_{N\to \infty}	\mathbb P\left(\tikz[scale=.3,baseline=(current bounding box.center)] {\draw (0,-1) \lozr; \filldraw (0,-1) circle(5pt); \draw (0,-1) node[below] {$(x,y)$}} \right) = \begin{cases}
			\frac{\phi_1}{\pi}, & x \text{ odd,}\\
				\frac{\psi_1}{\pi}, & x \text{ even.}
			\end{cases}
		\end{equation}
		\begin{equation}\label{prob paths hori main thm}
		\lim_{N \to \infty}	\mathbb P\left(\tikz[scale=.3,baseline=(current bounding box.center)] {\draw (0,0) \lozu; \filldraw (0,0) circle(5pt); \draw (0,0) node[below] {$(x,y)$} }\right) = \begin{cases}
		\frac{\phi_2}{\pi}, & x \text{ odd,}\\
		\frac{\psi_2}{\pi}, & x \text{ even,}
		\end{cases}
		\end{equation}
and 
	\begin{equation}\label{prob no paths main thm}
\lim_{N\to \infty}\mathbb P\left(\tikz[scale=.3,baseline=(current bounding box.center)] {\draw (0,0) \lozd; \filldraw (1,0) circle(5pt); \draw (1,0) node[below] {$(x,y)$}} \right) =\frac{\phi_3}{\pi}= \frac{\psi_3}{\pi}.
\end{equation}
\end{theorem}
Theorem \ref{thm:main} follows from Proposition \ref{prop:doubleintegrallimit} below, and the proof of
this proposition will be given in Section
\ref{sec:contour analysis}.

\medskip

Theorem \ref{thm:main} describes the situation in the liquid region $\mathcal{L}_{\alpha}$, but it also explains the behavior at the boundary of $\mathcal{L}_{\alpha}$. For each $(\xi,\eta) \in \mathcal L_\alpha$, both $s(\xi,\eta;\alpha)$ and $\overline{s(\xi,\eta;\alpha)}$ are simple zeros of $\Xi_{\alpha}$.
When the point $(\xi,\eta)$ approaches the boundary of $\mathcal L_{\alpha}$, the saddle  $s(\xi,\eta;\alpha)$ approaches the real line.  Thus, at the boundary $\partial \mathcal L_\alpha$, two zeros of $\Xi_{\alpha}$ collide 
to  form a double zero. Note also that when $s(\xi,\eta;\alpha)$ approaches the real line, the triangles $T_1$ and  $T_\alpha$ collapse with two of the angles approaching $0$ and the third approaching $\pi$. In view of Theorem \ref{thm:main}, this means that the tiling is frozen at the boundary of $\mathcal{L}_{\alpha}$.

\subsection{Structure in the low temperature regime}
Let us now discuss the low temperature regime in more detail. 
	
In the low temperature regime, each zero of $\Xi_{\alpha}$ is a zero of one of the functions $\Xi_{\alpha,+}$ or $\Xi_{\alpha,-}$ from \eqref{eq:defz1z2low}. These zeros are easy to find since each of the functions $\Xi_{\alpha,\pm}$ is as a rational function with a quadratic numerator. Setting the numerators equal to zero leads to the equations
	\begin{equation} \label{eq:saddleequationlow}
	(s-z_+)(s-z_-)
	= \pm \left[ \eta (s+1)(s+\alpha) - \xi s(s+ \tfrac{1+\alpha}{2})
	\right].
	\end{equation}
with $z_{\pm} = z_{\pm}(\alpha)$.
The equations \eqref{eq:saddleequationlow} are quadratic in $s$
with discriminants $D_{\pm} = D_{\pm}(\xi,\eta)$ that depend on
	 the coordinates $\xi$ and $\eta$:
	 \begin{equation} \label{eq:DplusDminus}
	 \begin{aligned} 
	 D_+(\xi,\eta)  & = \left(\tfrac{1+3\alpha}{2} - (1+\alpha)(\eta - \tfrac{\xi}{2})\right)^2 - 
	 4 \alpha (1-\eta)(1+\xi-\eta), \\
	 D_-(\xi,\eta)  & = \left(\tfrac{1+3\alpha}{2} +  (1+\alpha)(\eta - \tfrac{\xi}{2})\right)^2 - 
	 4 \alpha (1+\eta)(1-\xi+\eta) \\
	 & = D_+(-\xi,-\eta).  \end{aligned}
	 \end{equation}
	 The equations $D_+(\xi,\eta) = 0$, $D_-(\xi,\eta) = 0$
	 represent two ellipses in the $(\xi,\eta)$-plane. The ellipses are inside the hexagon and each one of them is
	 tangent to the boundary of the hexagon in four points.  The two ellipses are disjoint for $0 < \alpha < \frac{1}{9}$, and they become tangent at the origin for $\alpha = \frac{1}{9}$.
	 
Since a quadratic equation has two complex conjugate roots if and only if the discriminant is negative, we readily obtain the following proposition 
	 
	\begin{proposition} \label{prop:lowtemp} For each $0 < \alpha < \frac 19$, the liquid region $\mathcal{L}_\alpha$ is the disjoint union of the two open ellipses
		 $\mathcal L_\alpha^{\pm } $  defined by
		\begin{align}
		\mathcal L_\alpha^{\pm}  & =  \left\{	(\xi,\eta) \mid 
		D_{\pm}(\xi,\eta) <0\right\},
		\end{align}
		with $D_{\pm} = D_{\pm}(\xi,\eta)$ given by \eqref{eq:DplusDminus}.
		Moreover, the restrictions of $(\xi,\eta) \mapsto s(\xi,\eta;\alpha)$ to $\mathcal L_\alpha^{\pm}$ 
		are diffeomorphisms onto $\mathbb C^+$.
 	\end{proposition}
See Section \ref{propositionproofsec} for the proof, in particular of
the statement about the diffeomorphisms.

Let us now discuss the behavior of the ellipses near the boundary of the hexagon. The three poles $z=0$, $z=-\alpha$, $z=-1$ of $\Xi_{\alpha,\pm}(z)$ together with the point at infinity correspond under the map $s$ precisely to the points $(\xi,\eta)$ where the ellipses touch the hexagon, see Figure \ref{fig:low}. A computation gives the following explicit expressions for the points of tangency:
\begin{align*} A_{1,2}  &=
 \pm (-1, -\tfrac{\alpha}{1-\alpha}), &  B_{1,2} &= \pm (1, \tfrac{1-2\alpha}{1-\alpha}), \\
  C_{1,2} & = \pm  (\tfrac{1-\alpha}{1+\alpha} , 1), &  
D_{1,2} & =\pm (-\tfrac{1-\alpha}{1+\alpha}, \tfrac{2\alpha}{1+\alpha}),
 \end{align*}
  where the $+$ and $-$ signs correspond to the subscripts $1$ and $2$, respectively. 
  
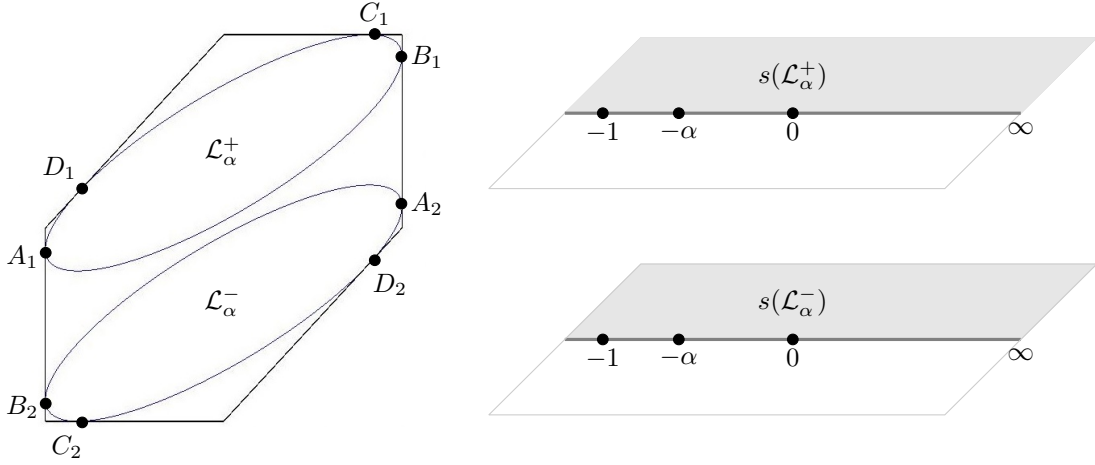
\begin{figure}[t]
	\begin{center}
	\begin{tikzpicture}
		\node (hexagonlow) at (3.5,2.5)  {\includegraphics[scale=0.3]{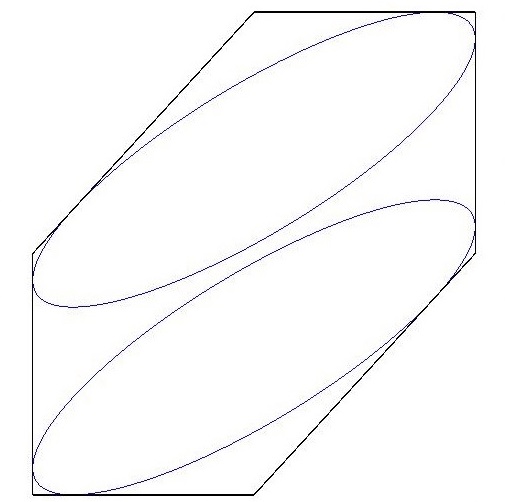}};

		 \draw (3.5,3.5) node {$\mathcal L_\alpha^{+}$};
		\draw (3.5,1.5) node  {$\mathcal L_\alpha^{-}$};
		
		\draw (1.2,2.1) node[left] {$A_1$};
		\filldraw (1.17,2.15)  circle (2pt);	 
		
		\draw  (1.35,3) node[above] {$D_1$};
		\filldraw (1.65,3)  circle (2pt);	 
		
		\draw (5.85,4.75) node[right] {$B_1$};
		\filldraw (5.85,4.75)  circle (2pt);	 
		
		\draw  (5.5,5.05) node[above] {$C_1$};
		\filldraw (5.5,5.05)  circle (2pt);	 
		
		\draw (1.2,0.1) node[left] {$B_2$};
		\filldraw (1.17,0.15)  circle (2pt);	 
		
		\draw  (1.45,-.15) node[below] {$C_2$};
		\filldraw (1.65,-.1)  circle (2pt);	 
		
		\draw (5.85,2.8) node[right] {$A_2$};
		\filldraw (5.85,2.8)  circle (2pt);	 
		
		\draw  (5.7,2.0) node[below] {$D_2$};
		\filldraw (5.5,2.05)  circle (2pt);

			\filldraw[gray!20!white]  (8,1) --++(6,0) --++(1,1) --++(-6,0) --++(-1,-1);
			\draw[very thick,gray] (8,1)--++(6,0); 
		\draw[gray!40!white] (7,0) --++(6,0) --++(2,2) --++(-6,0) --++(-2,-2);
		\draw (11,1.5) node {$s(\mathcal L_\alpha^{-})$};
		
			\draw[gray!40!white] (7,3) --++(6,0) --++(2,2) --++(-6,0) --++(-2,-2);
				\filldraw[gray!20!white]  (8,4) --++(6,0) --++(1,1) --++(-6,0) --++(-1,-1);
			\draw[very thick,gray] (8,4)--++(6,0); 
				\draw (11,4.5) node {$s(\mathcal L_\alpha^{+})$};		
				
				 \filldraw (8.5,1)  circle (2pt);	 
				\filldraw (9.5,1)  circle (2pt);	 
				\filldraw (11,1)  circle (2pt);	 
				\draw   (8.5,1) node[below] {$-1$};
				\draw   (9.5,1) node[below] {$-\alpha$};
				\draw   (11,1) node[below] {$0$};
				\draw   (14,1) node[below] {$\infty$};
				
				 \filldraw (8.5,4)  circle (2pt);	 
				\filldraw (9.5,4)  circle (2pt);	 
				\filldraw (11,4)  circle (2pt);	 
				\draw   (8.5,4) node[below] {$-1$};
				\draw   (9.5,4) node[below] {$-\alpha$};
				\draw   (11,4) node[below] {$0$};
				\draw   (14,4) node[below] {$\infty$};
				
			\end{tikzpicture}
			\caption{The liquid region (left) and the two disconnected sheets of $\mathcal{R}_{\alpha}$ (right) in the low temperature regime. The diffeomorphism $(\xi,\eta) \mapsto s(\xi,\eta;\alpha)$ maps the points $A_j$, $B_j$, $C_j$,
			$D_j$ to $-1$, $-\alpha$, $0$ and $\infty$,
			respectively.} 
			\label{fig:low}
	\end{center}
\end{figure}

Given two points $P, Q$ on one of the ellipses $\partial \mathcal{L}_\alpha^\pm$, we use the notation $\gamma_{PQ} \subset \partial \mathcal{L}_\alpha^\pm$ to denote the counterclockwise subarc of the ellipse which starts at $P$ and ends at $Q$.
As $(\xi, \eta) \in \mathcal{L}_\alpha$ approaches a point in $\gamma_{B_1C_1} \cup \gamma_{B_2C_2}$, the saddle point $s(\xi,\eta;\alpha)$ approaches a point in the interval $(-\alpha,0)$. Thus, in view of Theorem \ref{thm:main}, we see that
	\begin{equation}\label{eq:betweenBC}	\lim_{N\to \infty}	\mathbb P\left(\tikz[scale=.3,baseline=(current bounding box.center)] {\draw (0,-1) \lozu; \filldraw (0,-1) circle(5pt); \draw (0,-1) node[below] {$(x,y)$}} \right) =1,
	\end{equation}
	where $x,y$ and are such that \eqref{eq:scaled_variables} holds with $(\xi, \eta) \in \gamma_{B_1C_1} \cup \gamma_{B_2C_2}$.
This behavior extends into the frozen corners near $(\pm1,\pm1)$ where only lozenges of this type are present. Similarly, for $(\xi, \eta) \in \gamma_{C_1D_1} \cup \gamma_{C_2D_2}$,
	\begin{equation} \label{eq:betweenCD}
			\lim_{N\to \infty}	\mathbb P\left(\tikz[scale=.3,baseline=(current bounding box.center)] {\draw (0,0) \lozd; \filldraw (1,0) circle(5pt); \draw (1,0) node[below] {$(x,y)$}} \right) =1,
	\end{equation}
 and, for $(\xi, \eta) \in \gamma_{D_1A_1} \cup \gamma_{D_2A_2}$,
		\begin{equation} \label{eq:betweenDA} 
		\lim_{N\to \infty}	\mathbb P\left(\tikz[scale=.3,baseline=(current bounding box.center)] {\draw (0,-1) \lozr; \filldraw (0,-1) circle(5pt); \draw (0,-1) node[below] {$(x,y)$}} \right) =1.
\end{equation}
		
The situation is more interesting on the arcs $\gamma_{A_1B_1}$ and $\gamma_{A_2B_2}$. As $(\xi, \eta) \in \mathcal{L}_\alpha$ approaches one of these arcs, $s(\xi,\eta;\alpha)$ approaches the interval $(-1,-\alpha)$. In this limit we have $\phi_2= \pi$ and $\psi_1=\pi$, while all the other angles are zero. This means that at a point $(x,y)$ near this part of the boundary of the liquid domain, we have
 \begin{equation} \label{eq:betweenAB}
\begin{cases} \lim\limits_{N\to \infty}	\mathbb P\left(\tikz[scale=.3,baseline=(current bounding box.center)] {\draw (0,-1) \lozr; \filldraw (0,-1) circle(5pt); \draw (0,-1) node[below] {$(x,y)$}} \right) =1, & \textrm { if }  x  \textrm{ even,} \\
\lim\limits_{N\to \infty} \mathbb P\left(\tikz[scale=.3,baseline=(current bounding box.center)] {\draw (0,-1) \lozu; \filldraw (0,-1) circle(5pt); \draw (0,-1) node[below] {$(x,y)$}} \right) =1, &   \textrm { if }   x  \textrm{ odd,} 
\end{cases}
 \end{equation}
i.e., there is an alternating pattern involving two different types of lozenges, as is clearly visible in Figure \ref{fig:hexagonsample}.

\subsection{Structure in the high temperature regime}
In the high temperature regime $\frac19< \alpha \leq 1$, the equation $\Xi_{\alpha}(s;\xi,\eta)=0$ for the saddle points can be written after squaring as
\begin{equation} \label{eq:saddleequationhigh}
\left(s+ \sqrt{\alpha}\right)^2 (s-z_+)(s-z_-)
= \left(\eta (s+1)(s+\alpha)	- \xi s(s+ \tfrac{1+\alpha}{2}) \right)^2.
\end{equation}
The following proposition (which should be compared with Proposition \ref{prop:lowtemp}) shows that $s$ defines a diffeomorphism from the liquid region $\mathcal{L}_\alpha$ to the subset $\mathcal{R}_{\alpha}^+$ of $\mathcal{R}_{\alpha}$ defined by
\begin{align}\label{calRplusdef}
\mathcal{R}_{\alpha}^+ = \{(w,z) \in \mathcal R_{\alpha} \mid \Im z > 0 \}.
\end{align}

\begin{proposition}\label{prop:hightemp}
	For each $\frac19< \alpha \leq 1$, the map $(\xi,\eta) \mapsto s(\xi,\eta;\alpha)$ is a diffeomorphism from $\mathcal L_\alpha$ onto $\mathcal R_{\alpha}^+ $. Moreover, it maps the upper half
	$\mathcal L_\alpha^{+}=\left\{ (\xi,\eta) \in \mathcal L_\alpha  \mid  \eta>\frac{\xi}{2} \right\}$
	onto  $\{(w,z) \in \mathcal R_{\alpha,+} \mid \Im z>0\}$, and the lower half 
		$\mathcal L_\alpha^{-}=\left\{ (\xi,\eta) \in \mathcal L_\alpha  \mid  \eta  <\frac{\xi}{2}\right \}$
		onto   $\{(w,z) \in \mathcal R_{\alpha,-} \mid \Im z>0\}$.
\end{proposition}
Proposition \ref{prop:hightemp} is proved in Section  \ref{propositionproofsec}.

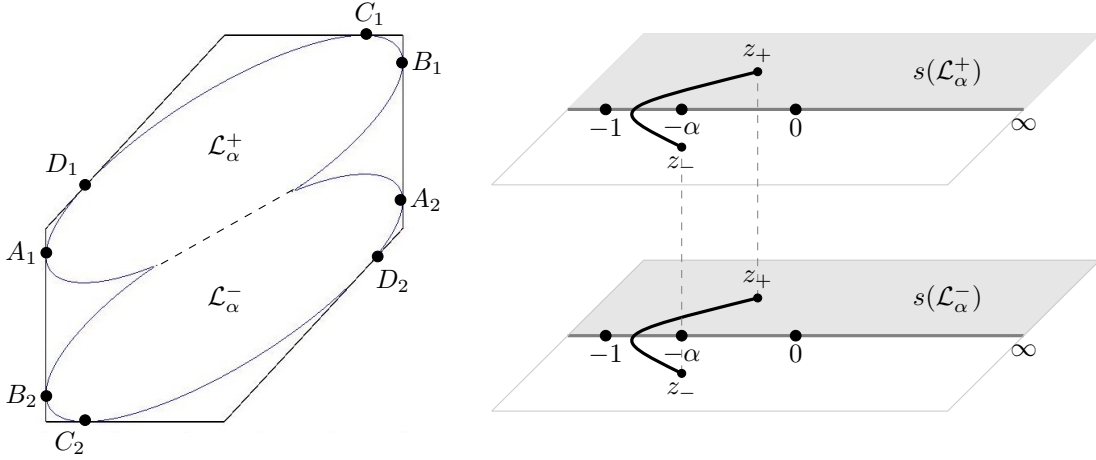
\begin{figure}[t]
	\begin{center}
		\begin{tikzpicture}(15,10)(0,0)
			\node (hexagonhigh) at (3.5,2.4)  {\includegraphics[scale=0.3]{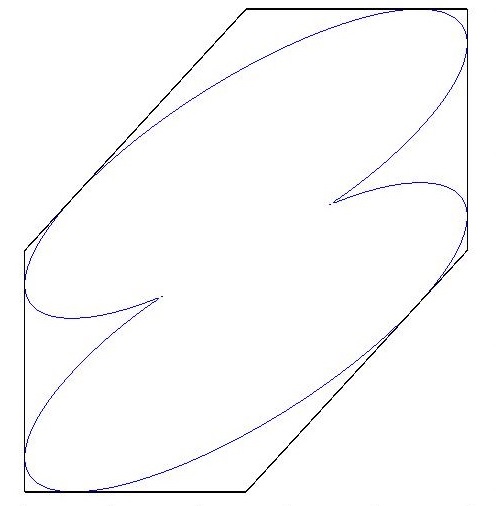}};
		
		\draw[dashed] (2.65,1.95) -- (4.5,3);
		 \draw (3.5,3.5) node {$\mathcal L_\alpha^{+}$};
		\draw (3.5,1.5) node  {$\mathcal L_\alpha^{-}$};
		
		\draw (1.15,2.1) node[left] {$A_1$};
		\filldraw (1.14,2.1)  circle (2pt);	 
		
		\draw  (1.35,3) node[above] {$D_1$};
		\filldraw (1.65,3)  circle (2pt);	 
		
		\draw (5.82,4.62) node[right] {$B_1$};
		\filldraw (5.82,4.62)  circle (2pt);	 
		
		\draw  (5.4,5.00) node[above] {$C_1$};
		\filldraw (5.35,5.00)  circle (2pt);	 
		
		\draw (1.15,0.2) node[left] {$B_2$};
		\filldraw (1.14,0.2)  circle (2pt);	 
		
		\draw  (1.45,-.15) node[below] {$C_2$};
		\filldraw (1.65,-.12)  circle (2pt);	 
		
		\draw (5.8,2.8) node[right] {$A_2$};
		\filldraw (5.8,2.8)  circle (2pt);	 
		
		\draw  (5.7,2.0) node[below] {$D_2$};
		\filldraw (5.5,2.05)  circle (2pt);	 
	
%		\draw  (4.32,2.9) node[below] {$E_1$};
%		\filldraw (4.32,2.9)  circle (2pt);	 
	
%		\draw  (2.63,1.95) node[below] {$E_2$};
%		\filldraw (2.63,1.95)  circle (2pt);	 

		\filldraw[gray!20!white]  (8,1) --++(6,0) --++(1,1) --++(-6,0) --++(-1,-1);
		\draw[very thick,gray] (8,1)--++(6,0); 
		\draw[gray!40!white] (7,0) --++(6,0) --++(2,2) --++(-6,0) --++(-2,-2);
		\draw (13,1.5) node {$s(\mathcal L_\alpha^{-})$};
		\draw[very thick](9.5,.5) .. controls (8.5,1) and (8.5,1) .. (10.5,1.5);
		
		\draw[gray!40!white] (7,3) --++(6,0) --++(2,2) --++(-6,0) --++(-2,-2);
		\filldraw[gray!20!white]  (8,4) --++(6,0) --++(1,1) --++(-6,0) --++(-1,-1);
		\draw[very thick,gray] (8,4)--++(6,0); 
		\draw (13,4.5) node {$s(\mathcal L_\alpha^{+})$};
			\draw[very thick](9.5,3.5) .. controls (8.5,4) and (8.5,4) ..(10.5,4.5);
			
			\draw[dashed,help lines] (9.5,.5)--(9.5,3.5);	
				\draw[dashed,help lines] (10.5,1.5)--(10.5,4.5);
					\draw (9.5,0.5) node[below] {$z_-$};		
						\draw (9.5,3.5) node[below] {$z_-$};		
							\draw (10.5,1.5) node[above] {$z_+$};		
								\draw (10.5,4.5) node[above] {$z_+$};		
								
					\filldraw (9.5,0.5) circle(1.5pt);	
				\filldraw  (9.5,3.5) circle(1.5pt);	
				\filldraw  (10.5,1.5) circle(1.5pt);		
			\filldraw (10.5,4.5) circle(1.5pt);	
					
					 \filldraw (8.5,1)  circle (2pt);	 
					\filldraw (9.5,1)  circle (2pt);	 
					\filldraw (11,1)  circle (2pt);	 
					\draw   (8.5,1) node[below] {$-1$};
					\draw   (9.5,1) node[below] {$-\alpha$};
					\draw   (11,1) node[below] {$0$};
					\draw   (14,1) node[below] {$\infty$};
					
					\filldraw (8.5,4)  circle (2pt);	 
					\filldraw (9.5,4)  circle (2pt);	 
					\filldraw (11,4)  circle (2pt);	 
					\draw   (8.5,4) node[below] {$-1$};
					\draw   (9.5,4) node[below] {$-\alpha$};
					\draw   (11,4) node[below] {$0$};
					\draw   (14,4) node[below] {$\infty$};
		\end{tikzpicture}
			\caption{\label{fig:high} The liquid region (left) and the two sheets of the Riemann surface $\mathcal{R}_{\alpha}$ (right) in the high temperature regime. The diffeomorphism $(\xi,\eta) \mapsto s(\xi,\eta;\alpha)$ maps the
			boundary points $A_j$, $B_j$, $C_j$ and $D_j$ to
			$-1$, $-\alpha$, $0$, and $\infty$, respectively.}
	\end{center}
\end{figure}

The boundary $\partial \mathcal L_\alpha$ of the liquid region
is part of the zero set of the discriminant of the quadratic equation \eqref{eq:saddleequationhigh}.
Since the discriminant is invariant under the map $(\xi,\eta) \mapsto (-\xi,-\eta)$, its zero set is symmetric with respect to the origin. Moreover, the zero set contains the line $\eta= \xi/2$, because \eqref{eq:saddleequationhigh} has a double zero at $s=-\sqrt{\alpha}$ when $\eta = \xi/2$. This line is however not part of the boundary of $\mathcal L_{\alpha}$. 

The discriminant also vanishes at all points $(\xi,\eta)$ which satisfy an algebraic equation of degree six. The real section of this algebraic curve is a curve inside the hexagon that touches the sides of the hexagon at the points (see Figure \ref{fig:high})
\begin{align*}
& A_{1,2} = \pm 
\left(-1, -\frac{1}{2} +  \frac{3(1-\sqrt{\alpha})}{4(1+\sqrt{\alpha})} \right), & & B_{1,2} = \pm \left(1, \frac{1}{2} +  \frac{3(1-\sqrt{\alpha})}{4(1+\sqrt{\alpha})} \right), \\
& C_{1,2} = \pm \left( \frac{5}{4} - \frac{3\sqrt{\alpha}}{2(1+\alpha)}, 1 \right), & & D_{1,2} = \pm
\left( -\frac{5}{4} + \frac{3\sqrt{\alpha}}{2(1+\alpha)}, 
- \frac{1}{4} + \frac{3 \sqrt{\alpha}}{2(1+\alpha)} \right).
\end{align*}
The liquid region is symmetric with respect to the line $\eta = \xi/2$.
The cusp points are located at 
\begin{align*}
E_{1,2} & = \pm (\xi_{cusp}, \eta_{cusp}),
\end{align*} 
where $\eta_{cusp} = \xi_{cusp}/2$ and
\begin{equation} \label{eq:xicusp}
	\xi_{cusp} =  \sqrt{\frac{5}{2} - \frac{3}{4} 
	\left(\sqrt{\alpha} + \tfrac{1}{\sqrt{\alpha}}\right)}
	= \sqrt{1- \frac{3}{4} \left( \alpha^{-1/4} - \alpha^{1/4}\right)^2}.
\end{equation}
We also have $\eta_{cusp} = \cos \frac{\theta_{\alpha}}{2}$.
Note that $\xi_{cusp} = 0$ for $\alpha = 1/9$ and $\xi_{cusp} = 1$
for $\alpha = 1$. 

At points on the subarc of the boundary $\partial \mathcal{L}_\alpha$ between $B_j$  and $C_j$ we have \eqref{eq:betweenBC}, between $C_j$ and $D_j$ we  have \eqref{eq:betweenCD}, and between $D_j$ and $A_j$ we have \eqref{eq:betweenDA}. This is a consequence of Theorem \ref{thm:main}
and it is the same as in the low temperature regime.
Finally, we have the alternating probabilities \eqref{eq:betweenAB} between $A_1$ and $B_2$, and between $A_2$ and $B_1$.

A notable difference compared with the low temperature regime is that the liquid region in the high temperature regime is connected. 
As a result, the frozen region with the two types of tiles (sometimes
called semi-frozen region)
becomes disconnected into two disjoint  components. 

 For $\alpha = 1$, the equation \eqref{eq:saddleequationhigh} has a double root at $s=-1$
	and two other roots that are the solutions of
	\[ s^2 + s+1 = (\eta(s+1) - \xi s)^2. \]
	The latter two roots coincide if $4 \xi^2 - 4 \xi \eta + 4 \eta^2 = 3$
	and this is the equation for the ellipse that is tangent to all six sides of the hexagon. The semi-frozen region disappears for $\alpha =1$. 
	
	\subsection{Local process in the bulk} 
	
We chose to present Theorem \ref{thm:main} as our main result, but we  stress that our method of proof allows us to compute much more complicated asymptotic behaviors (in this sense, our method of proof is the most important contribution of this paper). For instance, with a minor adaptation of the proof of Theorem \ref{thm:main} we compute the asymptotic behavior of local correlations in the bulk of the liquid region. 
	
	\begin{theorem} \label{thm:microsopiclimit}
		Let $0 < \alpha \leq 1$.
		For $j=1,2$, take
		\begin{align}\label{eq:microsopic_variables}
		x_j&= N \xi_N +u_j, \\
		y_j&= N \eta_N+v_j,
		\end{align}
		where $\xi_N$ and $\eta_N$ are such that $$\lim_{N\to \infty}(\xi_N,\eta_N)=(\xi,\eta)\in \mathcal L_{\alpha}$$ and $N\xi_N$ and $N \eta_N$ are integers for every $N\in \mathbb N$. We will additionally assume, without loss of generality,  that $N\xi_N$ is even. The  variables $u_1,u_2,v_1$ and $v_2$ are integer valued  local variables independent of~$N$.
		Then we have the limit 
		\begin{equation} \label{eq:HKintegral} 
		\lim_{N \to \infty} K_N(x_1,y_1,x_2,y_2) = \frac{1}{2\pi i} \int_{\overline{s}}^s  {(z+1)^{\lfloor \frac{u_1}{2}\rfloor-\lfloor \frac{u_2}{2} \rfloor}(z+\alpha)^{\lfloor \frac{u_1+1}{2}\rfloor-\lfloor \frac{u_2+1}{2}\rfloor } }\frac{ dz }{z^{v_1-v_2+1}}
		\end{equation}
		where $s = s(\xi,\eta;\alpha)$ and the integration path from $\overline{s}$ to $s$ in \eqref{eq:HKintegral}
		is in $\mathbb C \setminus (-\infty,0]$ if $u_1 \leq u_2$ and in $\mathbb C \setminus [0,\infty)$ if $u_1 >u_2$.
	\end{theorem}
The proof of this theorem is given in Section 7.

	If $u_1=u_2$ then the integral at the right-hand side of \eqref{eq:HKintegral} can be computed explicitly to be the discrete sine kernel. For general $u_1$ and $u_2$ this is thus a kernel that is an extension of that discrete sine kernel. In fact, it falls into the class of extensions of the discrete sine kernel introduced in \cite{BorodinSine}. It is to note that the limiting kernel, and thus its associated point processes, depends on $\alpha$. The periodicity in the horizontal direction is thus preserved in the limit. 
	
	Theorem \ref{thm:microsopiclimit} gives the limiting correlation kernel for the point process of the paths. However, from the path picture one can compute the correlation functions for the different lozenges. For instance, the particle/hole duality tells us that  the lozenges $\tikz[scale=.2]{ \draw (0,0) \lozd }$ form a determinantal point processes with $1- K_{N}$ as correlation kernel. Under the same assumptions of Theorem \ref{thm:microsopiclimit}  (but possibly with more  than two  points) we thus have
	\begin{equation} \label{prob point process}
	\lim_{N \to \infty} \mathbb P\left[ \begin{array}{l} \text{lozenge of the  type } \tikz[scale=.2]{ \draw (0,0) \lozd; \filldraw (1,0) circle(5pt); } \text{ at}  \\ (x_1,y_1), \ldots,  (x_k,y_k)
	\end{array} \right]
	=\det\left[1- \tilde K(u_i,v_i,u_j,v_j)\right]_{i,j=1}^k,
	\end{equation}
	where $\tilde K$ is the kernel at the right-hand side of \eqref{eq:HKintegral}. 

\subsection{Some comments on further asymptotic results}
We end this section by commenting on further possible results on the asymptotic behavior of the random tilings. 
\begin{remark}[Frozen regions] The complement of the liquid region $\mathcal L_\alpha$ inside the hexagon, is called the frozen region.  By definition, in the frozen region there are no solutions of $\Xi_{\alpha}(z;\xi, \eta)=0$ in $\mathbb C^+$ and all solutions are real.  By using a saddle point analysis similar to the one we give in the proof of Theorem \ref{thm:main}, one can show that this implies exponential decay of the fluctuations.  	 Thus, in the frozen regions the randomness disappears rapidly and the tiling converges to deterministic patterns. In the corners of the hexagon the patterns are simple in the sense that we only have one type of lozenge in each corner. For $\alpha<1$ there are also other frozen regions near the centers of the vertical sides. Also here the randomness decays rapidly, but there are two types of lozenges forming a stair case pattern (as we also see in the degenerate siuation $\alpha=0$ as shown in the left picture in Figure \ref{fig:hexagonExtr}). Frozen regions that have different types of lozenges have appeared in other models. Some examples are \cite{BD,Duits1} (after identifying Gelfand-Tsetlin patterns with lozenge tilings of the half plane). In fact,  lozenge tilings   of the infinite hexagon  (or plane partitions) with  an arbitrarily chosen back wall  have a been studied \cite{BKMM,M1,M2}. Part of this back wall  can be a frozen region with more complicated patterns than the staircase pattern of the present paper.

\end{remark}

\begin{remark}[Edge Universality] At the boundary of the liquid region (away from the points where the boundary  touches the sides of the hexagon, and, in the high temperature regime, away from the cusp points) we expect Airy behavior. There is a vast amount of literature around this type of universality, and we only refer to \cite{J17}  for an overview of results.
\end{remark}
\begin{remark}[Turning points] The turning points are the points where the boundary of the liquid region touches a side of the hexagon. Here we need to distinguish between the turning points that touch the hexagon at a vertical side from the other turning points. In both the low and high temperature regimes (assuming $\alpha < 1$) there are four such points. They separate two frozen regions: one that contains two different types of lozenges, while the other has only one type of lozenges.   We expect the local processes there to be the same as the processes that were found in (with a similar weight) in \cite{M1}. At the turning points that are not at the vertical sides of the hexagon we expect the GUE minor process \cite{JN} to appear. 
\end{remark}
\begin{remark}[Cusp points]
In the high temperature limit, the boundary of the liquid region has cusp points. Such cusp points have appeared before in the context of random tilings. It is known that the local limit process near such a cusp point is the Pearcey process \cite{AOvM, BK, OR2,TW}.  
\end{remark}
We strongly believe that all the above universal behaviors can be verified using rather straightforward  modifcations of the analysis that we present in this paper. More involved are the following remarks:
\begin{remark}[Tacnode]
	At the critical value $\alpha = \frac 19$ there is a transition from the low to high temperature regimes.
	The liquid region becomes a union of two ellipses that are tangent at the origin, and the origin is a tacnode. The  tacnode process was first characterized in \cite{AFvM} and alternative characterizations were given shortly afterwards in \cite{DKZ,Jtac}. See also \cite{AJvM,FV}. Preliminary computations indicate that  the same tacnode process appears, but we will return to this in a forthcoming paper. 
\end{remark}
\begin{remark}[Height fluctuations]
	Another interesting feature of random tilings are the fluctuation of the height function. It was found in \cite{KO} that the limiting height function can be described by the complex Burgers equation. In \cite{KO} it is also conjectured that the fluctuations are described by the Gaussian Free Field. There is by now a long list of random tiling models where  this conjecture has been verified, and we only mention \cite{BF,BuGo1,BuGo2,BuKn,Duits1,Duits2,Petrov2}. This turns out to be a very robust universality. Also in the model considered in this paper, we expect the Gaussian Free Field to appear, but with an interesting transition from the low to high temperature regimes. In the low temperature regime, the correlations between the  different ellipses are expected to converge to zero exponentially and we expect to obtain two independent Gaussian Free Fields (in the appropriate coordinates), whilst we have only one Gaussian Free Field in the high temperature regime. It is  natural 
	 to ask how these two fields merge to one in the transition from the low to high temperature regime. We plan to answer this question in a forthcoming paper. 
\end{remark}
	
\subsection{Overview of the rest of the paper}

In the next section we first prove Propositions \ref{prop:saddle}, \ref{prop:lowtemp} and \ref{prop:hightemp}.

The rest of the paper is devoted to the proof of Theorem \ref{thm:main}.
It is an asymptotic analysis of the double integral in
\eqref{eq:kernel} for $K_N(x,y,x,y)$ and for related double integrals 
that give the probabilities for each of the three lozenges.
These double integrals are presented in Theorem \ref{thm:doubleintegrals_for_lozenge_densities} below.

The asymptotic analysis has two main parts.
In the first part we analyze the orthogonal polynomials and
their reproducing kernel $R_N(w,z)$ in the large $N$ limit.
The orthogonal polynomials are characterized by a RH problem
that is essentially due to Fokas, Its and Kitaev \cite{FIK}. This is
recalled in section \ref{subsec:rhp}. The reproducing kernel
has a convenient formulation in terms of the solution
of the RH problem, see Proposition \ref{prop:RHPforRN}.
For the asymptotic analysis we use the Deift-Zhou steepest
descent method for RH problems. A main ingredient for the
analysis is the $g$-function, which in the present context
is associated with an equilibrium measure on a contour
in the complex plane.

This equilibrium measure is discussed in detail in section
\ref{gfunctionsec}. The transition at $\alpha = \frac{1}{9}$
is visible in the equilibrium measure since for $\frac{1}{9} <
\alpha \leq 1$ the equilibrium measure is supported on a
circular arc in the complex plane, while for $0 < \alpha 
\leq \frac{1}{9}$ it is supported on a full circle.
We are able to give explicit formulas for the equilibrium
measure, see Definition \ref{mu0gdef}.

The steepest descent analysis of the RH problem is done
in section \ref{sec:RHP}. We do not need strong asymptotics of
the reproducing kernel $R_N$, it suffices to have a
uniform bound on $\mathcal R_N(w,z) e^{N(g(w)-g(z))}$
(this is in Corollary \ref{cor:TandTinvsmall})
where  $\mathcal R_N(w,z)$ is a function related
to the reproducing kernel, and which is given by \eqref{eq:calRN}.

The second part of the asymptotic analysis is a saddle point
analysis of the double integrals like the one in \eqref{eq:kernel}.
The saddle points depend on the  asymptotic location $(\xi,\eta)$
in the hexagon. We focus on the lower left part of the liquid region which corresponds to $\eta \leq  \frac{\xi}{2} \leq  0$.
Then the saddle point $s = s(\xi,\eta;\alpha)$ is the zero
of the derivative of a function $\Phi_{\alpha}$ that is introduced
in section \ref{subsec:Phase}. We want to move the contours
in the double integrals to contours $\gamma_z$ and $\gamma_w$
passing through the saddles $s$ and $\overline{s}$, and such
that 
\[ \Re \Phi_{\alpha}(w) > \Re \Phi_{\alpha}(s) > \Re \Phi_{\alpha}(z) \]
whenever $w \in \gamma_{w} \setminus \{s, \overline{s}\}$
and $z \in \gamma_z \setminus \{s, \overline{s}\}$.
To be able to do the deformation we need an analysis of 
the critical level set
$\Re \Phi_{\alpha}(z) = \Re \Phi_{\alpha}(s)$ of $\Re \Phi_{\alpha}$ passing through the saddle. This is done in section \ref{subsec:levelRePhi}.

The actual deformation and splitting of contours is done in section \ref{sec:contour analysis}. It turns out that the 
limiting probabilities  in \eqref{prob paths up main thm},
\eqref{prob paths hori main thm}, \eqref{prob no paths main thm}
come from residue contributions that arise from pole crossings 
during the deformations of contours.
The remaining double contour integrals are then estimated
and we only need they tend to zero as $N \to \infty$.
The details of the deformations are different for the
low and high temperature regimes.

\section{Proofs of Propositions \ref{prop:saddle}, \ref{prop:lowtemp} and \ref{prop:hightemp}}\label{propositionproofsec}
 In this section we prove  Propositions \ref{prop:saddle}, \ref{prop:lowtemp} and \ref{prop:hightemp}. We consider the low and high temperature regimes separately. 
  
\subsection{The low temperature regime}
Since the saddle point equation $\Xi_{\alpha}(s;\xi,\eta)=0$ 
reduces to the two quadratic equations \eqref{eq:saddleequationlow} in the low temperature regime $0< \alpha< \frac 19$, and also in the critical regime
$\alpha =  \frac 19$, Proposition  \ref{prop:saddle} is straightforward to prove in this regime.
  
  \begin{proof}[Proof of Proposition \ref{prop:saddle} for $0< \alpha \leq \frac19$]
  	Any solution to $\Xi_{\alpha}(s;\xi,\eta)=0$ is a solution to one of the quadratic equations in \eqref{eq:saddleequationlow}. The discriminants for these quadratic equations are given in \eqref{eq:DplusDminus}. If, and only if, one of the discriminants is negative, then the corresponding quadratic equation has a zero in $\mathbb C^+$. Since the discriminants cannot be simultaneously negative, the statement follows. 
  \end{proof} 
  
  \begin{proof}[Proof of Proposition \ref{prop:lowtemp}]
It is clear from the discussion preceding Proposition \ref{prop:lowtemp} that $\mathcal L_\alpha= \mathcal L_{\alpha}^{+}\cup \mathcal L_{\alpha}^{-}$. It is therefore enough to show that the restrictions of $(\xi,\eta) \mapsto s(\xi,\eta;\alpha)$ to $\mathcal L_{\alpha}^{\pm}$ are diffeomorphisms onto $\mathbb{C}^+$. 
  	
We will show that for each $s$ with $\Im s>0$, there are unique points $(\xi_+,\eta_+) \in \mathcal L_{\alpha}^{+}$ and $(\xi_-,\eta_-) \in \mathcal L_{\alpha}^{-}$ such that $s= s(\xi_+,\eta_+)= s(\xi_-,\eta_-)$.  We rewrite \eqref{eq:saddleequationlow} as
  \begin{equation} \label{eq:sequationlow}
 \left( - \frac{s}{2(s+1)} - \frac{s}{2(s+\alpha)} \right)  
  	\xi + \eta = \pm \frac{(s-z_+)(s-z_-)}{(s+1)(s+ \alpha)}.
	\end{equation}
  	Since $\xi$ and $\eta$ are real, we obtain the following two real equations
  	by taking the real and imaginary parts of \eqref{eq:sequationlow}:
    \begin{equation} \label{eq:sequation2low}
\begin{pmatrix} 
  	\Re \left( - \frac{s}{2(s+1)} - \frac{s}{2(s+\alpha)} \right)  & 1 \\
  	\Im \left( - \frac{s}{2(s+1)} - \frac{s}{2(s+\alpha)} \right) & 0 
  	\end{pmatrix}
  	\begin{pmatrix} \xi \\[10pt] \eta \end{pmatrix}
  	= \pm \begin{pmatrix}  \Re  \frac{(s-z_+)(s-z_-)}{(s+1)(s+ \alpha)} \\[10pt]
  	 \Im \frac{(s-z_+)(s-z_-)}{(s+1)(s+ \alpha)}  
  	\end{pmatrix}. 
	\end{equation}
  	We readily see that 
  	\begin{equation}\label{eq:positivedeterminant}
  	\Im\left( - \frac{s}{2(s+1)} - \frac{s}{2(s+\alpha)} \right)
  	= \Im \left(-1 + \frac{1}{2(s+1)} + \frac{\alpha}{2(s+\alpha)} \right)  < 0, 
  	\end{equation}
  	for $ s\in \mathbb C^+$. Hence the $2 \times 2$ matrix on the left-hand side of \eqref{eq:sequation2low} is invertible whenever $\Im s > 0$.
  	It follows that given $s \in \mathbb C^+$ we can recover
  	$\xi_\pm$ and $\eta_\pm$ uniquely by
  	\begin{equation} \label{eq:inverseLow}
  		\begin{pmatrix} \xi \\[10pt] \eta \end{pmatrix}
  	= \pm \begin{pmatrix} 
  	\Re \left( - \frac{s}{2(s+1)} - \frac{s}{2(s+\alpha)} \right)  & 1 \\
  	\Im \left( - \frac{s}{2(s+1)} - \frac{s}{2(s+\alpha)} \right) & 0 
  	\end{pmatrix}^{-1} 
  	\begin{pmatrix} \Re  \frac{(s-z_+)(s-z_-)}{(s+1)(s+ \alpha)}  \\[10pt]
  	\Im \frac{(s-z_+)(s-z_-)}{(s+1)(s+ \alpha)} 
  	\end{pmatrix}.
  	\end{equation}
  	This proves that the restrictions of $s$  to $\mathcal L_\alpha^{\pm}$ are bijections onto $\mathbb C^+$. The differentiability is also clear, and thus we have proved the statement. 
  \end{proof}
  
  \subsection{The high temperature regime}
  We now consider the high temperature regime and thus assume $\frac{1}{9} < \alpha \leq 1$.
We start by defining the polynomial $\Pi_{\alpha}$ by 
  \begin{equation} \label{eq:saddleequationP}
  \Pi_{\alpha}(s) = \left(s+ \sqrt{\alpha}\right)^2 (s-z_+)(s-z_-)
  - \left(\eta (s+1)(s+\alpha)- \xi s(s+ \tfrac{1+\alpha}{2}) \right)^2.
  \end{equation} 
By \eqref{eq:saddleequationhigh}, the zero set of $\Pi_{\alpha}$ is the image of the zero set of $\Xi_{\alpha}$ under the natural projection $\mathcal{R}_{\alpha} \to \mathbb{C}$, $(w,z) \mapsto z$.
  
\begin{lemma} \label{Pilemmahigh}
  Let $(\xi, \eta) \in \mathcal H^{\mathrm{o}}$ (interior of
  the hexagon $\mathcal H$) and $\frac{1}{9} < \alpha <1$. 
  	\begin{enumerate}
  		\item[\rm (a)] The leading coefficient of $\Pi_{\alpha}$ is $1- (\eta-\xi)^2 > 0$.
  		\item[\rm (b)] $\Pi_{\alpha}(0) = \alpha^2 (1-\eta^2) > 0 $.
  		\item[\rm (c)] $\Pi_{\alpha}(-\alpha) = \frac{\alpha^2(1-\alpha)^2}{4} (1- \xi^2) > 0$.
  		\item[\rm (d)] $\Pi_{\alpha}(-\sqrt{\alpha}) = - \alpha(1-\sqrt{\alpha})^4
  		(\frac{\xi}{2} - \eta)^2 \leq 0$.
  		\item[\rm (e)] $\Pi_{\alpha}(-1) = \frac{(1-\alpha)^2}{4} (1-\xi^2) > 0$.
  	\end{enumerate}
  \end{lemma}
  \begin{proof}
  	These are all simple calculations based on \eqref{eq:saddleequationP}. The inequalities hold since $-1 < \xi < 1$, $-1 < \eta < 1$ and $-1 < \eta - \xi < 1$ for $(\xi, \eta) \in \mathcal H^{\mathrm{o}}$.
  \end{proof}
  
  \begin{corollary} \label{cor:Pihigh} 
  	Let $(\xi, \eta) \in \mathcal H^{\mathrm{o}}$ and $\frac{1}{9} < \alpha <1$. 
  	If $\eta = \xi/2$ then $\Pi_{\alpha}(s)$ has a double zero of at $s=-\sqrt{\alpha}$.
  	If $\eta \neq \xi/2$ then $\Pi_{\alpha}(s)$ has at least one zero in $(-1,-\sqrt{\alpha})$ and
  	at least one zero in $(-\sqrt{\alpha},-\alpha)$.
  \end{corollary}
  
  \begin{proof}
  	If $\eta \neq \xi/2$ then, by parts (c), (d), and (e) of Lemma \ref{Pilemmahigh},
  	$\Pi_{\alpha}$ has a sign change, and therefore a zero, in each of the intervals $(-1,-\sqrt{\alpha})$ and $(-\sqrt{\alpha},-\alpha)$.
  	For $\eta = \xi/2$, $\Pi_{\alpha}$ has a zero at $-\sqrt{\alpha}$ by
  	part (d), and in fact
  	\begin{equation} \label{eq:Piatmiddle} \Pi_{\alpha}(s) = (s+\sqrt{\alpha})^2
  	\left[ (s-z_+)(s-z_-) - \eta^2 (s-\sqrt{\alpha})^2 \right]
  	\quad \text{if } \eta = \xi/2, \end{equation}
  	as can be checked from \eqref{eq:saddleequationP}.
  	Hence $s=-\sqrt{\alpha}$ is a double zero if $\eta = \xi/2$. 
  \end{proof}
  
  We now give the proof of Proposition \ref{prop:saddle}
  in the high temperature regime.
  \begin{proof}[Proof of Proposition \ref{prop:saddle} for $\frac{1}{9} < \alpha \leq 1$]
  	From Corollary \ref{cor:Pihigh} it follows in particular that there 
  	are at least two zeros of $\Pi_{\alpha}$ in $(-1,-\alpha)$
  	in case $\alpha < 1$. 
  	The remaining two zeros  can also be real (frozen phase), 
  	or be a pair of complex conjugate non-real zeros (liquid phase).
  	There is at most one complex conjugate pair of non-real
  	zeros, and thus at most one zero with strictly positive imaginary
  	part. By continuity this last fact also holds for $\alpha = 1$.
  	This proves  Proposition \ref{prop:saddle} in the
  	high temperature regime. 
  \end{proof}

  \begin{proof}[Proof of Proposition \ref{prop:hightemp}]
  	The proof is similar to the proof of Proposition~\ref{prop:lowtemp}. 
	If  $s=s(\xi,\eta;\alpha)$ with $(\xi,\eta)\in \mathcal L_\alpha$ then
  	\[ \left( - \frac{s}{2(s+1)} - \frac{s}{2(s+\alpha)} \right)  
  	\xi + \eta = \pm s Q_{\alpha}(s)^{1/2}, \]
  	see \eqref{eq:Qalphahigh} and \eqref{eq:saddleequationhigh}.
  	As in the proof of Proposition \ref{prop:lowtemp}, we obtain two real equations
  	by considering the real and imaginary parts. 
  	It follows that given $s \in \mathcal R_{\alpha}^+$, where $\mathcal{R}_{\alpha}^+$ denotes the subset of $\mathcal{R}_{\alpha}$ defined in \eqref{calRplusdef}, we  recover
  	$\xi$ and $\eta$ from
  	\begin{equation}
   \label{eq:inverseHigh}\begin{pmatrix} \xi \\[10pt] \eta \end{pmatrix}
  	= \begin{pmatrix} 
  	\Re \left( - \frac{s}{2(s+1)} - \frac{s}{2(s+\alpha)} \right)  & 1 \\
  	\Im \left( - \frac{s}{2(s+1)} - \frac{s}{2(s+\alpha)} \right) & 0 
  	\end{pmatrix}^{-1} 
  	\begin{pmatrix} \Re \left(  s Q_{\alpha}(s)^{1/2} \right) \\[10pt]
  	\Im \left(  s Q_{\alpha}(s)^{1/2} \right)
  	\end{pmatrix},
  	\end{equation}
  	where the choice of square root in $Q_{\alpha}(s)^{1/2}$ is dictated
  	by the location of $s$ on the Riemann surface (different sign on
  	different sheets).
  	
  	This shows that $(\xi,\eta) \mapsto s(\xi,\eta;\alpha)$ is a bijection from $\mathcal L_{\alpha}$ to $\mathcal R_{\alpha}^+$.
  	It is clearly also differentiable (but not analytic!) 
  	and therefore it is a diffeomorphism.
  	It also extends continuously to the boundary of $\mathcal L_{\alpha}$
  	mapping for example $A_{1,2}$ to $-1$, $B_{1,2}$ to $-\alpha$, $C_{1,2}$ to $0$,  $D_{1,2}$ to $\infty$, and $E_{1,2}$ to $-\sqrt{\alpha}$, where the points with
  	subscript $1$ are mapped to the first sheet and points with subscript $2$ to the second sheet, see also Figure \ref{fig:high}.
  	
  	We finally prove that the line segment 
  	$\{ (\xi, \xi/2) \mid - \xi_{cusp} < \xi < \xi_{cusp} \}$
  	is mapped bijectively onto $\mathcal C^+ = \mathcal C \cap \mathcal R_{\alpha}^+$
  	where $(0,0)$ is mapped to the branch point $z_+$ and 
  	$\pm (\xi_{cusp},  \xi_{cusp}/2)$ is mapped to $z = -\sqrt{\alpha}$
  	with opposite $w$ values $w = \pm 2 \alpha (1+\cos \theta_{\alpha})$.
  	  	
  	For $\eta = \xi/2$, we see from \eqref{eq:Piatmiddle} that $\Pi_{\alpha}(s)$  has a double zero at $-\sqrt{\alpha}$ while the two remaining zeros satisfy
  	\[ (s-z_+)(s-z_-) - \eta^2(s-\sqrt{\alpha})^2 = 0 \]
  	which is also
  	\[ (1-\eta^2) (s^2 + \alpha)  + (-2 \cos \theta_{\alpha}+ 2 \eta^2)\sqrt{\alpha} s = 0  \]
  	since $z_+z_-=\alpha$ and $z_+ + z_- = 2 \sqrt{\alpha} \cos \theta_{\alpha}$.
  	
  	Suppose $\eta \in [0, \eta_{cusp}]$. Since $\eta_{cusp} = \cos \frac{\theta_{\alpha}}{2}$, we can write
  	$ \eta = \cos \frac{\theta}{2}$ with $\theta_{\alpha} \leq \theta \leq \pi$.
  	There is a unique $\psi \in [\theta_{\alpha}, \pi]$ with
  	\[ \sin \frac{\psi}{2}  \sin \frac{\theta}{2} = \sin \frac{\theta_{\alpha}}{2} \]
  	and  with the aid of trigonometric identities one can show that
  	$ s = \sqrt{\alpha} e^{i \psi}$  	is a zero of $\Pi_{\alpha}(s)$. 
  	If $\eta$ increases from $0$ to $\eta_{cusp}$, then $\theta$ decreases from $\pi$ to $\theta_{\alpha}$, and 
  	$\psi$ increases from $\theta_{\alpha}$ to $\pi$. It follows that
  	$s$ moves along the circle with radius $\sqrt{\alpha}$ from $z_+$
  	to $-\sqrt{\alpha}$, that is, it moves along one side of the
  	cut $\mathcal C$ on the Riemann surface.
  	By symmetry, if $\eta$ decreases from $0$ to $-\eta_{cusp}$
  	then the saddle moves along the same circle but on the other side
  	of $\mathcal C$.
  \end{proof}
\section{Equilibrium measure and $g$-function}\label{gfunctionsec}

% \subsection{Equilibrium measure: definition}

\subsection{Preliminaries}
The  orthogonality \eqref{eq:orthogonality1} does not depend on the
specific choice of contour $\gamma$. By analyticity we can deform
it to any other contour $\gamma_0$ that goes around~$0$ once
in the positive direction. For the asymptotic analysis we need to select
the `correct' contour.
The correct contour is typically (but not always...) the contour
that attracts the zeros of the orthogonal polynomials as
the degree tends to infinity.  In \eqref{eq:orthogonality1}
the orthogonality weight
\begin{equation} \label{eq:weight} 
e^{-NV(z)} = \frac{(z+1)^N(z+\alpha)^N}{z^{2N}} \end{equation}
varies with $N$, where we put 
\begin{equation} \label{eq:Valpha} 
V(z) = V_{\alpha}(z) =  2\log(z) - \log(z+1) - \log(z+ \alpha). 
\end{equation}

Such problems were studied in approximation theory where $V$ is
referred to as an external field \cite{ST}.
Since the works of Stahl \cite{Stahl} and Gonchar-Rakhmanov \cite{GR} 
it is known that the zeros tend to a contour with a certain 
symmetry property for the  logarithmic potential of its equilibrium measure. Such  contours are now called $S$-contours.  
Later, Rakhmanov \cite{Rak} made a systematic study of a max-min
characterization of $S$-contours, and with Mart\'inez-Finkelshtein 
\cite{MFR} introduced the notion of a critical measure and identified the $S$-contours as trajectories of quadratic differentials.
See \cite{KS, MFR2} for further developments and historical remarks.

For $\alpha = 1$ the external field \eqref{eq:Valpha} has only
two logarithmic singularities and in such a case the orthogonal
polynomials can be written in terms of classical Jacobi polynomials. 
Indeed, the $n$th degree polynomial $p_n$ is a multiple of the 
Jacobi polynomial
\begin{equation} \label{eq:JacobiPN}    
P_n^{(-2N,2N)}(2z+1) 
\end{equation}
in case $\alpha = 1$. The Jacobi polynomial is non-standard, since one of
the parameters is negative. The asymptotic zero distribution of Jacobi
polynomials with varying non-standard parameters was studied in
\cite{KMF, MFO, MFMGT}. The case \eqref{eq:JacobiPN} is contained
in \cite{MFO}, see also \cite{DD}, and it is known that the zeros
of \eqref{eq:JacobiPN} tend to an arc on the unit circle as $n, N \to \infty$ with  $n/N \to 1$.

\subsection{Equilibrium measure}

In order to successfully apply the RH steepest descent analysis
to the RH problem \ref{rhp:Y}, we need a contour $\gamma_0$ 
going around
$0$ and a probability measure $\mu_0$ on $\gamma_0$ with a
corresponding $g$-function
\begin{equation} \label{eq:gfunction} 
g(z) = \int \log(z-s) d\mu_0(s) 
\end{equation}
such that, for some constant $\ell \in \C$,
\begin{align} \label{eq:gcondition1}
\Re \left[ g_+(z) + g_-(z) - V(z) + \ell \right] &
\begin{cases} 	= 0, & \text{ for } z \in \supp(\mu_0), \\
\leq 0, & \text{ for } z \in \gamma_0 \setminus \supp(\mu_0),
\end{cases}	
\\
\Im \left[g_+(z) + g_-(z) - V(z) \right] & 
\begin{array}{l}
\text{ is constant on each connected } \\
\text{ component of } \supp(\mu_0), \end{array}
\label{eq:gcondition2}
\end{align}
with $V$ as in \eqref{eq:Valpha}. We call a probability measure $\mu_0$ 
satisfying \eqref{eq:gfunction}-\eqref{eq:gcondition2} 
an \textit{equilibrium measure in the external field} $V$.

For a given $\gamma$ we consider  the probability measure $\mu$ on $\gamma$ that minimizes the  energy functional 
\[ \iint \log \frac{1}{|s-t|} d\mu(s) d\mu(t) 
+ \Re \int V d\mu \]
among all probability measures on $\gamma$. By classical results
from logarithmic potential theory \cite{ST}, there is a unique
minimizer and it satisfies the conditions \eqref{eq:gcondition1}
on the real part of $g_+ + g_- - V$.
In order to be an equilibrium measure for $V$ (as
we defined it) we also need the condition \eqref{eq:gcondition2}
on the imaginary part. This condition characterizes $S$-contours.

Indeed, by the Cauchy--Riemann equations  the property \eqref{eq:gcondition2}  is equivalent to 
\[ \frac{ \partial}{\partial n_+} \left[ U^{\mu_0} + \frac{\Re V}{2} \right]
= \frac{\partial}{\partial n_-} \left[ U^{\mu_0} + \frac{\Re V}{2} \right] \]
on the support $\Sigma_0 = \supp(\mu_0)$,  where
\[ U^{\mu_0}(z) = \int \log \frac{1}{|z-s|} d\mu_0(s) \]
and $ \frac{\partial}{\partial n_{\pm}}$ denotes the normal derivatives 
on $\gamma$.
This property is known as the $S$-property of $\Sigma_0$, and $\gamma_0$
is an $S$-contour.

We remark that the equilibrium measure is not necessarily unique.
For example, if $V(z) = \log z$ then the normalized Lebesgue measure
$d\mu = \frac{ds}{2\pi i s}$ on any circle centered at the origin
is an equilibrium measure for $V$. The radius is arbitrary and
the equilibrium measure is not unique. This is a more general
phenomenon in case the support is a full closed contour.

\subsection{Construction of the equilibrium measure}
From conditions \eqref{eq:gcondition1}-\eqref{eq:gcondition2}
it follows that we are looking for $\mu_0$ such that
$g_+ + g_- - V$ is piecewise constant on the support of $\mu_0$ and 
therefore 
\[ g_+' + g_-' - V' = 0  \quad \text{ on } \Sigma_0 = \supp(\mu_0). \] 
This means that $ (g' - \frac{1}{2} V')_+ = - (g' - \frac{1}{2} V')_-$
and therefore
\begin{equation} \label{eq:Qalpha} 
Q(z) = \left[ \int \frac{d\mu_0(s)}{z-s} 
- \frac{V'(z)}{2} \right]^2 
\end{equation}
is analytic across the support of $\mu_0$.
Thus $Q$ is an analytic function in the complex plane 
with singularities determined by the  singularities of $V'$. 
We can furthermore recover  $\mu_0$ from $Q$. Indeed with an appropriate branch
of the square root,
\[ \int \frac{d\mu_0(s)}{z-s} = \frac{V'(z)}{2} + Q(z)^{1/2} \]
and then by the Sokhotski Plemelj formula
\begin{equation} \label{eq:mu0withQ}
d\mu_0(s) = \frac{1}{\pi i} Q_-(s)^{1/2} ds. 
\end{equation}
In our case of interest we have  \eqref{eq:Valpha} and
\begin{equation} \label{eq:Vprime} 
V_{\alpha}'(z) = \frac{2}{z} - \frac{1}{z+1} - \frac{1}{z+\alpha}
\end{equation}
is rational with three simple poles.  Therefore by \eqref{eq:Qalpha}
$Q = Q_{\alpha}$ is a rational function with double poles
at $z=0$, $z=-1$, and $z=-\alpha$.
We can determine $Q_{\alpha}$ explicitly, and it is given
by the formulas in Definition \ref{def:Qa}, see also section \ref{subsec:comments} below.
We will prove that the associated measure \eqref{eq:mu0withQ} is indeed an equilibrium measure 
with external field $V_{\alpha}$. 

\begin{remark}
We recall from section \ref{sec:zeta} that 
\begin{equation} \label{eq:Qasqrtlow} 
Q_{\alpha}(z)^{1/2} = \frac{(z-z_+)(z-z_-)}{z(z+1)(z+\alpha)}, \qquad \text{if } 0 < \alpha \leq \frac{1}{9}, 
\end{equation}
while for $\frac{1}{9} < \alpha \leq 1$ the square root $Q_{\alpha}(z)^{1/2}$ was considered
as a function on the first sheet of the Riemann surface $\mathcal R_{\alpha}$
shown in the right panel of Figure \ref{fig:riemannsurfaces}.
From now on it will be more convenient to change the branch cut of the Riemann surface from $\mathcal C$ to
\begin{equation} \label{eq:Sigma0high}
\Sigma_0 = \{  \sqrt{\alpha} e^{i t} \mid -\theta_{\alpha} \leq t \leq \theta_{\alpha} \}
\end{equation}
where 	$\theta_{\alpha} = \arg z_+ = - \arg z_-$.
We also modify the definition of $Q_{\alpha}(z)^{1/2}$ so that now
\begin{equation} \label{eq:Qasqrthigh} 
Q_{\alpha}(z)^{1/2} = \frac{(z+\sqrt{\alpha}) ((z-z_+)(z-z_-))^{1/2}}{z(z+1)(z+\alpha)}, 
\quad \text{if } \frac{1}{9} < \alpha \leq 1, \end{equation}
is defined and analytic for $z \in \mathbb C \setminus \Sigma_0$ with
the square root such that $Q_{\alpha}(z)^{1/2} \sim \frac{1}{z}$ as $z \to \infty$.
The circular arc \eqref{eq:Sigma0high} will be the support of the equilibrium measure $\mu_0$.
\end{remark}

We let $\gamma_0$ denote the circle of radius $\sqrt{\alpha}$ centered at $0$ oriented in the counterclockwise direction.

With \eqref{eq:Qasqrtlow} and \eqref{eq:Qasqrthigh}, we define the measure $\mu_0$, 
the associated $g$-function, and the variational constant $\ell$ as follows.

\begin{definition}\label{mu0gdef} \

(a) If $\frac{1}{9} \leq \alpha \leq 1$, then we define the measure $\mu_0$ by
\begin{align} \nonumber	
d\mu_0(s) & = \frac{1}{\pi i} Q_{\alpha,-}(s)^{1/2} ds \\
& = \frac{1}{\pi i} \frac{(s+\sqrt{\alpha})
	\, ((s-z_+)(s-z_-))^{1/2}_-}{s(s+1)(s+\alpha)} \, ds, 
\qquad s \in  \Sigma_0,  \label{eq:mu0high}
\end{align}
where $\Sigma_0$ is given by \eqref{eq:Sigma0high} with counterclockwise orientation,
and $Q_{\alpha,-}(s)^{1/2}$ denotes the limit of $Q_{\alpha}(z)^{1/2}$ as $z \to s \in \Sigma_0$
with $z$ in the exterior of the circle $\gamma_0$.
Recall  $z_{\pm}= z_{\pm}(\alpha)$ are given by \eqref{eq:zpmhigh}.	

The associated $g$-function is defined by
\begin{equation} \label{eq:gfunctionhigh} 
g(z) = \int_{\Sigma_0} \log(z-s) d\mu_0(s), \qquad z \in \C \setminus 
\left((-\infty, -\sqrt{\alpha}] \cup \{\sqrt{\alpha} e^{i t} \mid -\pi \leq t \leq \theta_{\alpha} \} \right),
\end{equation}
where for each $s \in \Sigma_0$, the branch of the logarithm
$z \mapsto \log(z-s)$ is taken that is analytic in $\C \setminus ((-\infty, -\sqrt{\alpha}] \cup 
\{ \sqrt{\alpha} e^{i t} \mid -\pi \leq t \leq \arg s\}$ and behaves like 
$\log(z-s) \sim \log |z| + i \arg(z) $, $-\pi < \arg z < \pi$ as  $z \to \infty$. 

\medskip	
(b) If $0 < \alpha \leq \frac{1}{9}$, then we define the measure $\mu_0$ by
\begin{align} \nonumber	
d\mu_0(s) & = \frac{1}{\pi i} Q_{\alpha}(s)^{1/2} ds \\
& = \frac{1}{\pi i} \frac{(s-z_+)(s-z_-)}{s(s+1)(s+\alpha)} ds, \qquad s \in  \Sigma_0,  \label{eq:mu0low}
\end{align}
where $\Sigma_0 = \gamma_0 = \supp(\mu_0)$ is the full circle of radius $\sqrt{\alpha}$ oriented 
in the counterclockwise direction and $z_{\pm}= z_{\pm}(\alpha)$ are given by \eqref{eq:zpmlow}.

The associated $g$-function is defined by
\begin{equation} \label{eq:gfunctionlow} 
g(z) = \int_{\Sigma_0} \log(z-s) d\mu_0(s), \qquad z \in \C \setminus \left((-\infty, -\sqrt{\alpha}] \cup \Sigma_0\right)
\end{equation}
where $z \mapsto \log(z-s)$ is defined in the same way as
in the high temperature regime.  

\medskip
(c) We define the
variational constant $\ell \in \C$  by
\begin{align}\label{variationalelldef}
\ell = \begin{cases}
-2g_{-}(\sqrt{\alpha}) + V_{\alpha}(\sqrt{\alpha}) - \pi i, & \mbox{if } 0 < \alpha \leq \frac{1}{9} \\ 
- 2g(z_{+}) + V_{\alpha}(z_{+}), &  \mbox{if } \frac{1}{9} < \alpha \leq 1
\end{cases}.
\end{align}
\end{definition}

The definition \eqref{variationalelldef} is such that equality
holds in \eqref{eq:gcondition1} at $z = z_{+} \in \Sigma_0$ for $\frac{1}{9}<\alpha \leq 1$ and at $z = \sqrt{\alpha}\in \Sigma_{0}$ for $0<\alpha \leq \frac{1}{9}$.

For the steepest descent analysis of the RH problem, it is convenient to introduce a function $\phi(z)$ which  is a primitive function of $Q_\alpha(z)^{1/2}$ (with appropriate choices of the branch).

\begin{definition}\label{phidef} \

(a) If $\frac{1}{9} < \alpha \leq 1$, then $\phi:\C  \setminus ((-\infty,0] \cup \{\sqrt{\alpha} e^{it} \mid 
-\pi\leq t \leq \theta_{\alpha} \})  \to \C$ is defined by
\begin{equation} \label{phidefhigh} 
\phi(z) = \int_{z_{+}}^z Q_{\alpha}(s)^{1/2} ds, 
\end{equation}
with $Q_{\alpha}^{1/2}$ given by \eqref{eq:Qasqrthigh}, and  the integration path from $z_{+}$ to $z$ does not intersect
$(-\infty,0] \cup \{\sqrt{\alpha} e^{it} \mid 
-\pi\leq t \leq \theta_{\alpha} \}$.
\medskip

(b) If $0 < \alpha < \frac 19$, then $\phi:\C  \setminus ((-\infty,0] \cup \Sigma_0)  \to \C$ is defined by
\begin{align}\label{phideflow}
\phi(z) = \begin{cases}
\ds - \frac{\pi i}{2} + \int_{\sqrt{\alpha}}^{z} Q_{\alpha}(s)^{1/2} ds, &  \text{for } 
|z| > \sqrt{\alpha}, 
\\[10pt] \ds \frac{\pi i}{2}
 -\int_{\sqrt{\alpha}}^{z} Q_{\alpha}(s)^{1/2} ds, & \text{for }
|z| < \sqrt{\alpha},
\end{cases}
\end{align}
with $Q_{\alpha}^{1/2}$ given by \eqref{eq:Qasqrtlow}, and
the integration path from $\sqrt{\alpha}$ to $z$ does not intersect 
$(-\infty,0] \cup \Sigma_0$.
\end{definition}

The formulas \eqref{eq:mu0high} and \eqref{eq:mu0low}
define $\mu_0$ as a complex measure on $\Sigma_0$. The fact that it is
a probability measure is part of the statement of the following
proposition whose proof is given in Section \ref{measurepropproofsec}. 

\begin{proposition} \label{prop:mu0} 
Let $0 < \alpha \leq 1$ and let $\gamma_0$
be the circle of radius $\sqrt{\alpha}$ centered at $0$ oriented positively.
Then the measure $\mu_0$ defined in \eqref{eq:mu0high} and \eqref{eq:mu0low} is a probability measure
on $\Sigma_0$ and is an equilibrium measure in the external field $V_{\alpha}$. 
The functions $g$ and $\phi$ are analytic in their domains of definitions and are related by
\begin{align}\label{phigrelation}
\phi(z) = g(z) - \frac{V_\alpha(z)}{2} + \frac{\ell}{2} 
\end{align}
for all $z$ in the domain of $\phi$. Moreover,
\begin{align}\label{gjumponSigma0}
g_+(z) + g_-(z) - V_\alpha(z) & = - \ell, \quad \text{ for } z \in \Sigma_0,
\\ \label{gjumponSigma0b}
g_+(z) - g_-(z) - 2\phi_+(z) & = 0, \qquad \text{for } z \in \Sigma_0.
\end{align}
\end{proposition}

\subsection{The zero set of $\Re \phi$}
To prepare for the proof of Proposition \ref{prop:mu0} we first
present a lemma about the quadratic differential $Q_{\alpha}(z) dz^2$.

A smoothly parametrized curve $z=z(t)$, $t \in [a,b]$, is a {\it trajectory} of 
a quadratic differential $Q(z) dz^2$ if $Q(z(t)) z'(t)^2 < 0$ for every $t \in (a,b)$. It is 
an {\it orthogonal trajectory} if $Q(z(t)) z'(t)^2 > 0$ 
for every $t \in (a,b)$.
A trajectory or an orthogonal trajectory is {\it critical} if it contains
a zero or a pole of $Q$.

\begin{lemma} \label{lem:trajectory}
\begin{enumerate} 
	\item[(a)] For every $\alpha \in (0,1]$, the curve
	$\Sigma_0$ is a trajectory of the quadratic differential
	$Q_{\alpha}(z) dz^2$. If $\alpha \geq \frac{1}{9}$, then it is a critical 
	trajectory passing through the zeros $z_\pm(\alpha)$ of $Q_{\alpha}$.
	
	\item[(b)] For every $\alpha \in (\frac{1}{9},1]$, the complementary arcs
	on the circle $|z| = \sqrt{\alpha}$, with parametrizations
	$z(t) = \sqrt{\alpha} e^{it}$, $t \in (\theta_{\alpha}, \pi)$
	or $t \in (-\pi, -\theta_{\alpha})$ are critical orthogonal trajectories 
	that connect $z_{\pm}(\alpha)$ with the double zero at $-\sqrt{\alpha}$. \end{enumerate}
\end{lemma}
\begin{proof}
Let $z= z(t) = \sqrt{\alpha} e^{i t}$, so that $z' = i z$. 
For $\alpha \geq \frac{1}{9}$, we write $z_{\pm} = \sqrt{\alpha} e^{\pm i \theta_{\alpha}}$ with $0 < \theta_{\alpha} \leq \pi$, and then
by \eqref{eq:Qalphahigh}
\begin{align} \nonumber
Q_\alpha(z) (z')^2 & = -
\frac{(z+ \sqrt{\alpha})^2 (z-z_+)(z-z_-)}{(z+1)^2(z+\alpha)^2} \\
& = - \alpha^2 \frac{(e^{it} + 1)^2  \nonumber
	(e^{it}-e^{i\theta_{\alpha}})(e^{it}-e^{-i\theta_{\alpha}})}
{(\sqrt{\alpha} e^{it} + 1)^2 (\sqrt{\alpha} e^{it} + \alpha)^2} \\
& = - 16 \alpha \frac{ \left(\cos \frac{t}{2} \right)^2
	\sin \left( \frac{\theta_{\alpha}-t}{2} \right)
	\sin \left( \frac{\theta_{\alpha}+t}{2} \right)}
{(1+ \alpha + 2\sqrt{\alpha} \cos t)^2}. \label{eq:zQzhigh}
\end{align}
This expression is indeed $< 0$ for $-\theta_{\alpha} < t < \theta_{\alpha}$
and $>0$ for $\theta_{\alpha} < t < \pi$ and $-\pi < t < -\theta_{\alpha}$.

For $0 < \alpha < \frac{1}{9}$, a similar computation
using \eqref{eq:Qalphalow} and \eqref{eq:zpmlow} gives
\begin{align} 
Q_\alpha(z) (z')^2 & = -
\frac{(z- z_+)^2 (z-z_-)^2}{(z+1)^2(z+\alpha)^2} = - \frac{(z^2 + \frac{1+3\alpha}{2} z + \alpha)^2}
{(z+1)^2(z+\alpha)^2} = - \frac{ \left( \frac{1+3\alpha}{2} + 2 \sqrt{\alpha} \cos t\right)^2}
{(1+\alpha + 2 \sqrt{\alpha} \cos t)^2}. \label{eq:zQzlow}
\end{align}
Since $0<\alpha < \frac{1}{9}$ we have 
$\frac{1+3 \alpha}{2} > 2 \sqrt{\alpha}$ and therefore the numerator
is always $>0$. Thus $Q_\alpha(z) (z')^2 < 0$ for every $t \in [-\pi, \pi]$.
\end{proof}

For $\alpha > \frac{1}{9}$ we recall that $z_{\pm}$ are simple
zeros of $Q_{\alpha}$. From the local structure of trajectories
of a quadratic differential there are three critical trajectories
emanating from each of the points $z_{\pm}$. One of these is an arc on the circle
$|z|= \sqrt{\alpha}$, as we have seen. The other critical trajectories
also connect $z_+$ with $z_-$ and a representative situation
is shown in Figure \ref{fig: crit traj alpha 03}.

\begin{figure}[t]
\begin{center}
	\begin{tikzpicture}
	\node at (0,0) {\includegraphics[width = 10cm]{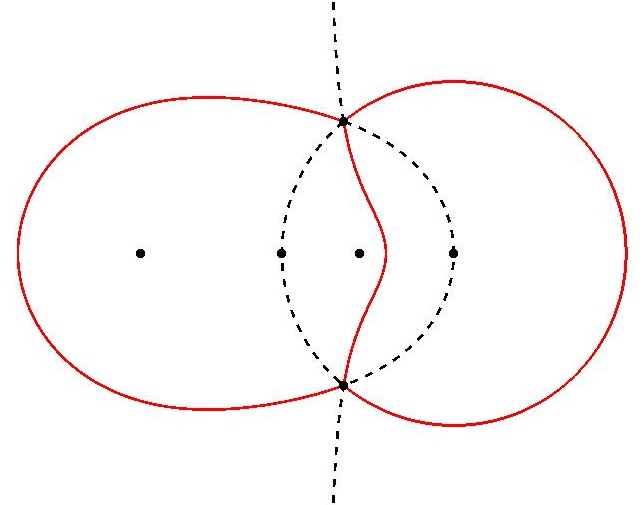}};
	\node at (-1.7,0) {$-$};
	\node at (2.7,0) {$+$};
	\node at (0,3) {$+$};
	\end{tikzpicture} 
\end{center}
\caption{\label{fig: crit traj alpha 03}The critical trajectories (in full red lines) and the critical orthogonal trajectories (the dashed black lines) of $Q_{\alpha}$ for $\alpha = 0.3$. The dots are the zeros and poles of $Q_{\alpha}$: $z_{+}$, $z_{-}$, $-\sqrt{\alpha}$, and $-1$, $-\alpha$, $0$. The critical trajectories are level lines $\Re \phi = 0$
	and their complement consists of three regions where the sign of $\Re \phi $ is constant, as shown by $+$ and $-$ in the figure.}
\end{figure}
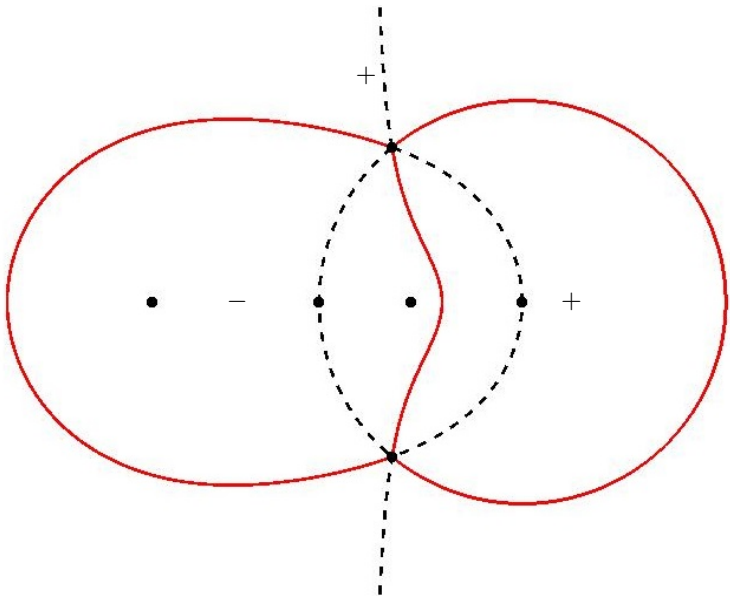

The trajectories of the quadratic differential $Q_{\alpha}(z)dz^2$
are level lines of $\Re \phi$, since $\phi$ is a primitive 
function of $\pm Q_{\alpha}^{1/2}$ as follows from Definition \ref{phidef}.
The orthogonal trajectories are level lines of $\Im \phi$.

Since $\sqrt{\alpha} \in \Sigma_0$ we in fact have that
$\Re \phi = 0$ on $\Sigma_0$ as  well as on the other critical
trajectories (in the high temperature regime) that are shown
in Figure \ref{fig: crit traj alpha 03} for $\alpha = 0.3$. 
The three critical trajectories are boundaries of three regions
in the complex plane on which $\Re \phi$ has a constant sign.
Namely $\Re \phi < 0$ in the region containing $-1$, and
$\Re \phi > 0$ in the region containing $0$ and in the unbounded region.

To prove this we introduce
\begin{equation} \label{eq:Nphi}
\mathcal N_\phi = \{ z  \mid  \Re \phi(z)=0\}.
\end{equation}
Then $\Sigma_0$ is contained in $\mathcal N_\phi$, but 
$\mathcal N_\phi$ also contains other parts, see Figures
\ref{fig:sign of xi high} and \ref{fig:sign of xi low} for representative
figures in the high and low temperature regimes.

The first thing to observe is that $\Re \phi$ extends to a continuous function on $\mathbb C$ away from  $-1$, $-\alpha$, and $0$. 
Indeed, $Q_{\alpha}^{1/2}$ has simple poles at these three
values, and   therefore by integration as in definitions \eqref{phidefhigh} and \eqref{phideflow}, we find that
$\phi$ has logarithmic behavior.
However, since the residues of $Q_{\alpha}^{1/2}$ are real, the real
part of $\phi$ is single-valued.
Thus $\Re \phi$ is continuous on $\mathbb C \setminus \{-1,-\alpha,0\}$
and harmonic on $\mathbb C \setminus (\Sigma_0 \cup \{-1,-\alpha,0\})$.
We also note 
\begin{equation}
\label{eq:Nphinearpoles} 
\begin{aligned} 
\phi(z) & = - \log z + \bigO(1)  \text{ as } z \to 0, & 
\lim_{z \to 0} \Re \phi(z) = + \infty \\
\phi(z) & = \frac{1}{2} \log(z+\alpha) + \bigO(1) \text{ as } z \to -\alpha, &
\lim_{z \to -\alpha} \Re \phi(z) = - \infty \\
\phi(z) & = \frac{1}{2} \log(z+1) + \bigO(1) \text{ as } z \to -1, &
\lim_{z \to -1} \Re \phi(z) = - \infty \\
\phi(z) & = \log(z) + \bigO(1) \text{ as } z \to \infty, &
\lim_{z \to \infty} \Re \phi(z) = + \infty. 
\end{aligned}
\end{equation}

In the high temperature regime the level set $\mathcal N_{\phi}$
consists of the critical trajectories of the quadratic
differential $Q_{\alpha}(z) dz^2$ emanating from $z_+(\alpha)$.

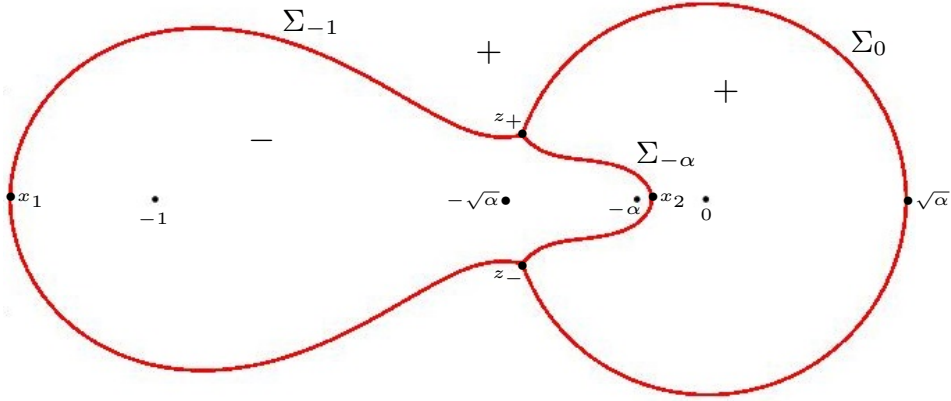
\begin{figure}[t]
	\begin{center}
		\begin{tikzpicture}[master,scale = 1.3,every node/.style={scale=1.3}]
		\node at (0,0) {\includegraphics[width=9.23cm]{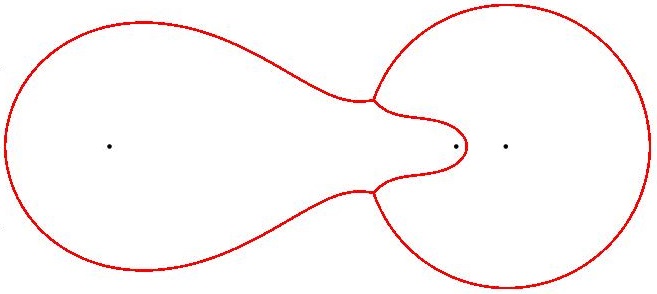}};
		\node at (0.3,1.5) {$+$};
		\node at (2.7,1.1) {$+$};
		\node at (-2,0.6) {$-$};
		\node at (4.15,1.6) {\small $\Sigma_0$};
		\node at (2.1,0.45) {\small $\Sigma_{-\alpha}$};
		\node at (-1.5,1.8) {\small $\Sigma_{-1}$};
		
		\node at (-4.4,-0)  {\tiny $\bullet x_1$};
		\node at (-3.1,-0.2) {\tiny $-1$};
		\node at (0.2,-0)  {\tiny $-\sqrt{\alpha} \bullet$};
		\node at (1.67,-0.1) {\tiny $-\alpha$};
		\node at (2.1,-0)  {\tiny $\bullet x_2$};
		\node at (2.5,-0.15) {\tiny $0$};
		\node at (4.72,-0)  {\tiny $\bullet \sqrt{\alpha}$};
		
		\node at (0.5,0.8) {\tiny $z_+$};
		\node at (0.5,-0.8) {\tiny $z_-$};
		\node at (0.64,0.67) {\tiny $\bullet$};
		\node at (0.64,-0.67) {\tiny $\bullet$};
		\end{tikzpicture}
	\end{center}
	\caption{\label{fig:sign of xi high} The set $\mathcal{N}_{\phi}=\{z \in \mathbb{C} : \Re \phi(z) = 0\} = \Sigma_{-1} \cup \Sigma_{-\alpha} \cup \Sigma_0$ is shown for $\alpha = \frac{1}{8}$. This set divides $\mathbb{C}$ into three regions, and the sign of $\Re \phi$ is shown in each of these regions.}
\end{figure}

\begin{lemma}\label{lem:Nphihigh}
	Let $\frac{1}{9} < \alpha \leq 1$. The set $\mathcal N_\phi$ consists of
	three analytic arcs connecting $z_+$ and $z_-$ which we denote by $\Sigma_{-1}$, $\Sigma_{-\alpha}$ and $\Sigma_0$. The arc $\Sigma_{-1}$
	intersects the real axis at $x_1 \in  (-\infty,-1)$ and $\Sigma_{-\alpha}$
	intersects the real axis at $x_2 \in  (-\alpha,0)$. The arc $\Sigma_0$
	is the support of the measure $\mu_0$ and is part of the circle $|z| = \sqrt{\alpha}$.
\end{lemma}
\begin{proof} Because of the local behavior of trajectories
	of a quadratic differential at a simple zero, there 
	are three trajectories emanating from $z_+$. One of these trajectories 
	is $\Sigma_0$. The other two trajectories have to remain bounded
	and stay away from the poles $-1$, $-\alpha$, $0$ by \eqref{eq:Nphinearpoles}.
	They have to come to the real axis. Indeed, if not, they would have to form a close loop in the upper haf plane and, since $\Re \phi$ is harmonic inside this closed loop, we obtain a contradiction with the maximum/minimum principle for harmonic functions. 	Therefore, the trajectories come to the real axis and, by symmetry,
	they continue to the other simple zero $z_- = \overline{z_+}$.
	The three trajectories enclose two bounded domains and $\Re \phi = 0$  
	on the boundary of these domains. Again, note that $\Re \phi$ is harmonic in the 
	interior, except at $-1$, $-\alpha$, $0$,
	where it tends to $\pm \infty$, see \eqref{eq:Nphinearpoles}.
	By the maximum/minimum principle of harmonic functions each
	of the domains should contain at least one of the singularities.
	
	Again by \eqref{eq:Nphinearpoles} there are
	points $x_1 \in (-\infty,-1)$ and $x_2 \in (-\alpha,0)$
	with $\Re \phi(x_1) = \Re \phi(x_2) = 0$. Also $\Re \phi(\sqrt{\alpha}) = 0$
	and we claim that $x_1, x_2, \sqrt{\alpha}$ are the only points
	in $\mathcal N_{\phi} \cap \mathbb R$. 
	
	To see this we recall that $\phi' = Q_{\alpha}^{1/2}$,
	with a branch cut along $\Sigma_0$ for the square root. 
	From the formula \eqref{eq:Qasqrthigh} we then see that
	$\phi'$ changes sign in the five values $-1$, 
	$-\sqrt{\alpha}$, $-\alpha$, $0$, and $\sqrt{\alpha} \in \Sigma_0$.
	Thus $\phi' > 0$ (and $\Re \phi$ is strictly increasing) on
	the intervals  $(-1,-\sqrt{\alpha})$, $(-\alpha, 0)$, and $(\sqrt{\alpha}, \infty)$, while
	$\phi' < 0$ (and $\Re \phi$ is stictly decreasing) on 
	$(-\infty, -1)$, $(-\sqrt{\alpha},-\alpha)$, and $(0, \sqrt{\alpha})$.
	Since $\Re \phi(\sqrt{\alpha}) = 0$, we conclude that there
	are no other zeros of $\Re \phi$ in $[0,\infty)$. Also
	$x_1$ is the only  zero in $(-\infty,-1]$ and $x_2$ is the
	only zero of $\Re \phi$ in $[-\alpha,0]$. 
	On the remaining interval $(-1,-\alpha)$, we see that
	$\Re \phi$ assumes its maximum value at $-\sqrt{\alpha}$.
	At $-\sqrt{\alpha}$ we have by \eqref{phigrelation}
	\[ \Re \phi = \Re \left(g - \frac{V_{\alpha}}{2} + \frac{\ell}{2}\right)	
	< 0 \]
	where the inequality holds because of the variational inequality \eqref{eq:gcondition1}
	at $-\sqrt{\alpha} \in \gamma_0 \setminus \Sigma_0$, which
	in the high temperature regime is a strict inequality, see
	also \eqref{eq:gminVstrict}.
	Therefore $\Re \phi$ has no zeros in $(-1,-\alpha)$,
	and we proved the claim that 
	\[ \mathcal N_{\phi} \cap \mathbb R = \{x_1, x_2, \sqrt{\alpha} \}. \]
	
	We conclude that one critical trajectory comes to $x_1$
	and another one to $x_2$. This defines the contours $\Sigma_{-1}$ and 
	$\Sigma_{-\alpha}$. 
	
	It remains to prove there are no other parts in $\mathcal N_\phi$.
	Any potential other part of $\mathcal N_{\phi}$ cannot intersect
	the real axis, as we already saw. Then such a part would be
	a closed contour in the upper or lower half plane and we arrive, again,
	at a contradiction because of the maximum/minimum principle for 
	harmonic functions.
\end{proof}

The structure of 
$\mathcal N_{\phi}$ is different in the low temperature regime,
see Figure~\ref{fig:sign of xi low}.

\begin{lemma} \label{lem:Nphi}
	Let $0 < \alpha < \frac 19$.
	The set $\mathcal N_\phi$ is the disjoint union of three analytic closed curves which we denote by $\Sigma_{-1}$, $\Sigma_{-\alpha}$ and  $\Sigma_0$. The closed curve $\Sigma_0$ is the circle of radius $\sqrt{\alpha}$ around $0$, as before,
	and $\Sigma_{-1}$, $\Sigma_{-\alpha}$ are two closed curves lying in
	the exterior/interior of $\Sigma_0$ and going
	around $-1$ and $-\alpha$, respectively.
\end{lemma}

\begin{figure}[t]
	\begin{center}
		\begin{tikzpicture}[master,every node/.style={scale=1.3}]
		\node at (0,0) {\includegraphics[width=9.23cm]{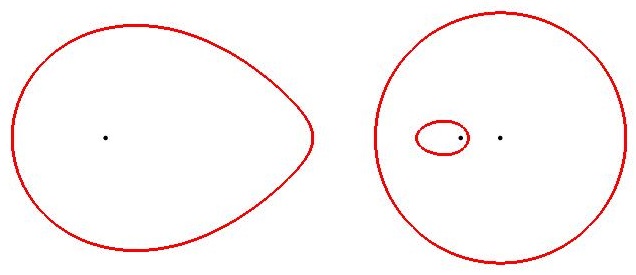}};
		\node at (3.5,1.6) {$+$};
		\node at (0.4,1.5) {$+$};
		\node at (2.25,0.0) {$-$};
		\node at (-2,0.6) {$-$};
		\node at (5.3,2) {\small $\Sigma_0$};
		\node at (2.9,0.52) {\small $\Sigma_{-\alpha}$};
		\node at (-1.6,2) {\small $\Sigma_{-1}$};
		
		\node at (-5.58,-0)  {\tiny $\bullet x_1$};
		\node at (-4.1,-0.2) {\tiny $-1$};
		\node at (-0.3,-0)  {\tiny $x_2 \bullet$};
		\node at (0.7,-0)  {\tiny $-\sqrt{\alpha} \bullet$};
		\node at (1.7,-0) {\tiny $x_3 \bullet$};
		\node at (2.49,-0.13) {\tiny $-\alpha$};
		\node at (3.07,-0)   {\tiny $\bullet x_4$};
		\node at (3.45,-0.2) {\tiny $0$};
		\node at (6.08,-0)  {\tiny $\bullet \sqrt{\alpha}$};
		\end{tikzpicture}
	\end{center}
	\caption{\label{fig:sign of xi low} The set $\mathcal{N}_{\phi}=\{z \in \mathbb{C} \mid \Re \phi(z) = 0\} = \Sigma_{-1} \cup \Sigma_{-\alpha} \cup \Sigma_0$ is shown for $\alpha = \frac{1}{10}$. This set divides $\mathbb{C}$ into four regions, and the sign of $\Re \phi$ is shown in each of these regions.}
\end{figure}
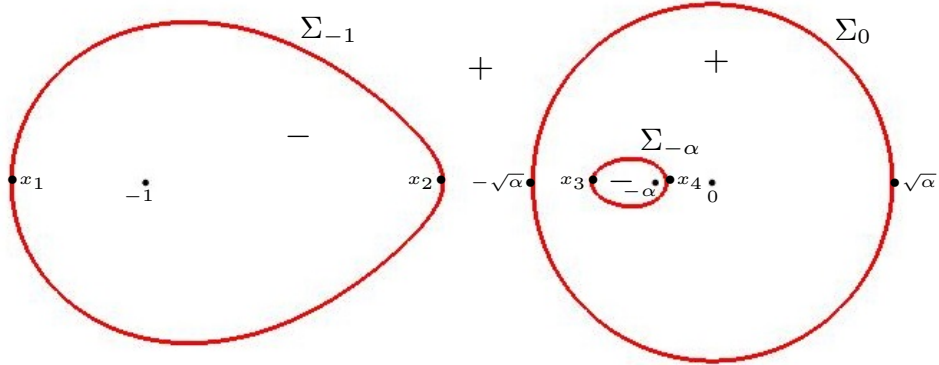
\begin{proof}
	Because of \eqref{eq:Nphinearpoles} the level set $\mathcal N_\phi$ 
	is bounded and stays away from the poles  $-1$, $-\alpha$,
	and $0$ of $Q_{\alpha}$. Since we already know from Lemma \ref{lem:trajectory} that $\Re \phi(-\sqrt{\alpha}) = 0$, we infer from \eqref{phideflow} that the zeros $z_{\pm}$ of $Q_{\alpha}$ are not on 
	$\mathcal N_{\phi}$. Therefore $\mathcal N_{\phi}$ does not contain
	any critical trajectories and hence consists
	of a finite union of disjoint closed curves. 
	Because of the maximum/minimum principle for harmonic functions
	each component of $\overline{\mathbb{C}} \setminus \mathcal N_\phi$ 
	contains 
	at least one of the singularities $-1$,$-\alpha$, $0$, or $\infty$. 
	
	A closer inspection of  $\Re \phi(z)$  for $ z  \in \mathbb R$ 
	(also based on \eqref{eq:Qasqrtlow}, \eqref{phideflow} and  \eqref{eq:Nphinearpoles} reveals  that 
	$\mathcal N_\phi$ has six intersection points with $\mathbb R$. 
	Two of them are the points $\pm \sqrt \alpha$ that belong to $\Sigma_0$. 
	Then we have one point in each of the intervals $(-\infty,-1)$, $(-1,-\sqrt \alpha)$, $(-\sqrt \alpha, -\alpha)$ and $(-\alpha,0)$. 
	This shows that there is a closed curve $\Sigma_{-\alpha}$ inside $\Sigma_0$ and a closed curve $\Sigma_{-1}$ outside $\Sigma_0$ as indicated in the statement. 
\end{proof}

\subsection{Proof of Proposition \ref{prop:mu0}}\label{measurepropproofsec}

We compute $\int\limits_{\Sigma_0} d\mu_0$ by means of a
residue calculation. Let us first consider the case $\frac{1}{9} \leq \alpha \leq 1$.
Then by \eqref{eq:mu0high} and the fact
that $Q_{\alpha,+}(s)^{1/2} = - Q_{\alpha,-}(s)^{1/2}$ for $s \in \Sigma_0$, we have
\begin{equation}  \label{eq:mu0contour} 
\int_{\Sigma_0} d\mu_0 =
\frac{1}{2\pi i} \oint_{C} Q_{\alpha}(s)^{1/2} ds \end{equation}
where $C$ is a closed contour going around $\Sigma_0$ once
in the positive direction, and without enclosing any of the poles.
Deforming the contour $C$ to infinity, we pick up residue
contributions at the poles. 
It is a straightforward calculation to show that
\begin{equation} \label{eq:Qalphares}
\Res_{s=0} Q_{\alpha}(s)^{1/2} = -1, \quad
\Res_{s=-1} Q_{\alpha}(s)^{1/2} = \frac{1}{2}, \quad
\Res_{s=-\alpha} Q_{\alpha}(s)^{1/2} = \frac{1}{2}.
\end{equation}
The residues add up to zero, and since 
$Q_{\alpha}(s)^{1/2} = \frac{1}{s} + \bigO(s^{-2})$ as $s \to \infty$, we 
thus find from \eqref{eq:mu0contour} 
\begin{align} \label{eq:mu0total}
\int_{\Sigma_0} d\mu_0  = 1.
\end{align}

Let $z(t) = \sqrt{\alpha} e^{i t}$, 
$-\theta_{\alpha} < t < \theta_{\alpha}$, be a parametrization of
$\Sigma_0$. Then the mapping
\begin{equation} \label{eq:mu0mass} 
t \mapsto \int_{z_-}^{z(t)} d\mu_0 = \frac{1}{\pi i} 
\int_{z_-}^{z(t)} Q_{\alpha,-}(s)^{1/2} ds \end{equation}
has as its derivative
\[ \frac{1}{\pi i} Q_{\alpha,-}(z(t))^{1/2} \cdot z'(t) \]
which is real and non-zero for $t \in (-\theta_{\alpha}, \theta_{\alpha})$ 
since $Q_{\alpha}(z) (z')^2 < 0$ as $\Sigma_0$ is a trajectory of the quadratic differential by Lemma \ref{lem:trajectory} (a).

Note also that the right-hand side of \eqref{eq:mu0mass} vanishes for $t = - \theta_{\alpha}$
and equals $1$ for $t= \theta_{\alpha}$ by \eqref{eq:mu0total}.
Therefore \eqref{eq:mu0mass} is monotonically increasing from $0$ to $1$
as $t$ goes from $-\theta_{\alpha}$ to $\theta_{\alpha}$,
and this is enough to conclude that  $\mu_0$ is a probability  
measure on $\Sigma_0$. 

It now also follows (compare \eqref{phidefhigh} and
\eqref{eq:mu0mass}, and use $Q_{\alpha,+}^{1/2} = - Q_{\alpha,-}^{1/2}$ on
$\Sigma_{0}$)
that  $\phi_-$ is purely imaginary along $\Sigma_0$ and we have
\begin{equation} \label{eq:phipm} 
\phi_+(z) = - \phi_-(z), \qquad \text{for } z \in \Sigma_0. 
\end{equation}
Next we calculate $g'(z) = \int\limits_{\Sigma_0} \frac{d\mu_0(s)}{z-s}$.
We write $g'$ as a contour integral
\[ g'(z) = \frac{1}{2\pi i} \oint_C \frac{Q_{\alpha}(s)^{1/2}}{z-s} ds,
\qquad z \in \mathbb C \setminus \Sigma_0,
\]
with the same closed contour $C$ as in \eqref{eq:mu0contour}, but we now also
assume that $z$ is in the exterior of $C$.
We deform the contour to infinity where we now pick up a residue contribution
from the pole at $s=z$ as well, which is $Q_{\alpha}(z)^{1/2}$.
We use \eqref{eq:Qalphares} to calculate the other residue contributions.
There is no contribution from infinity and the result is that
\begin{align} \nonumber
g'(z) & = \frac{1}{z} -  \frac{1}{2(z+1)} - \frac{1}{2(z+\alpha)}
+ Q_{\alpha}(z)^{1/2} \\
& =   \frac{V_{\alpha}'(z)}{2} + \phi'(z),
\qquad z \in \mathbb C \setminus \Sigma_0.
\label{eq:gprimehigh}
\end{align}
Integrating \eqref{eq:gprimehigh} from $z_{+}$ to $z$ along a
path that does not intersect $(-\infty, 0] \cup
\{\sqrt{\alpha} e^{it} \mid -\pi \leq t \leq \theta_{\alpha}\})$, 
we find
\begin{align}\label{gzminusgzplus}
g(z) - g(z_{+}) = \frac{V_\alpha(z) - V_\alpha(z_{+})}{2} + \phi(z) - \phi(z_{+}),
\end{align} 
which proves \eqref{phigrelation} for $\alpha \in [\frac 19,1]$
by the definition \eqref{variationalelldef} of $\ell$ and the fact
that $\phi(z_{+})=0$.

From \eqref{phigrelation} and \eqref{eq:phipm} 
we obtain for $z \in \Sigma_0$,
\begin{align*}
g_+(z) + g_-(z) - V_{\alpha}(z) = \phi_+(z) + \phi_-(z) - \ell
= -\ell,
\end{align*} 
which proves \eqref{gjumponSigma0}. Also by \eqref{phigrelation}
and \eqref{eq:phipm} 
\begin{align*}
g_+(z) - g_-(z) = \phi_+(z) - \phi_-(z) = 2 \phi_+(z)
\end{align*}
which is \eqref{gjumponSigma0b}.

We have also shown that $\phi_-(z) \in i \mathbb R$ 
for $z \in \Sigma_0$,  
and similarly $\phi(z) \in i \mathbb R$
on the other critical trajectories that emanate from $z_+$ and $z_-$, 
see Figure \ref{fig: crit traj alpha 03}. Moreover,
$\Im \phi $ is constant on orthogonal trajectories.
We also saw that $\Im \phi_-(z)$ increases as $z$
moves away from $z_-$ to $z_+$ along $\Sigma_0$. Then by the Cauchy-Riemann equations,
we have $\Re \phi > 0$ in the domain on the minus side of $\Sigma_0$
and by continuity it holds in the outer domain
bounded by the critical trajectories. Then $\Re \phi < 0$
if we cross the critical trajectory going around $-1$,
and in particular $\Re \phi(z) < 0$ for $z$ on the critical orthogonal
trajectory from $z_+$ to $-\sqrt{\alpha}$. 
In view of \eqref{phigrelation}, this gives 
\begin{equation} \label{eq:gminVstrict} 
\Re \left[ 2g(z) - V_{\alpha}(z) + \ell\right] < 0,  
\end{equation}
for $z$ on this orthogonal trajectory, which is part of $\gamma_0 \setminus \Sigma_0$.
This proves the inequality in \eqref{eq:gcondition1}.
By symmetry the inequality also holds for $z$ on the
critical orthogonal trajectory from $z_-$ to $-\sqrt{\alpha}$.
This completes the proof for the case $\alpha \geq \frac{1}{9}$.

\medskip
The proof for $0 < \alpha < \frac{1}{9}$ is simpler. In this case 
\eqref{eq:Qasqrtlow}  is a rational function   with partial fraction decomposition
\begin{equation} \label{eq:Qalphasqrtlow}
Q_{\alpha}(s)^{1/2} = \frac{1}{s} + \frac{1}{2(s+1)} - \frac{1}{2(s+\alpha)}.
\end{equation}
The total integral of $\mu_0$ defined by \eqref{eq:mu0low} is
\[ \int_{\Sigma_0} d\mu_0 = \frac{1}{\pi i} \oint_{\gamma_0}
\left(\frac{1}{s} + \frac{1}{2(s+1)} - \frac{1}{2(s+\alpha)} \right) 
ds = 1 \]
by a simple residue calculation with contributions only from the poles at $s=0$
and $s=-\alpha$. The total mass is $1$ and as before
it follows that $\mu_0$ is a probability measure.

We compute $g'(z)$ with another residue calculation
\begin{align*} g'(z) & = \frac{1}{\pi i}
\oint_{\gamma_0} \frac{1}{z-s} \left(\frac{1}{s} + \frac{1}{2(s+1)} - \frac{1}{2(s+\alpha)} \right) 
ds \\
& = \begin{cases}
\frac{2}{z} - \frac{1}{z+\alpha}, & \text{ if } |z| > \sqrt{\alpha}, \\
-\frac{1}{z+1}, 
& \text{ if } |z| < \sqrt{\alpha}.
\end{cases} 
\end{align*}
Recalling the definition \eqref{phideflow} of $\phi(z)$ and the expression \eqref{eq:Vprime} for $V_\alpha'(z)$, we conclude
\begin{align}\label{eq:gprimelow}
\phi'(z) = g'(z) - \frac{V_\alpha'(z)}{2}.
\end{align}
Integrating \eqref{eq:gprimelow} from  $\sqrt{\alpha}$ to $z$ along a path that does not intersect $(-\infty,0] \cup \Sigma_0$, we find
\begin{equation} \label{eq:phiglow} 
\phi(z) = -\frac{\pi i}{2} +
g(z) - \frac{V_\alpha(z)}{2} - g_-(\sqrt{\alpha}) +  \frac{V_\alpha(\sqrt{\alpha})}{2}, 
\end{equation}
if $|z| > \sqrt{\alpha}$. For $|z| < \sqrt{\alpha}$ we 
similarly find
\[ \phi(z) = \frac{\pi i}{2} + g(z) - \frac{V_{\alpha}(z)}{2} - 
g_+(\sqrt{\alpha}) + \frac{V_{\alpha}(\sqrt{\alpha})}{2}. \]
Then \eqref{eq:phiglow} also holds for $|z| < \sqrt{\alpha}$,
since $g_+(\sqrt{\alpha}) = g_-(\sqrt{\alpha}) + \pi i$,
as can be verified from the definition of the branch of $\log(z-s)$
that was used in the definition of $g$.
Thus \eqref{phigrelation} holds for $0 < \alpha < \frac{1}{9}$ 
in the low temperature regime  because of the definition of the
constant $\ell$. The identities \eqref{gjumponSigma0}
and \eqref{gjumponSigma0b} follow from \eqref{phigrelation}
in the same way as in the case $\frac{1}{9} < \alpha \leq 1$.

\subsection{Calculations leading to $Q_{\alpha}$} \label{subsec:comments}

The reader may wonder how to obtain the expressions \eqref{eq:Qalphahigh}
and \eqref{eq:Qalphalow}. One clue is that we need the residues
\eqref{eq:Qalphares}.
This translates into the three conditions (which are consistent with  \eqref{eq:Qalpha})
\begin{equation} \label{eq:Qalphalimits}
\begin{aligned} 
\lim_{z \to 0} z^2 Q_{\alpha}(z) & = 1, \\
\lim_{z \to -1} (z+1)^2 Q_{\alpha}(z) & = \frac{1}{4}, \\
\lim_{z \to -\alpha} (z+\alpha)^2 Q_{\alpha}(z) & = \frac{1}{4}.
\end{aligned}
\end{equation}
It is also clear from \eqref{eq:Qalpha} and \eqref{eq:Vprime}
that
\[ \lim_{z \to \infty} z^2 Q_{\alpha}(z) = 1. \]
Then
\begin{equation} \label{eq:Qtbd} 
Q_{\alpha}(z) = \frac{z^4 + A z^3 + B z^2 + C z + D}{z^2(z+1)^2(z+\alpha)^2} \end{equation}
and the limits \eqref{eq:Qalphalimits} give us three equations
for the coefficients, namely
\begin{equation} \label{BCDintermsofA}
%\begin{aligned}
D  = \alpha^2, \qquad
C = \alpha A, \qquad
B  = (\alpha+1)A - \frac{3}{4} \alpha^2 - \frac{1}{2} \alpha - \frac{3}{4}. 
%\end{aligned}
\end{equation}
which leaves us with one parameter $A$ only.

To proceed, we make the one-cut assumption which says 
that $Q_{\alpha}$ should have at least one multiple zero. It means
that the discriminant of the numerator polynomial should be zero.
The discriminant factors as %(computation done by Maple)
\[ \alpha^2(1-\alpha)^2 (A-\alpha-3)^2
(A-3\alpha-1)^2
\left(A^2-  \frac{3}{2} (1+ \alpha) A + \frac{9}{16}(1-\alpha)^2 \right) \]
which leaves us with four possible choices for $A$,
namely $A_1 = 3 +\alpha$, $A_2 = 3\alpha+1$, $A_3 = \frac{3}{4}(1-\sqrt{\alpha})^2$, and $A_4 = \frac{3}{4}(1+\sqrt{\alpha})^2$.

For $\alpha=1$ we should recover \eqref{eq:Qalphais1}
which means that we have to take $A=A_4$ for $\alpha =1$,
and then by continuity also for $\alpha$  between $1$ and a critical
value of $\alpha$. This leads to the formulas \eqref{eq:Qalphahigh}
and \eqref{eq:zpmhigh}. 
The critical value is when $z_+(\alpha) = z_-(\alpha)$, and this
happens for $\alpha = 1/9$.

For $\alpha = \frac{1}{9}$, the two values $A_2$ and $A_4$ coincide,
and for $\alpha < \frac{1}{9}$ we find that $A_2$ takes over.
This leads to the formulas \eqref{eq:Qalphalow} and \eqref{eq:zpmlow}
with two double zeros of $Q_{\alpha}$.

 \section{Orthogonal polynomials and Riemann--Hilbert problem}
 \label{sec:RHP}

We will now prove the existence of the orthogonal polynomials and pose a RH problem for the reproducing kernel $R_N(w,z)$ that appears in the double contour integral in the kernel \eqref{eq:kernel}.
 
 \subsection{Existence of the orthogonal polynomials}

 \begin{proposition} \label{prop:prop51}
 	Let $0 < \alpha \leq 1$ and $N \in \mathbb N$. Then for every $n=0,1, \ldots, 2N$
 	there is a unique monic polynomial $p_n$ of degree $n$ such that
 	\begin{equation}
 	\label{eq:orthogonality}
 	\frac{1}{2\pi i} \oint_{\gamma} p_n(z) z^j
 	\frac{(z+1)^N (z+\alpha)^N}{z^{2N}} dz = 0, 
 	\qquad j=0,1, \ldots, n-1. 
 	\end{equation}
 	In addition, if $n \leq 2N-1$, then
 	\begin{equation} \label{eq:normconstant} 
 	\kappa_n = \frac{1}{2\pi i} \oint_{\gamma}
 	\left(p_n(z)\right)^2 \frac{(z+1)^N (z+\alpha)^N}{z^{2N}} dz \neq 0.
 	\end{equation}
 \end{proposition} 
 
 \begin{proof}
 	The orthogonality condition \eqref{eq:orthogonality} is associated with the
 	non-Hermitian bilinear form
 	\[ \langle f, g \rangle =
 	\frac{1}{2\pi i}	\oint_{\gamma} f(z) g(z) 
 	\frac{(z+1)^N (z+\alpha)^N}{z^{2N}} dz \]
 	defined for polynomials $f$ and $g$. The polynomial $p_n$ exists and is unique
 	if and only if the $n \times n$ matrix of moments
 	\begin{equation} \label{eq:matrixMn} 
 	M_n = \begin{bmatrix} \langle z^j, z^k \rangle \end{bmatrix}_{j,k=0}^{n-1}
	\end{equation}
 	is invertible.
 	We use the Lindstr\"om-Gessel-Viennot (LGV) lemma to prove that this is the case for $n \leq 2N$.
	
 	Consider the directed graph on $\mathbb Z^2$ with an edge
 	from $(i,j)$ to $(i',j')$ if and only if $i' = i + 1$ and
 	$j' - j \in \{0,1\}$. The weights on the edges are 
 	\begin{align*} 
 	w((i,j), (i+1, j)) & = \begin{cases} \alpha & \text{ if $i$ is even}, \\
 	1 & \text{ if $i$ is odd}, \end{cases} \\
 	w((i,j), (i+1, j+1)) & = 1. 
 	\end{align*}
 	For two vertices $A, B \in \mathbb{Z}^2$ we define
 	\[ w(A,B) = \sum_{P : A \to B } \prod_{e \in P} w(e), \]
 	where the sum is over all directed paths $P$ on the graph from vertex
 	$A$ to vertex $B$. If there are no such paths then $w(A,B) = 0$.
 	
 	We assume $0 \leq n \leq 2N$ and we take  
 	vertices $A_j = (0, j)$ and $B_j = (2N, 2N-n+j)$ for $j=0, 1, \ldots, n-1$.
 	The LGV lemma \cite{GV} states that
 	$ \det \left[ w(A_j, B_k) \right]_{j,k=0}^{n-1}$
 	is equal to the weighted sum of all non-intersecting path systems 
 	from $A_0, \ldots A_{n-1}$ to $B_0, \ldots, B_{n-1}$.
 	It is easy to verify that there exist such non-intersecting path systems
 	(due to the fact that $0 \leq n \leq 2N$). Each non-intersecting
 	path system has a positive weight since $\alpha > 0$. Therefore
 	$\det \left[ w(A_j, B_k) \right]_{j,k=0}^{n-1}  > 0$, which, in particular,
 	implies that 
 	\begin{equation} \label{eq:matrixWn} 
 	W_n = \left[ w(A_j, B_k) \right]_{j,k=0}^{n-1} \end{equation}
 	is an invertible matrix.
 	
 	To calculate $w(A_j,B_k)$ we observe that any path from $A_j$
 	to $B_k$ is of length $2N$ with $n-k+j$ horizontal edges.
 	The weight of such a path is $\alpha^l$ where $l$ is the
 	number of horizontal edges at an even level.
 	We pick $l$ out of the possible $N$ even levels, and $n-k+j-l$
 	out of the possible $N$ odd levels, and we see that there
 	are  $ \binom{N}{l} \binom{N}{n-k+j-l} $ paths from $A_j$ to $B_k$
 	with weight $\alpha^l$. Summing over $l$ yields
 	\[ w(A_j,B_k) = \sum_{l=0}^{N} \binom{N}{l} \binom{N}{n-k+j-l} \alpha^l. \]
 	
 	This sum over products of binomial coefficients is easily seen to be 
 	equal to the coefficient of $z^{2N-n+k-j}$ in the product
 	$(z+1)^N (z+\alpha)^N$. Therefore, by Cauchy's theorem
 	\begin{align*} 
 	w(A_j,B_k) & = \frac{1}{2\pi i} \oint_{\gamma}
 	\frac{(z+1)^N (z+\alpha)^N}{z^{2N-n+k-j+1}} dz \\
 	& = \langle z^j, z^{n-k-1} \rangle.
 	\end{align*}
 	Comparing \eqref{eq:matrixMn} and \eqref{eq:matrixWn} we then see
 	that $M_n$ is obtained from $W_n$ by reversing the order of the 
 	columns. Since $W_n$ is invertible, also $M_n$ is invertible,
 	and it follows that $p_n$ uniquely exists.
 	
 	To prove \eqref{eq:normconstant} let us assume that $\kappa_n = 0$. Then by orthogonality we have $\langle p_n, z^j \rangle = 0$ 
 	not only for $j=0,1, \ldots, n-1$ but also for $j=n$. It follows
 	again by orthogonality of $p_{n+1}$ in case $n \leq 2N-1$, that 
 	$ \langle p_{n+1} + p_n, z^j \rangle = 0 $ for every $j=0, 1, \ldots, n$.
 	However, we established that $p_{n+1}$ is the only monic polynomial
 	of degree $n+1$ with these properties (if $n\leq 2N-1$). This
 	contradiction shows that $\kappa_n \neq 0$.
 \end{proof}
 \subsection{Riemann-Hilbert problem} \label{subsec:rhp}
 
 It is well-known that the orthogonal polynomials and the associated Christoffel--Darboux kernel can be  characterized by a RH problem.
 
 \begin{rhp}\label{rhp:Y} Let $\gamma_0$ be the circle
 	of radius $\sqrt{\alpha}$ around $0$ with positive direction.
 Find a function $Y : \mathbb{C}\setminus \gamma_0 \to \mathbb{C}^{2\times 2}$  with the following properties:
 \begin{itemize}
 	\item[(a)] $Y : \mathbb{C}\setminus \gamma_0 \to \mathbb{C}^{2\times 2}$ is analytic.
 	\item[(b)] The limits of $Y(z)$ as $z$ approaches $\gamma_0$ from inside and outside exist, are continuous on $\gamma_0$ and are denoted by $Y_+$ and $Y_-$, respectively. Furthermore they are related by
 	\begin{equation}\label{jump relations of Y}
 	Y_{+}(z) = Y_{-}(z) \begin{pmatrix}
 	1 & \frac{(z+1)^{N}(z+\alpha)^{N}}{z^{2N}} \\ 0 & 1
 	\end{pmatrix} \qquad \text{ for } z \in \gamma_0.
 	\end{equation}
 	\item[(c)] $Y(z) = \left(I + \bigO (z^{-1})\right) \begin{pmatrix}
 	z^{N} & 0 \\ 0 & z^{-N}
 	\end{pmatrix}$ as $z \to \infty$.
 \end{itemize}
\end{rhp}

The RH problem \ref{rhp:Y} is due to Fokas, Its, and Kitaev \cite{FIK}.
Its solution contains the orthogonal polynomials of degrees
$N$ and $N-1$ that uniquely exist by Proposition \ref{prop:prop51},
\begin{equation} \label{eq:Yz} 
	Y(z) = \begin{pmatrix} p_N(z) & 
	 \ds \frac{1}{2\pi i} \oint_{\gamma_0} p_N(s) \frac{(s+1)^N(s+\alpha)^N}{s^{2N}} \frac{ds}{s-z} \\[10pt]
	-\kappa_{N-1}^{-1} p_{N-1}(z) & 
	\ds -\frac{\kappa_{N-1}^{-1}}{2\pi i} \oint_{\gamma_0} p_{N-1}(s) \frac{(s+1)^N(s+\alpha)^N}{s^{2N}} \frac{ds}{s-z} 
\end{pmatrix}, \end{equation}
for $z \in \mathbb C \setminus \gamma_0$. 
 
 \begin{proposition}\label{prop:RHPforRN}
 	\begin{enumerate}
 		\item[\rm (a)] 
 	The kernel $R_N$ is given in terms of the solution $Y$ of the
 RH problem \ref{rhp:Y} by
\begin{align}\label{RnzwY}
 R_N(w,z) = \frac{1}{z-w} \begin{pmatrix} 0 & 1 \end{pmatrix}
 Y^{-1}(w) Y(z) \begin{pmatrix} 1 \\ 0 \end{pmatrix}.
 \end{align}
 \item[\rm (b)] Also for $w, z \in \mathbb C \setminus \gamma_0$,
 	\begin{align} \nonumber
 	\mathcal R_N(w,z) & := \begin{pmatrix}
	 	1 & 0 \end{pmatrix} Y^{-1}(w) Y(z) \begin{pmatrix} 1 \\ 0 \end{pmatrix} \\
	 	& = \frac{1}{2\pi i} \oint_{\gamma_0} R_N(s,z)
	 		\frac{(s+1)^N(s+\alpha)^N}{s^{2N}}
	 		\frac{s-z}{s-w} ds. \label{eq:calRN}
 	\end{align}
 \end{enumerate}
 \end{proposition}
 
\begin{proof}
	The formula \eqref{RnzwY} is a reformulation of the 
	Christoffel-Darboux formula \eqref{eq:CDkernel}, as can be 
	readily checked from \eqref{eq:Yz} together with the
	fact that $\det Y \equiv 1$.
	The formula \eqref{eq:calRN} is obtained from \eqref{eq:Yz}
	in a similar way. 
\end{proof}
 
\subsection{First transformation of the RH problem}

The steepest descent analysis of the RH problem \ref{rhp:Y} for $Y$ is fairly
standard by now. It is modelled after the method developed by Deift et al.~\cite{DKMVZ} for
orthogonal polynomials on the real line. The extension to the complex plane
is standard, once one has identified the correct contour $\gamma_0$
with the equilibrium  measure $\mu_0$. In the high temperature regime we basically follow \cite{DKMVZ}
including the construction of Airy parametrices for the local analysis at 
branch points $z_{\pm}$.
The RH analysis in the low temperature regime is even simpler since we can separate
contours and no local analysis is needed.  The critical case $\alpha = 1/9$
is more difficult, but can be handled with the construction of a local
parametrix built out of Lax pair solutions associated with the Hastings-McLeod
solution of Painlev\'e II. This is similar to the construction in \cite{CK}
for orthogonal polynomials on the real line in cases where the equilibrium
density vanishes quadratically at an interior point of its support. 
We will not give any details for this case.

In terms of the function $V_\alpha$ defined in \eqref{eq:Valpha}, the jump relation \eqref{jump relations of Y} for $Y$ can be expressed as
\[ Y_{+}(z) = Y_{-}(z) \begin{pmatrix}
1 & e^{-NV_\alpha(z)} \\ 0 & 1
\end{pmatrix} \qquad \text{ for } z \in \gamma_0. \]

The first transformation $Y \mapsto T$ uses the $g$-function 
from Definition~\ref{mu0gdef} to normalize the RH problem at infinity. 
We define
\begin{equation} \label{Y to T transformation}
T(z) = e^{\frac{N\ell}{2}\sigma_{3}}Y(z)e^{-Ng(z)\sigma_{3}}e^{-\frac{N\ell}{2}\sigma_{3}}, \qquad \sigma_3 \coloneqq \begin{pmatrix} 1 & 0 \\0 & -1 \end{pmatrix}.
\end{equation}
The jumps in the RH problem for $T$ are conveniently expressed in terms of
the function $\phi$ defined in \eqref{phidefhigh} and \eqref{phideflow}.
From the identities \eqref{phigrelation}, \eqref{gjumponSigma0}, and \eqref{gjumponSigma0b} and the definition \eqref{Y to T transformation}, we find the following RH problem.
\begin{rhp} \label{rhp:T}
	$T$ satisfies
	\begin{itemize}
		\item[(a)] $T : \mathbb{C}\setminus \gamma_0 \to \mathbb{C}^{2\times 2}$ is analytic.
		
		\item[(b)] $T$ has boundary values on $\gamma_0$ that satisfy
\begin{align}	
\hspace{-1.3cm}		T_{+}(z) & = T_{-}(z) \begin{pmatrix}
		e^{-2N\phi_{+}(z)} & 1 \\
		0 & e^{-2N\phi_{-}(z)}
		\end{pmatrix}, & & \hspace{-1cm} \text{ for } z \in  \Sigma_0 \subset \gamma_0, \label{eq:Tjump1} \\
\hspace{-1.3cm}		T_{+}(z) & = T_{-}(z) \begin{pmatrix}
		1 & e^{2N\phi(z)} \\
		0 & 1
		\end{pmatrix}, & & \hspace{-1cm} \text{ for } z \in \gamma_0 \setminus \Sigma_0. \label{eq:Tjump2}
		\end{align}
		\item[(c)] $T(z) = I + \bigO(z^{-1})$ as $z \to \infty$.
	\end{itemize}
\end{rhp}

Note that $T$ depends on $N$, even though this is not indicated in the 
notation.
What is important for us, is that $T$ and $T^{-1}$ remain bounded
as $N \to \infty$, provided we stay away from the branch points $z_{\pm}(\alpha)$ (only in the high temperature regime).
We summarize what we need from the RH analysis 
in the following proposition.

\begin{proposition} \label{prop:TandTinvsmall}
	\begin{enumerate}
		\item[\rm (a)] If $0 < \alpha \leq \frac{1}{9}$, then 
		both $T(z)$ and $T(z)^{-1}$ are uniformly bounded for $z \in \mathbb{C}\setminus \gamma_{0}$ as $N \to \infty$.
		\item[\rm (b)] If $\frac{1}{9} < \alpha \leq 1$, then
		$T(z) = \bigO(N^{1/6})$ and $T^{-1}(z) = \bigO(N^{1/6})$
		as $N \to \infty$, uniformly for $z \in \mathbb C \setminus \gamma_0$.
		In addition, for every $\delta > 0$, we have that
		$T(z)$ and $T^{-1}(z)$ are
		bounded as $N \to \infty$ uniformly for $z$ 
		in the domain 
		\begin{equation} \label{eq:zawayfrombranchpoints} 
		\{ z \in \mathbb C \mid |z-z_+(\alpha)| \geq \delta, 
		|z-z_-(\alpha)| \geq \delta \}. \end{equation}
	\end{enumerate}
\end{proposition}
The proposition is a result of the steepest descent analysis that
we will perform next for the two regimes separately. 

Because of \eqref{Y to T transformation} and the formula \eqref{eq:calRN}
for $\mathcal R_N$, we have
\begin{equation} \label{eq:calRNinT} 
	\mathcal R_N(w,z) = \begin{pmatrix} 1 & 0 \end{pmatrix}	
	T^{-1}(w) T(z) \begin{pmatrix} 1 \\ 0 \end{pmatrix}	e^{N(g(z)-g(w))} 
	\end{equation}
and before turning to the proof of Proposition \ref{prop:TandTinvsmall}
we  note the following consequence.

\begin{corollary} \label{cor:TandTinvsmall}
\begin{enumerate}
	\item[\rm (a)] If $0 < \alpha \leq \frac{1}{9}$ then
	$\mathcal R_N(w,z) e^{N(g(w) - g(z))}$ remains bounded
	as $N \to \infty$, uniformly for $w \in \mathbb C \setminus \gamma_0$
	and $z \in \mathbb C\setminus \gamma_0$.
	\item[\rm (b)] If $\frac{1}{9} < \alpha \leq 1$ then
	$\mathcal R_N(w,z) e^{N(g(w) - g(z))}$
	remains bounded as $N \to \infty$, uniformly 
	for $w \in \mathbb C \setminus \gamma_0$
	and $z \in \mathbb C$, both in the domain \eqref{eq:zawayfrombranchpoints} 
	for some $\delta > 0$.
	\item[\rm (c)] If $\frac{1}{9} < \alpha \leq 1$, then
	 the analytic continuation of 
	$w \mapsto \mathcal R_N(w,z) e^{N(g(w)-g(z))}$
	from the disk $|w| < \sqrt{\alpha}$ across $\gamma_0 \setminus \Sigma_0$ 
	into the domain
	bounded by $\Sigma_{-1}$ and $\gamma_0 \setminus \Sigma_0$
	remains bounded as $N \to \infty$, again uniformly
	for $w$ and $z$ in the domain  \eqref{eq:zawayfrombranchpoints} 
	for some $\delta > 0$.
\end{enumerate}	
\end{corollary}
\begin{proof}
	Parts (a) and (b) are immediate from 	
	 \eqref{eq:calRNinT} and Proposition \ref{prop:TandTinvsmall}.
	 
	Because of \eqref{eq:calRNinT} and the jump condition \eqref{eq:Tjump2} 
	for $T$ along $\gamma_0 \setminus \Sigma_0$, the analytic
	continuation from part (c) is given by
	\[ \begin{pmatrix} 1 & - e^{2N \phi(w)} \end{pmatrix}	T^{-1}(w)
		T(z) \begin{pmatrix} 1 \\ 0 \end{pmatrix} \]
	Since $\Re \phi(w) < 0$ for $w$ in the region under consideration
	in part (c), see for example Figure \ref{fig:sign of xi high},
	part (c) follows from Proposition \ref{prop:TandTinvsmall}
	as well. 	
\end{proof}

\subsection{Proof of Proposition \ref{prop:TandTinvsmall} (a) }
\begin{proof}
	Suppose $ 0 < \alpha < \frac{1}{9}$. Then we can find contours $\gamma_{+}$ 
	and $\gamma_{-}$ lying in the interior and exterior of $\gamma_0 = \Sigma_0$, respectively,
	such that  
	\begin{equation}\label{rephiongammaplusminus}
	\Re \phi(z) > \epsilon > 0 \qquad \text{for all } z \in \gamma_{+}\cup \gamma_{-}
	\end{equation}
	for some fixed $\epsilon > 0$, see Figure \ref{fig:opening of the lenses low}. 
	\begin{figure}[ht]
		\begin{center}
			\begin{tikzpicture}[master,scale = 1.3,every node/.style={scale=1.3}]
			\node at (0,0) {\includegraphics[width=9cm]{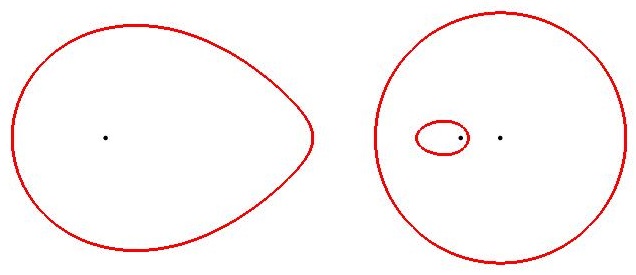}};
			\draw [line width=0.75 mm] (2.58,-0.01) circle(1.77);
			%	\draw [line width=0.65 mm] (-3.05,0) circle(0.8);
			\draw[black,arrows={-Triangle[length=0.33cm,width=0.22cm]}]
			($(2.52,1.76)$) --  ++(-0.0001,0);
			%		\draw[black,arrows={-Triangle[length=0.33cm,width=0.22cm]}]
			%		($(-3.2,0.8)$) --  ++(-0.0001,0);
			
			% the lenses
			
			\draw [line width=0.75 mm] (2.58,-0.01) circle(2.2);
			
			\draw [line width=0.75 mm] (2.58,-0.01) circle(1.35);
			
			\draw[black,arrows={-Triangle[length=0.33cm,width=0.22cm]}]
			($(2.52,2.19)$) --  ++(-0.0001,0);
			
			\draw[black,arrows={-Triangle[length=0.33cm,width=0.22cm]}]
			($(2.52,1.34)$) --  ++(-0.0001,0);
			
			\node at (2.7,1.12) {\footnotesize $\gamma_{+}$};
			\node at (2.7,2.4) {\footnotesize $\gamma_{-}$};
			\node at (2.65,1.57) {\footnotesize $\gamma_{0}$};
			%	\node at (-2.99,1.01) {\footnotesize $\gamma_*$};
			\node at (-.7,1.2) {\footnotesize $\Sigma_{-1}$};
			\node at (2.35,0.3) {\footnotesize $\Sigma_{-\alpha}$};
			\end{tikzpicture}
		\end{center}
		\caption{\label{fig:opening of the lenses low} The jump contour $\gamma_0 \cup \gamma_+ \cup \gamma_-$ for the RH problem \ref{rhp:S} for $S$ (black), the curves $\Sigma_{-1}$ and $\Sigma_{-\alpha}$ (red), and the points $-1, -\alpha, 0$ (black dots) in the low temperature regime.}
	\end{figure}
	
	We define 
	\begin{equation}\label{eq:TtoSlow}
	S(z) = T(z) \times \begin{cases}
	\begin{pmatrix}
	1 & 0 \\ -e^{-2N\phi(z)} & 1 
	\end{pmatrix}, & \text{for $z$ between $\gamma_{0}$ and $\gamma_+$}, \\
	\begin{pmatrix}
	1 & 0 \\ e^{-2N\phi(z)} & 1 
	\end{pmatrix}, & \text{for $z$ between $\gamma_{0}$ and $\gamma_-$}, \\
	I, & \text{elsewhere}.
	\end{cases}
	\end{equation}
	Then $S$ satisfies the following RH problem.
	
	\begin{rhp} \label{rhp:S} \,
		\begin{itemize} 
			\item[(a)] $S : \mathbb{C}\setminus (\gamma_0 \cup \gamma_{+} 
			\cup \gamma_{-}) \to \mathbb{C}^{2\times 2}$ is analytic.
			\item[(b)] $S$ has boundary values on $\gamma_0$, $\gamma_+$
			and $\gamma_-$ that satisfy 
			\begin{align}
			S_{+}(z) & = S_{-}(z) \begin{pmatrix}
			1 & 0 \\ e^{-2N\phi(z)} & 1
			\end{pmatrix}, & & \text{ for } z \in \gamma_{+} \cup \gamma_{-}, \\
			S_{+}(z) & = S_{-}(z) \begin{pmatrix}
			0 & 1 \\ 	-1 & 0
			\end{pmatrix}, & & \text{ for } z \in \gamma_{0}.
			\end{align}
			\item[(c)] $S(z) = I + \bigO(z^{-1})$ as $z \to \infty$.
		\end{itemize}
	\end{rhp}
	
	We remove the constant jump on $\gamma_0$ by defining
	\begin{equation} \label{eq:StoRlow} 
	R(z) = S(z) \times \begin{cases}
	\begin{pmatrix} 0 & -1 \\ 1 & 0 \end{pmatrix}, & \text{ for $z$ inside $\gamma_0$}, \\
	I, & \text{ for $z$ outside $\gamma_0$}.
	\end{cases} \end{equation}
	Of course $R$ should  not be confused with the reproducing
	kernel $R_N$, as these are totally different objects.
	The matrix valued function 	$R$ satisfies the following RH problem.
	\begin{rhp} \label{rhp:R} \,
		\begin{itemize} 
			\item[(a)] $R: \mathbb{C}\setminus (\gamma_{+} \cup \gamma_{-}) 
			\to \mathbb{C}^{2\times 2}$ is analytic.
			\item[(b)] $R$ has boundary values on $\gamma_+$
			and $\gamma_-$ that satisfy 
			\begin{align}
			R_{+}(z) & = R_{-}(z) \begin{pmatrix}
			1 & -e^{-2N\phi(z)} \\ 0 & 1
			\end{pmatrix}, & & \text{ for } z \in \gamma_{+}, \\
			R_{+}(z) & = R_{-}(z) \begin{pmatrix}
			1 & 0 \\ e^{-2N\phi(z)} & 1
			\end{pmatrix}, & & \text{ for } z \in \gamma_{-}.
			\end{align}
			\item[(c)] $R(z) = I + \bigO(z^{-1})$ as $z \to \infty$.
		\end{itemize}
	\end{rhp}		
	Since $\Re \phi > \epsilon > 0$ for $z \in \gamma_+ \cup \gamma_-$
	the jumps in the RH problem for $R$ are exponentially close
	to the identity matrix as $N \to \infty$. By standard estimates
	on small norm RH problems \cite{Deift}, we find
	$R(z) = I + \bigO(e^{-\epsilon N})$ as $N \to \infty$,
	and in particular $R$ and $R^{-1}$ are uniformly bounded as
	$N \to \infty$, uniformly on $\mathbb C$.
	Tracing back the transformations \eqref{eq:StoRlow} and
	\eqref{eq:TtoSlow} it then also follows that $T$ and $T^{-1}$
	are uniformly bounded as $N \to \infty$, uniformly
	on $\mathbb C$, since $\Re \phi \geq 0$ in the annular 
	region bounded by $\gamma_+$ and $\gamma_-$. This proves
	Proposition \ref{prop:TandTinvsmall} for $\alpha < \frac{1}{9}$.
	
	In case $\alpha = \frac{1}{9}$, we are not able to choose
	$\gamma_+$ and $\gamma_-$ such that \eqref{rephiongammaplusminus} 
	holds on the full contours. Instead we let $\gamma_+$ and $\gamma_-$ 
	go to  $\gamma_0$ at the 
	critical point $-\sqrt{\alpha} = - \frac{1}{3}$, and
	we can do it
	in such a way $\Re \phi > 0$
	on $(\gamma_+ \cup \gamma_-) \setminus \{-\frac{1}{3}\}$.
	Then we can proceed as in the case $\alpha < \frac{1}{9}$
	described above,
	except that we have to build a local parametrix at $-\frac{1}{3}$.
	This is done with the help of a special parametrix \cite{CK}
	that we will not describe here. We only need to know that it is
	uniformly bounded as $N \to \infty$ and then Proposition \ref{prop:TandTinvsmall} follows as before.
\end{proof}

\begin{figure}[ht]
	\begin{center}
		\begin{tikzpicture}[master,scale = 1.3,every node/.style={scale=1.3}]
		\node at (0,0) {\includegraphics[width=9.23cm]{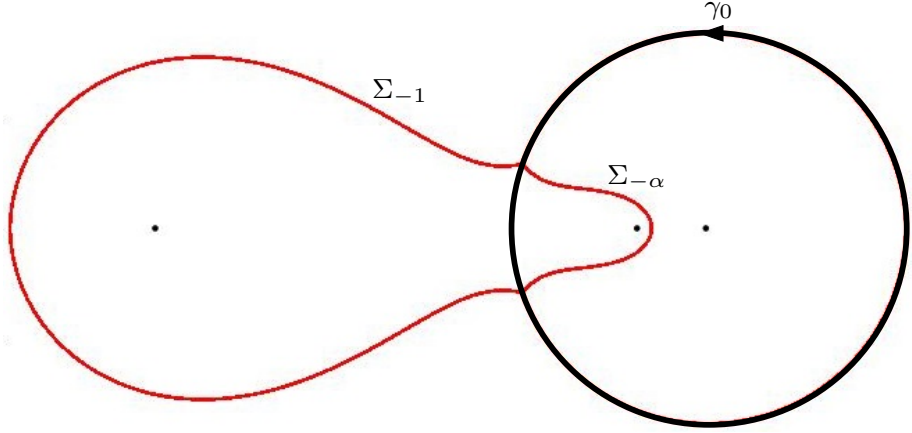}};
		\draw [line width=0.75 mm] (2.528,0.007) circle(2);
		%		\draw [line width=0.65 mm] (-3.05,0) circle(0.8);
		\draw[black,arrows={-Triangle[length=0.33cm,width=0.22cm]}]
		($(2.42,2)$) --  ++(-0.0001,0);
		%		\draw[black,arrows={-Triangle[length=0.33cm,width=0.22cm]}]
		%		($(-3.2,0.8)$) --  ++(-0.0001,0);
		
		\node at (2.63,2.23) {\footnotesize $\gamma_0$};
		%		\node at (-3,1.02) {\footnotesize $\gamma_*$};
		\node at (-.6,1.4) {\footnotesize $\Sigma_{-1}$};
		\node at (1.8,.53) {\footnotesize $\Sigma_{-\alpha}$};
		
		\end{tikzpicture}
	\end{center}
	\caption{\label{fig:valid contour C high} 
		The jump contour $\gamma_0$ for the RH problem \ref{rhp:Y} for $Y$ (black), the curves $\Sigma_{-1}$ and $\Sigma_{-\alpha}$ (red), and the points $-1, -\alpha, 0$ (black dots) in the high temperature regime.}
\end{figure}

\subsection{Proof of Proposition \ref{prop:TandTinvsmall} (b) }
\begin{proof}
	Suppose $\frac{1}{9} <  \alpha  \leq 1$ and let $Y(z)$ denote the solution 
	of the RH problem \ref{rhp:Y} with jump contour $\gamma_0$.
	See Figure \ref{fig:valid contour C high} for $\gamma_0$ together
	with the contours $\Sigma_{-1}$ and $\Sigma_{-\alpha}$
	that enclose the bounded domain where $\Re \phi < 0$
	in the high temperature regime.
	
	The first transformation $Y \to T$ is given by \eqref{Y to T transformation} 
	and $T$ satisfies the RH problem \ref{rhp:T}. 
	In the second transformation, we open up lenses $\gamma_{+}$ and 
	$\gamma_{-}$ around $\Sigma_0 \subset \gamma_0$ as in Figure \ref{fig:opening of the lenses high} such that $\Re \phi> 0$
	on $(\gamma_+ \cup \gamma_-) \setminus \{z_+(\alpha), z_-(\alpha)\}$
	and define $S$ as (it is similar to  \eqref{eq:TtoSlow})
	\begin{equation}\label{eq:TtoShigh}
	S(z) = T(z) \times \begin{cases}
	\begin{pmatrix}
	1 & 0 \\ -e^{-2N\phi(z)} & 1 
	\end{pmatrix}, & \text{for $z$ between $\Sigma_{0}$ and $\gamma_+$}, \\
	\begin{pmatrix}
	1 & 0 \\ e^{-2N\phi(z)} & 1 
	\end{pmatrix}, & \text{for $z$ between $\Sigma_{0}$ and $\gamma_-$}, \\
	I, & \text{elsewhere}.
	\end{cases}
	\end{equation}
	Then $S$ satisfies
	\begin{rhp} \label{rhp:Shigh} \,
		\begin{itemize} 
			\item[(a)] $S : \mathbb{C}\setminus (\gamma_0 \cup \gamma_{+} 
			\cup \gamma_{-}) \to \mathbb{C}^{2\times 2}$ is analytic.
			\item[(b)] $S$ has boundary values on $\gamma_0$, $\gamma_+$
			and $\gamma_-$ that satisfy 
			\begin{align}
			S_{+}(z) & = S_{-}(z) \begin{pmatrix}
			1 & 0 \\ e^{-2N\phi(z)} & 1
			\end{pmatrix}, & & \text{ for } z \in \gamma_{+} \cup \gamma_{-}, \\
			S_{+}(z) & = S_{-}(z) \begin{pmatrix}
			0 & 1 \\ 	-1 & 0
			\end{pmatrix}, & & \text{ for } z \in \Sigma_{0}, \\
			S_{+}(z) & = S_{-}(z) \begin{pmatrix}
			1 & e^{2N \phi(z)} \\ 0 & 1
			\end{pmatrix}, & & \text{ for } z \in \gamma_{0} \setminus \Sigma_0.
			\end{align}
			\item[(c)] $S(z) = I + \bigO(z^{-1})$ as $z \to \infty$.
		\end{itemize}
	\end{rhp}
	
	\begin{figure}[ht]
		\begin{center}
			\begin{tikzpicture}[master,scale = 1.3,every node/.style={scale=1.3}]
			\node at (0,0) {\includegraphics[width=9.23cm]{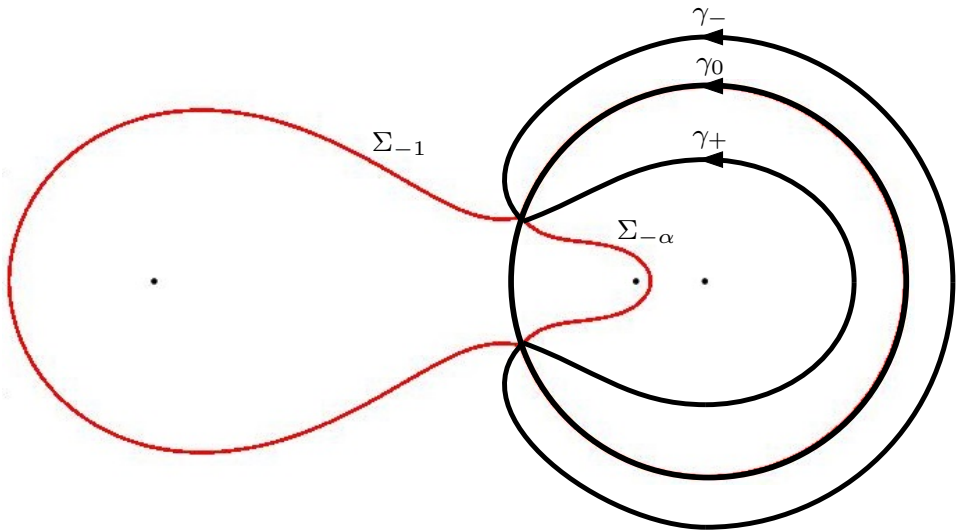}};
			\draw [line width=0.75 mm] (2.528,0.007) circle(2);
			%		\draw [line width=0.65 mm] (-3.05,0) circle(0.8);
			\draw[black,arrows={-Triangle[length=0.33cm,width=0.22cm]}]
			($(2.42,2)$) --  ++(-0.0001,0);
			%		\draw[black,arrows={-Triangle[length=0.33cm,width=0.22cm]}]
			%		($(-3.2,0.8)$) --  ++(-0.0001,0);
			
			\draw[black, line width=0.65 mm] (0.65,0.62) to [out=135, in=180] (2.5,2.5);
			\draw[black, line width=0.65 mm] (2.5,2.5) to [out=0, in=90] (5,0);
			\draw[black, line width=0.65 mm] (5,0) to [out=-90, in=0] (2.5,-2.5);
			\draw[black, line width=0.65 mm] (2.5,-2.5) to [out=180, in=-135] (0.65,-0.62);
			\draw[black,arrows={-Triangle[length=0.33cm,width=0.22cm]}]
			($(2.42,2.5)$) --  ++(-0.0001,0);

			\draw[black, line width=0.65 mm] (0.65,0.62) to [out=20, in=180] (2.5,1.25);
			\draw[black, line width=0.65 mm] (2.5,1.25) to [out=0, in=90] (4,0);
			\draw[black, line width=0.65 mm] (4,0) to [out=-90, in=0] (2.5,-1.25);
			\draw[black, line width=0.65 mm] (2.5,-1.25) to [out=180, in=-20] (0.65,-0.62);
			\draw[black,arrows={-Triangle[length=0.33cm,width=0.22cm]}]
			($(2.42,1.25)$) --  ++(-0.0001,0);
			
			\node at (2.55,2.2) {\footnotesize $\gamma_0$};
			%		\node at (-3,1.05) {\footnotesize $\gamma_*$};
			\node at (2.55,1.45) {\footnotesize $\gamma_{+}$};
			\node at (2.55,2.7) {\footnotesize $\gamma_{-}$};
			\node at (-.6,1.4) {\footnotesize $\Sigma_{-1}$};
			\node at (1.9,.5) {\footnotesize $\Sigma_{-\alpha}$};		
			\end{tikzpicture}
		\end{center}
		\caption{\label{fig:opening of the lenses high} 
			The jump contour $\gamma_0 \cup \gamma_+ \cup \gamma_-$ for the RH problem for $S$ (black) and the curves $\Sigma_{-1}$ and $\Sigma_{-\alpha}$ (red) in the high temperature regime.}
	\end{figure}
	
	The global parametrix $P^{(\infty)}$ is given by 
	\begin{equation} \label{eq:Pinfinityhigh}
	P^{(\infty)}(z) = \begin{pmatrix}
	\frac{1}{2}(a(z)+a(z)^{-1}) & \frac{1}{2i}(a(z)-a(z)^{-1}) \\[5pt]
	-\frac{1}{2i}(a(z)-a(z)^{-1}) & \frac{1}{2}(a(z)+a(z)^{-1})
	\end{pmatrix},
	\end{equation}
	where $a(z) := \big(\frac{z-z_{+}}{z-z_{-}}\big)^{1/4}$ is
	defined with a branch cut along $\Sigma_0$ and in such a way
	that $a(z) \to 1$ as $z \to \infty$. 
	
	In small disks $\mathcal{D}_{z_{+}}$ and $\mathcal{D}_{z_{-}}$
	around the endpoints of $\Sigma_0$ we construct local parametrices
	$P^{(z_+)}$ and $P^{(z_-)}$ with the aid of Airy functions. 
	This construction is standard by now and we do not give details.
	The only thing that concerns us is that the local parametrices depend
	on $N$ and they slightly grow with $N$, namely
	\begin{equation} \label{eq:PandPinvgrow} 
	P^{(z_{\pm})}(z) = \bigO(N^{\frac{1}{6}}), 
	\quad 
	P^{(z_{\pm})}(z)^{-1} = \bigO(N^{\frac{1}{6}})
	\quad \text{ as } N \to \infty, 
	\end{equation}
	uniformly for $z \in \mathcal D_{z_{\pm}}$.
	The third and final transformation $S \mapsto R$ is defined by
	\begin{equation} \label{eq:StoRhigh}
	R(z) = \begin{cases}
	S(z) P^{(\infty)}(z)^{-1}, & \text{ for } z \in \mathbb{C}\setminus  (\mathcal{D}_{z_{+}} \cup \mathcal{D}_{z_{-}}), \\
	S(z) P^{(z_{+})}(z)^{-1}, & \text{ for } z \in \mathcal{D}_{z_{+}}, \\
	S(z) P^{(z_{-})}(z)^{-1}, & \text{ for } z \in \mathcal{D}_{z_{-}}.
	\end{cases}
	\end{equation}
	
	Then $R$ is defined and analytic in $$
	\mathbb C \setminus
	\Big(\big( (\gamma_0 \cup \gamma_{+} \cup \gamma_{-})\setminus (\mathcal{D}_{z_{+}} \cup \mathcal{D}_{z_{+}})\big) \cup \partial \mathcal{D}_{z_+}
	\cup \partial \mathcal D_{z_-}\Big)
	$$ with jump matrices that
	are $I + \bigO(N^{-1})$ as $N \to \infty$.
	It follows that $R(z) = I + \bigO(N^{-1})$ uniformly for $z \in \mathbb C$, and in particular $R$ and $R^{-1}$ remain bounded as $N \to \infty$.
Observe that in the construction of the local parametrics, the disks $\mathcal{D}_{z_{\pm}}$ can be chosen arbitrarily small (but independent of $N$), and we choose them with radii smaller than $\delta$.
	Then following the transformations \eqref{eq:TtoShigh} and \eqref{eq:StoRhigh}, and taking note of \eqref{eq:PandPinvgrow}
	we find that $T$ and $T^{-1}$ are uniformly of order $N^{\frac{1}{6}}$
	as $N \to \infty$. Outside the disks $\mathcal D_{z_\pm}$ the global parametrix \eqref{eq:Pinfinityhigh} applies,
	which does not change with $N$, and then $T$ and $T^{-1}$ remain
	uniformly bounded. Part (b) of Proposition \ref{prop:TandTinvsmall} is now also proven.
\end{proof}

\section{Phase functions $\Phi_{\alpha}$ and $\Psi_{\alpha}$}

\subsection{Definitions} \label{subsec:Phase}
In the last two sections we analyzed the RH problem with
the $g$-function coming from the equilibrium measure as its main input.
The outcome of this analysis is in Corollary \ref{cor:TandTinvsmall}
which states that
$\mathcal R_N(w,z) e^{N(g(w)-g(z))}$ remains uniformly bounded
in certains regions, and actually (very roughly)
\begin{equation} \label{eq:calRNsim} 
 \mathcal R_N(w,z) \sim e^{N(g(z)-g(w))} \end{equation}
 as $N \to \infty$.

We now turn to the asymptotic analysis of the double contour integrals 
coming from \eqref{eq:kernel} and that give the probabilities for
the three types of lozenges, see also Theorem \ref{thm:doubleintegrals_for_lozenge_densities} below.

After deforming contours and splitting up integrals, we are able 
to rewrite the integrals with an integrand containing
\begin{equation} \label{eq:calRNtimesF}
\mathcal R_N(w,z) \frac{F(z;x_{1},y_{1})}{F(w;x_{2},y_{2})} 
\end{equation}\todo{C: changes here}
as the main $N$-dependent entry, where
\begin{equation} \label{eq:Fzxy} 
F(z;x,y) = \frac{(z+1)^{\lfloor \frac{x}{2} \rfloor} 
	(z+\alpha)^{\lfloor \frac{x+1}{2} \rfloor}}{z^y}, 
\end{equation}
see Propositions \ref{prop:deformationlow} and \ref{prop:deformationhigh}.
Recall that $x,y$ will be varying with $N$ as in \eqref{eq:scaled_variables}.
Then in view of \eqref{eq:calRNsim}, \eqref{eq:Fzxy} we see that
\eqref{eq:calRNtimesF} behaves roughly like
$e^{N (\Phi_{\alpha}(z) - \Phi_{\alpha}(w))}$ with a certain function 
$\Phi_{\alpha}$
that depends ons $(\xi,\eta) \in \mathcal H$, and that is
defined next, along with a companion function $\Psi_{\alpha}$. 

\begin{definition} For $(\xi,\eta) \in \mathcal H$ we define
	\begin{align} \nonumber
	\Phi_{\alpha}(z) & = \Phi_{\alpha}(z;\xi,\eta) \\ 
	& = g(z) + \frac{1+\xi}{2} \log \left((z+1)(z+\alpha)\right)
	- (1+\eta) \log z + \frac{\ell}{2} \nonumber \\
	 \label{Phidef} 
	& =  \phi(z) + \frac{\xi}{2}\log\left((z+1)(z+\alpha)\right) 
	- \eta \log z, \\ 
	\Psi_{\alpha}(z) & =  \Psi_{\alpha}(z;\xi,\eta) = -\Phi_{\alpha}(z;-\xi,-\eta)
	\\ \label{Psidef}
	& =  -\phi(z) + \frac{\xi}{2}\log\left((z+1)(z+\alpha)\right) - \eta \log z.
	\end{align}
\end{definition}
The equality leading to the third line in \eqref{Phidef} follows from \eqref{phigrelation} and \eqref{eq:Valpha}. Recall that $\phi' =  \pm Q_{\alpha}^{1/2}$ by Definition \ref{phidef}
and therefore by the definitions \eqref{Phidef} and \eqref{Psidef}
we have that both $\Phi_{\alpha}'$ and $\Psi_{\alpha}'$
satisfy the algebraic equation \eqref{eq:algebraiccurve} for $\Xi_{\alpha}$.

Thus $\Phi_{\alpha}'$ and $\Psi_{\alpha}'$ are two branches
of the algebraic function $\Xi_{\alpha}$. Taking note of the
different choice of branch cuts in the high temperature regime 
we can verify that
\begin{equation} \label{eq:relationPhiXi}
	\Phi_{\alpha}'(z) = \begin{cases}
		\Xi_{\alpha,+}(z), & |z| > \sqrt{\alpha}, \\
		\Xi_{\alpha,-}(z), & |z| < \sqrt{\alpha},
		\end{cases}, \qquad
		\Psi_{\alpha}'(z) = \begin{cases}
		\Xi_{\alpha,-}(z), & |z| > \sqrt{\alpha}, \\
		\Xi_{\alpha,+}(z), & |z| < \sqrt{\alpha},
		\end{cases}
	\end{equation}
in both regimes.
	
The two functions are defined and analytic in
$\mathbb C \setminus ((-\infty,0] \cup \Sigma_0)$ in 
case $0 < \alpha \leq \frac{1}{9}$ and in
$\mathbb C \setminus ((-\infty,0] \cup 
\{ \sqrt{\alpha} e^{it} \mid -\pi \leq t \leq \theta_{\alpha}\}$
in case $\frac{1}{9} < \alpha \leq 1$.
The behavior at the singularities and at infinity can be seen
from \eqref{eq:Nphinearpoles} and the definitions \eqref{Phidef}-\eqref{Psidef}, namely for $(\xi,\eta) \in \mathcal H^o$, 
\begin{equation} \label{eq:NPhinearpoles}
\begin{aligned} 
\Phi_{\alpha}(z) & = -(1+\eta) \log z + O(1) \text{ as } z \to 0, &
\lim_{z \to 0} \Re \Phi_{\alpha}(z) & = +\infty, \\
\Phi_{\alpha}(z) & = \frac{1}{2}(1+\xi) \log (z+\alpha) + O(1) \text{ as } z \to -\alpha, &
\lim_{z \to -\alpha} \Re \Phi_{\alpha}(z) & = - \infty, \\
\Phi_{\alpha}(z) & = \frac{1}{2}(1+\xi) \log (z+1) + O(1) \text{ as } z \to -1, &
\lim_{z \to - 1} \Re \Phi_{\alpha}(z) & = -\infty, \\
\Phi_{\alpha}(z) & = (1+\xi-\eta) \log z + O(1) \text{ as } z \to \infty, &
\lim_{z \to \infty} \Re \Phi_{\alpha}(z) & = +\infty 
\end{aligned}
\end{equation} 
and similarly
$\Re \Psi_{\alpha}(z) \to -\infty$ as $z \to 0$ or $z \to \infty$, and
$\Re \Psi_{\alpha}(z) \to + \infty$ as $z \to -1$ or $z \to -\alpha$.
For the limits it is important that $(\xi,\eta) \in \mathcal H^o$
so that $-1 < \xi, \eta-\xi < 1$.

For each $(\xi, \eta) \in \mathcal{L}_\alpha$, the saddle $s(\xi,\eta;\alpha)$ defined in Definition \ref{def:liquid} is 
a zero of either $\Phi'_{\alpha}$  and $\Psi'_{\alpha}$.

\begin{lemma} \label{lem:Ldivision} 
	Let $(\xi,\eta) \in \mathcal L_{\alpha}$ and $s = s(\xi,\eta;\alpha)$. 
	Then we have 
	\begin{enumerate}
		\item[\rm (a)]
		$\Phi_{\alpha}'(s) = 0$ and $|s| < \sqrt{\alpha}$
		if and only if $\xi < 0$ and $\eta < \frac{\xi}{2}$, 
		\item[\rm (b)] 
		$\Phi_{\alpha}'(s)=0$ and $|s| > \sqrt{\alpha}$ if and only if
		$\xi < 0$ and $\eta > \frac{\xi}{2}$,
		\item[\rm (c)] 
		$\Psi_{\alpha}'(s)=0$ and $|s| < \sqrt{\alpha}$ if and only if
		$\xi > 0$ and $\eta > \frac{\xi}{2}$,
		\item[\rm (d)] 
		$\Psi_{\alpha}'(s)=0$ and $|s| >\sqrt{\alpha}$ if and only 
		if $\xi > 0$ and $\eta < \frac{\xi}{2}$,
		\item[\rm (e)] $|s| = \sqrt{\alpha}$ if and only if $\xi = 0$
		or $\eta = \frac{\xi}{2}$.
	\end{enumerate}
\end{lemma}
\begin{proof}
	We use the explicit inverses for the map $(\xi,\eta) \mapsto s(\xi,\eta;\alpha)$ given in \eqref{eq:inverseLow} and \eqref{eq:inverseHigh}. 
	
	Let us  consider the low temperature regime. 
	From the formula 
	\eqref{eq:defz1z2low} for $\Xi_{\alpha,\pm}$ 
	and \eqref{eq:DplusDminus} 	it follows that $s$ is 
	a zero of $\Xi_{\alpha,\pm}$ if and only if $D_{\pm}<0$, 
	and we note that the regions $D_{\pm}<0$ are contained 
	in the regions $\eta > \frac{\xi}{2}$ and $\eta < \frac{\xi}{2}$, respectively.
	Using \eqref{eq:positivedeterminant} and  \eqref{eq:inverseLow} we see that, in the low temperature regime, $\xi$ has the same sign as
	\begin{equation} \label{eq:signofxi}
	\mp \Im \frac{(s-z_+)(s-z_-)}{(s+ \alpha)(s+1)},
	\end{equation}
	with a $\mp$-sign if $s=s(\xi,\eta;\alpha)$ is a  zero of $\Xi_{\alpha,\pm}$. The imaginary part in \eqref{eq:signofxi} is positive if $|s|>\sqrt \alpha$, negative if $|s|< \sqrt{\alpha}$ and zero if $|s|= \sqrt{\alpha}$. Combining this with  
	\eqref{eq:relationPhiXi} the statements of the lemma follow in the
	low temperature regime.  
	
	For the high temperature regime, we  use Proposition \ref{prop:hightemp}
	and the proof is analogous to the proof in the low temperature regime, but now \eqref{eq:signofxi} is replaced by $\mp \Im s Q_\alpha(s)^{\frac12}$,
	with the same choice of branch for the square root as in \eqref{eq:inverseHigh}, i.e., $Q_\alpha(s)^{\frac12}$ has a branch on $\mathcal C$. 
\end{proof}

\subsection{Critical level set of $\Re \Phi_{\alpha}$} 
 \label{subsec:levelRePhi} 
In what follows we focus on the case (a) of Lemma \ref{lem:Ldivision},
namely $(\xi,\eta) \in \mathcal L_{\alpha}$ with $\eta < \frac{\xi}{2} < 0$,
and its extension $\eta = \frac{\xi}{2} < 0$, 
which is the lower left part of the liquid region. The corresponding
saddle $s = s(\xi,\eta;\alpha)$ satisfies $\Phi'_{\alpha}(s)=0$
and $|s| < \sqrt{\alpha}$ if $\eta < \frac{\xi}{2}$. For $\eta = \frac{\xi}{2} < 0$
(which is only relevant in the high temperature regime) we have $|s| = \sqrt{\alpha}$ with
$\theta_{\alpha} < \arg s < \pi$, and we still have $\Phi'_{\alpha}(s) = 0$.

We are interested in the level set of $\Re \Phi_{\alpha}$ that contains $s$,
\begin{equation} \label{eq:NPhi}
 \mathcal N_{\Phi} = \{ z \in \mathbb C \mid
	\Re \Phi_{\alpha}(z) = \Re \Phi_{\alpha}(s) \}.
	\end{equation}
	We emphasize that $\Phi_{\alpha}$ has a branch cut along
	$\Sigma_0$. However $\Re \Phi_{\alpha}$ is well-defined and
	continuous, also on $\Sigma_0$.

Typical behaviors of $\mathcal N_{\Phi}$ are shown in Figures
\ref{fig: level zero curve in L Phi ext high 2},
\ref{fig: level zero curve in L Phi ext high} and
\ref{fig: level zero curve in L Phi ext low}.
The level set $\mathcal N_{\Phi}$ makes a cross locally at $s$
since it is a simple saddle. Four curves emanate from $s$
that are denoted by $\Gamma_1$, \ldots, $\Gamma_4$ in the figures.

It is important for us that three of these curves stay
inside $\Sigma_0$ (in low temperature regime) or inside
$\Sigma_0 \cup \Sigma_{-1}$ and connect $s$ with $\overline{s}$.
Only one of them (denoted by $\Gamma_4$ in the figures)
meets with either $\Sigma_0$ or $\Sigma_0 \cup \Sigma_{-1}$.
	
To be able to prove this we need information on the behavior
of the two functions $z \mapsto \log |z|$
and $z \mapsto \log \left| \frac{(z+1)(z+\alpha)}{z} \right|$
on $\Sigma_{-1} \cup \Sigma_{0}$. We start with a lemma.
\begin{lemma} \label{lem:zQreal} 
We have the following for $0 < \alpha \leq 1$,
\begin{enumerate}
\item[\rm (a)] $z^2 Q_{\alpha}(z) \in [0,\infty)$ if and only if
$z \in \Sigma_0 \cup \mathbb R \setminus \{-1,-\alpha\}$.
\item[\rm (b)] If $\alpha \leq \frac{1}{9}$ then 
$\Im \left[\frac{z^2-\alpha}{(z-z_+)(z-z_-)}\right] > 0$
for $ z \in \mathbb C^+$. 
\item[\rm (c)] If $\alpha > \frac{1}{9}$ then
\[ \frac{(z-z_+)(z-z_-)}{(z- \sqrt{\alpha})^2} \in (0,\infty) \]
if and only if $z \neq \sqrt{\alpha}$ and $z \in \left(\gamma_0 \setminus \Sigma_0 \right) \cup \mathbb R$.
\end{enumerate}
\end{lemma}
\begin{proof}
(a) We consider the case $0 < \alpha < 1$. Observe that 
$z^2 Q_{\alpha}(z)$ tends
to $1$ as $z \to \infty$, and there are no sign changes on the
real line. Thus $z^2 Q_{\alpha}(z) \geq 0$ for real values of $z$,
with double poles at $z=-1$ and $z=-\alpha$, and a local minimum 
at $z= \sqrt{\alpha}$. There is a  minimum at 
$z=-\sqrt{\alpha}$ in case $\alpha \geq \frac{1}{9}$,
and a local maximum at $z=-\sqrt{\alpha}$ in case $\alpha < \frac{1}{9}$. In the latter case there are local minima at
$z= z_{\pm}$. It can be verified that
\[ 0 \leq \alpha Q_{\alpha}(-\sqrt{\alpha}) <
	\alpha Q_{\alpha}(\sqrt{\alpha}) < 1. \]

From an inspection of the graph, it follows that for any $x > 
\alpha Q_{\alpha}(\sqrt{\alpha})$, $x \neq 1$, there
are four real solutions to the equation
 \begin{equation} \label{eq:z2Qzisx}  
z^2 Q_{\alpha}(z) = x. 
 \end{equation}
For $x = 1$ there are three real solutions and a solution at infinity,
while for $\alpha Q_{\alpha}(-\sqrt{\alpha})< x < \alpha Q_{\alpha}(\sqrt{\alpha})$  there are two real solutions. 
If $\alpha \leq \frac{1}{9}$, there are again four real solutions (counting multiplicities) 
for each $0 \leq x \leq \alpha Q_{\alpha}(-\sqrt{\alpha})$.

To summarize, \eqref{eq:z2Qzisx} with $x \geq 0$ admits four solutions
in $\mathbb R \cup \{\infty\}$ except in the following cases:
\begin{equation} \label{only two zeros}
	\begin{cases} 0 \leq x < \alpha Q_{\alpha}(\sqrt{\alpha}), 
	& \text{ and } \frac{1}{9}<\alpha < 1, \\
	\alpha Q_{\alpha}(-\sqrt{\alpha}) < x < \alpha Q_{\alpha}(\sqrt{\alpha}), & \text{ and } 0<\alpha  \leq \frac{1}{9}. \end{cases}
\end{equation}
and in the cases \eqref{only two zeros} there are only two
real solutions.

On the other hand, the calculations \eqref{eq:zQzhigh} and \eqref{eq:zQzlow}
in the proof of Lemma \ref{lem:trajectory} tell us that
$z^2 Q_{\alpha}(z)$ is also real and positive for $z \in \Sigma_0$. For $\frac{1}{9}\leq \alpha <1$, the function decreases from $\alpha Q_{\alpha}(\sqrt{\alpha})$ to $0$
if $z$ moves over $\Sigma_0$ from $\sqrt{\alpha}$ to either $z_+$
or $z_-$. Similarly, for $0 < \alpha \leq \frac{1}{9}$, the function decreases from $\alpha Q_{\alpha}(\sqrt{\alpha})$ to $\alpha Q_{\alpha}(-\sqrt{\alpha})$
if $z$ moves over $\Sigma_0$ from $\sqrt{\alpha}$ to $-\sqrt{\alpha}$ in  either the lower or upper half plane. It means that the equation \eqref{eq:z2Qzisx}
has two additional solutions on $\Sigma_0$ precisely for the cases specified in \eqref{only two zeros}.

Since \eqref{eq:z2Qzisx} is a polynomial equation  of
degree four (if we multiply it through by the denominator)
if $x \neq 1$ and of degree three if $x =1$,
there are four solutions for every $x$, where we include
the solution $\infty$ in case $x=1$.
 For $x \geq 0$ we
found four solutions in $\Sigma_0 \cup \left(\mathbb R \setminus \{-1,-\alpha\} \right) \cup \{\infty\}$, 
and thus
there are no other solutions in the complex plane.
This proves part (a) for $0<\alpha <1$. The proof for $\alpha = 1$ is similar and easier, and we omit it.

\medskip
(b) For $0 < \alpha < \frac{1}{9}$ we have inequalities 
$z_- < - \sqrt{\alpha} < z_+ < \sqrt{\alpha}$
 between the zeros and the poles and therefore
 \begin{equation} \label{eq:Qpflow}
 	 \frac{(z-z_+)(z-z_-)}{z^2-\alpha}
 	= 1 + \frac{A}{z+\sqrt{\alpha}} + \frac{B}{z-\sqrt{\alpha}} 
 	\end{equation}
 with $A, B > 0$. Then $\Im \frac{(z-z_+)(z-z_-)}{z^2-\alpha}
 	< 0$ for $\Im z >0 $.
 In case $\alpha = \frac{1}{9}$ we have
 \eqref{eq:Qpflow} with $A = 0$ and $B>0$ and again 
 $\Im \frac{(z-z_+)(z-z_-)}{z^2-\alpha}
 	< 0$ for $\Im z >0$. This gives (b). 

\medskip
(c) If $z = \sqrt{\alpha} e^{it}$ then (where we recall $z_{\pm}
	= \sqrt{\alpha} e^{i \theta_{\alpha}}$)
\begin{align} \label{eq:Qpfhigh}
	\frac{(z-z_+)(z-z_-)}{(z-\sqrt{\alpha})^2}
		& = \frac{(e^{it}- e^{i\theta_{\alpha}})(e^{it}-e^{-i\theta_{\alpha}})}{(e^{it}-1)^2} \\ \nonumber
		& =  \frac{\cos \theta_{\alpha} - \cos t}{1-\cos t},
\end{align}
	which is in $(0,\frac{1+\cos \theta_{\alpha}}{2}]$ for $ \theta_{\alpha} < |t| \leq \pi$.
The rational function in the left-hand side of \eqref{eq:Qpfhigh} 
is also real and positive for real $z$, $z \neq \sqrt{\alpha}$, and admits a minimum at $z = -\sqrt{\alpha}$. Then, with an argument similar to the one we used to prove part (a),
we check that these are the only  $z$ for which \eqref{eq:Qpfhigh}
is in $(0,\infty)$.   This proves part (c).
\end{proof}

\begin{lemma} \label{lem:LogsIncSigmaMinus1} 
If $z$ moves along $(\Sigma_{-1} \cup \Sigma_0) \cap \mathbb C^+$  
from left to right, then 
\begin{enumerate}
	\item[\rm (a)] $z \to \log |z|$ is strictly decreasing
		on $\Sigma_{-1} \cap \mathbb C^+$ and constant on
		$\Sigma_0 \cap \mathbb C^+$,
	\item[\rm (b)] $z \to \log \left|\frac{(z+1)(z+\alpha)}{z}\right|$ is stricly increasing.
	\end{enumerate}
\end{lemma} 
\begin{proof}
(a) It is clear that $\log |z|$ is constant on the circle $\Sigma_0$.

Let $z=z(t)$, $t \in [0,1]$, be a smooth parametrization of 
$\Sigma_{-1} \cap 
\overline{\mathbb C^+}$ such that $z(0) = x_1$ and $z(1) = x_2$
(in the low temperature case) or $z(1) = z_+(\alpha)$ (in the high temperature
case). Since $\Sigma_{-1}$ is a trajectory of the quadratic differential, $z'(t) Q_{\alpha}(z(t))^{1/2}$ is purely imaginary, and
with our choice of square root, and parametrization of $\Sigma_{-1}$,
we have
\begin{equation} \label{eq:zprimeonSigma-1} 
	z'(t) Q_{\alpha}(z(t))^{1/2} = - i \psi(t),  \qquad
	\text{with } \psi(t) > 0. \end{equation}
	Then with $z = z(t)$, $0 < t < 1$,
	\begin{align} \nonumber
	\frac{d}{dt} \log |z(t)|
	& = \frac{d}{dt} \Re \log (z(t)) 
	 = \Re \left[ \frac{z'(t)}{z(t)} \right] \\
	& = \Re \left[ \frac{- i\psi(t)}{z Q_{\alpha}(z)^{1/2}} \right] 
	=  \psi(t) \Im \left[ \frac{1}{z Q_{\alpha}(z)^{1/2}} \right].
	\label{eq:logzonSigma} 
	\end{align}

By part (a) of Lemma \ref{lem:zQreal},
 $z Q_{\alpha}(z)^{1/2} \not\in \mathbb R$ for $z \in \mathbb C^+ \setminus \Sigma_0$, and by our choice of square root we have
 $\Im \left[ z Q_{\alpha}(z)^{1/2} \right] > 0$
 for $z \in \mathbb C^+ \setminus \Sigma_0$ (this can be seen from example from an expansion of $z Q_{\alpha}(z)^{1/2}$ as $z \to i \infty$), and 
 in particular for $z \in \Sigma_{-1} \cap \mathbb C^+$.
 Then $\Im \left[ \frac{1}{z Q_{\alpha}(z)^{1/2}} \right] < 0$,
 and we find from \eqref{eq:logzonSigma} with $\psi(t) > 0$
 that
 $\frac{d}{dt} \log |z(t)| < 0$. This proves part (a).
 
 \medskip

(b)
Let $z(t)$, $t \in [0,1]$ be a smooth parametrization of
$\Sigma_{-1} \cap \mathbb C^+$ as in the proof of part (a).
Let $\psi(t) > 0$ be as in \eqref{eq:zprimeonSigma-1}.
Then  with $z=z(t)$,
\begin{align} \nonumber
\frac{d}{dt}\log \left|\frac{(z(t)+1)(z(t)+\alpha)}{z(t)}\right| 
	  &=\Re \left[  \left(\frac{1}{z+1} + \frac{1}{z+ \alpha}-\frac{1}{z} 	\right) z'(t)  \right] \\ \label{eq:logonSigma2}
	& \hspace{-1.1cm} = \psi(t) \Im \left[ \left(\frac{z^2-\alpha}{z(z+1)(z+ \alpha)}\right)
	\frac{1}{Q_{\alpha}(z)^{1/2}} \right].
	  \end{align}
	  
If $0 < \alpha \leq \frac{1}{9}$, then
\[ \left(\frac{z^2-\alpha}{z(z+1)(z+ \alpha)}\right)
	\frac{1}{Q_{\alpha}(z)^{1/2}}
	=  \frac{z^2-\alpha}{(z-z_+)(z-z_-)} \]
and this has positive imaginary part for 
$z \in \Sigma_{-1} \cap \mathbb C^+$ 
by part (b) of Lemma~\ref{lem:zQreal}.

If $\frac{1}{9} < \alpha \leq 1$ then
\[ \left(\frac{z^2-\alpha}{z(z+1)(z+ \alpha)}\right)
	\frac{1}{Q_{\alpha}(z)^{1/2}}
	=  \frac{z- \sqrt{\alpha}}{((z-z_+)(z-z_-))^{1/2}}. \]
By part (c) of Lemma \ref{lem:zQreal}, this cannot be real
for $z \in \mathbb C^+ \setminus \{\sqrt{\alpha} e^{it} \mid	
\theta_{\alpha} \leq |t| \leq \pi \}$, since otherwise its
square would be $>0$ and that would contradict the
statement of the lemma. It follows that the sign of
its imaginary part is piecewise 
constant on $\mathbb C^+ \setminus \gamma_{0}$ (recall that $Q_{\alpha}(z)^{1/2}$ is discontinuous along $\Sigma_{0}$). It is
in fact $> 0$ on the outer component, and this
includes $(\Sigma_{-1} \setminus \{z_+\}) \cap \mathbb C^+$.

Thus in both cases we find that \eqref{eq:logonSigma2}
is positive for $0 < t < 1$, and therefore $z \mapsto 
\log \left|\frac{(z+1)(z+\alpha)}{z}\right|$ increases
along $\Sigma_{-1} \cap \mathbb C^+$ as claimed in part (b). 

The increase along $\Sigma_0 \cap \mathbb C^+$ is immediate, since
both $z \mapsto |z+1|$ and $z \mapsto |z+\alpha|$ are strictly 
increasing if $z$ moves
along the circle $\Sigma_0$ from $-\sqrt{\alpha}$ to $\sqrt{\alpha}$,
while $z \mapsto |z|$ is constant.
\end{proof}

\begin{corollary} \label{cor:RePhi} Suppose $\eta \leq \frac{\xi}{2} < 0$. Then
	$z \mapsto \Re \Phi_{\alpha}(z)$ is strictly decreasing
	as $z$ traverses $(\Sigma_{-1} \cup \Sigma_0) \cap \mathbb C^+$
	from left to right.
\end{corollary}
\begin{proof}
	Indeed, from the definition \eqref{Phidef}
	and the fact that $\Re \phi = 0$ on $\Sigma_{-1}$ and $\Sigma_0$,
	we obtain for $z \in \Sigma_{-1} \cup \Sigma_0$,
	\begin{align} \Re \Phi_{\alpha}(z) & = \frac{\xi}{2} \log |(z+1)(z+\alpha)|
	- \eta \log |z| \\
	& = \frac{\xi}{2} \log \left|\frac{(z+1)(z+\alpha)}{z} \right|	
	+ \left( \frac{\xi}{2} - \eta \right) \log|z|,
	\label{eq:RePhionSigma}
	\end{align}
	and by Lemma \ref{lem:LogsIncSigmaMinus1} the sum at the right-hand-side of \eqref{eq:RePhionSigma} 
	is strictly decreasing since 
	$\xi < 0$ and $\frac{\xi}{2} - \eta \geq  0$.
\end{proof}

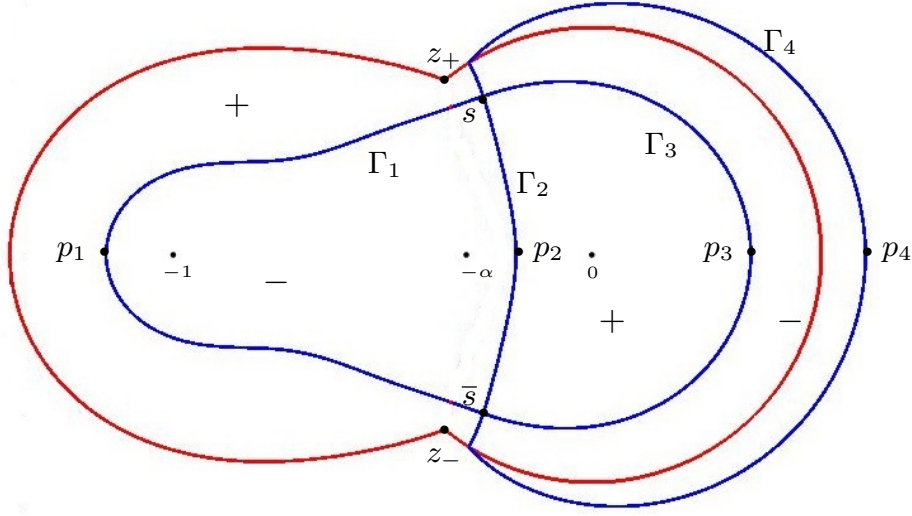
\begin{figure}[t]
	\begin{center}
		\begin{tikzpicture}[master,scale = 1.3,every node/.style={scale=1.3}]
		\node at (0,0) {\includegraphics[width=9cm]{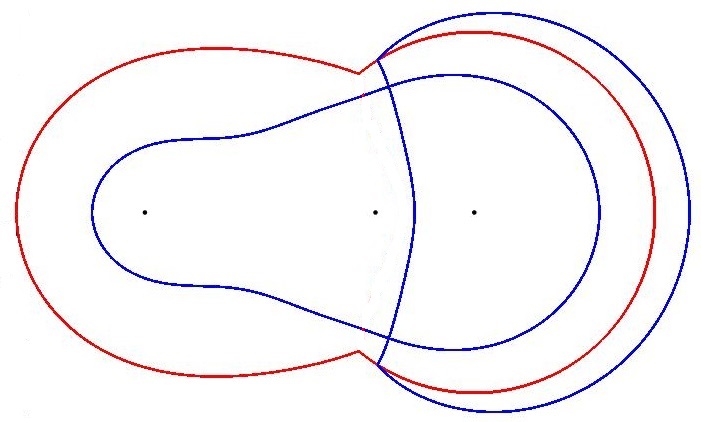}};
		\node at (1.8, -0.7) {$+$};
		\node at (3.6, -0.7) {$-$};
		\node at (-2,1.5) {$+$};
		\node at (-1.6,-0.3) {$-$};
		\node at (-0.5,0.9) {\small $\Gamma_1$};
		\node at (1.0,0.7) {\small $\Gamma_2$};
		\node at (2.3,1.1) {\small $\Gamma_3$};
		\node at (3.5,2.15) {\small $\Gamma_4$};
		%	\node at (-2.6,0.3) {\small $\Gamma_5$};
		%	\node at (-1.1,1.4) {\small $\Sigma_{-1}$};
		%	\node at (3.65,1.15) {\small $\Sigma_0$};
		
		%	\node at (-4.15,-0)  {\tiny $\bullet x_1$};
		\node at (-3.55,-0) {\small $p_1$ \tiny $\bullet$};
		\node at (-2.6,-0.2) {\tiny $-1$};
		%\node at (-0.5,-0)  {\tiny $-\sqrt{\alpha} \bullet$};
		\node at (0.43,-0.2) {\tiny $-\alpha$};
		\node at (1.05,-0) {\tiny $\bullet$ \small $p_2$};
		\node at (1.6,-0.2) {\tiny $0$};
		\node at (3,-0) {\small $p_3$ \tiny $\bullet$};
		%	\node at (4.1,-0)  {\tiny $\bullet \sqrt{\alpha}$};
		\node at (4.58,-0) {\tiny $\bullet$ \small $p_4$};
		
		\node at (0.1,1.98) {\small $z_+$};
		\node at (0.1,-2.1) {\small $z_-$};
		\node at (0.1,1.77) {\tiny $\bullet$};
		\node at (0.1,-1.8) {\tiny $\bullet$};
		
		\node at (0.352,1.4) {\small $s$};
		\node at (0.352,-1.45) {\small $\overline{s}$};
		\node at (0.49,1.56) {\tiny $\bullet$};
		\node at (0.5,-1.635) {\tiny $\bullet$};
		\end{tikzpicture}
	\end{center}
	\caption{\label{fig: level zero curve in L Phi ext high 2} 
		The level set $\mathcal{N}_{\Phi}$ (blue) 
		in the high temperature regime (for $\alpha = 0.3$) in case $\Gamma_1$ intersects
		the real line at $p_1 < -1$. The $+$ and $-$ signs indicate the
		sign of $\Re (\Phi_{\alpha} - \Phi_{\alpha}(s))$.} \end{figure}

Due to Corollary \ref{cor:RePhi}, we see that 
the level set \eqref{eq:NPhi} has at most one point of 
intersection with $(\Sigma_{-1} \cup \Sigma_0) \cap \mathbb C^+$,
because $\Re \Phi_{\alpha}$ is strictly decreasing  there.
Therefore at least three of the $\Gamma_j$'s, say $\Gamma_1, \Gamma_2,
\Gamma_3$, do not intersect
$(\Sigma_{-1} \cup \Sigma_0) \cap \mathbb C^+$, which means that 
they have to go to the real line inside the domain enclosed
by $\Sigma_{-1} \cup \Sigma_0$ (or inside the disk bounded by $\Sigma_0$
in the low temperature regime), and then by symmetry end at 
$\overline{s}$ inside that domain.
Taking $p_j \in \Gamma_j \cap \mathbb R$ for $j=1,2,3$, we choose 
the ordering of the $\Gamma_j$'s such that  $p_1 < p_2 < p_3$.

The contours $\Gamma_1$, $\Gamma_2$, $\Gamma_3$ enclose two
bounded domains for which $\Re \Phi_{\alpha}$ is constant on the boundaries
and harmonic inside, except at the  singularities $-1$, $-\alpha$, $0$,
where $\Re \Phi_{\alpha}$ is unbounded by \eqref{eq:NPhinearpoles}.
By the maximum principle for harmonic functions, each of
the two domains has to contain at least one of the singularities.
Also $\Re (\Phi_{\alpha} -\Phi_{\alpha}(s))$
has opposite signs on the two bounded domains.
Then again by \eqref{eq:NPhinearpoles} one domain contains $0$ 
and the other domain contains
$-\alpha$, and possibly also $-1$, since at both these
points $\Re \Phi_{\alpha}$ tends to $-\infty$.
Thus 
\[ p_1 < -\alpha < p_2 < 0 < p_3 < \sqrt{\alpha}. \]

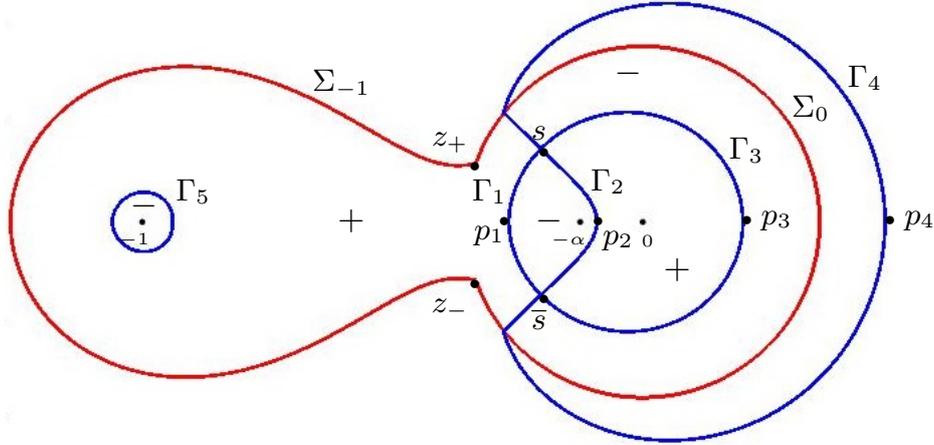
\begin{figure}[t]
	\begin{center}
		\begin{tikzpicture}[master,scale = 1.3,every node/.style={scale=1.3}]
		\node at (0,0) {\includegraphics[width=9cm]{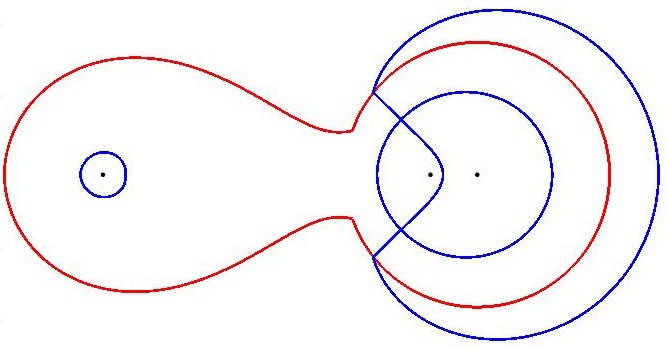}};
		\node at (1.8,1.5) {$-$};
		\node at (1,0) {$-$};
		\node at (2.3,-0.5) {$+$};
		\node at (-1,0) {$+$};
		\node at (-3.1,0.15) {$-$};
		\node at (0.4,0.3) {\small $\Gamma_1$};
		\node at (1.6,0.4) {\small $\Gamma_2$};
		\node at (3,0.72) {\small $\Gamma_3$};
		\node at (4.2,1.46) {\small $\Gamma_4$};
		\node at (-2.6,0.3) {\small $\Gamma_5$};
		\node at (-1.1,1.4) {\small $\Sigma_{-1}$};
		\node at (3.65,1.15) {\small $\Sigma_0$};
		
		\node at (-3.2,-0.2) {\tiny $-1$};
		\node at (0.4,-0.15) {\small $p_1$};
		\node at (0.55,-0) {\tiny $\bullet$};
		
		\node at (1.2,-0.2) {\tiny $-\alpha$};
		\node at (1.5,-0) {\tiny $\bullet$}; 
		\node at (1.7,-0.2) {\small $p_2$};
		\node at (2.0,-0.2) {\tiny $0$};
		\node at (3.2,-0) { \tiny $\bullet$ \small $p_3$};
		\node at (4.65,-0) {\tiny $\bullet$ \small $p_4$};
		
		\node at (0,0.8) {\small $z_+$};
		\node at (0,-0.9) {\small $z_-$};
		\node at (0.25,0.55) {\tiny $\bullet$};
		\node at (0.25,-0.65) {\tiny $\bullet$};
		
		\node at (0.9,0.9) {\small $s$};
		\node at (0.9,-1) {\small $\overline{s}$};
		\node at (0.95,0.7) {\tiny $\bullet$};
		\node at (0.95,-0.8) {\tiny $\bullet$};
		\end{tikzpicture}
	\end{center}
	\caption{\label{fig: level zero curve in L Phi ext high} The level
		set $\mathcal{N}_{\Phi}$ (blue) and the
		contours $\Sigma_{-1}\cup \Sigma_0$ in the high temperature regime (here $\alpha = \frac{1}{8}$)
		in case $-1 < p_1 < -\alpha$. The set $\mathcal{N}_{\Phi}$ divides the plane into five regions and the sign of $\Re (\Phi_{\alpha}-\Phi_{\alpha}(s))$ is indicated in each of these five regions by $+$ or $-$.}
\end{figure}

If $\Gamma_4$ would remain inside $\Sigma_{-1} \cup \Sigma_0$
as well, then it would also go to the real line, say at a point $p_4$,
and continue to $\overline{s}$ inside this domain. 
If $p_3 < p_4 < \sqrt{\alpha}$ then $\Gamma_4$ and $\Gamma_3$ 
would enclose a domain with $\Re \Phi_{\alpha}$ is constant on the boundary, and
harmonic inside, and we would have
a contradiction with the maximum principle. If $p_4 < p_1$ then
$\Gamma_4$ and $\Gamma_1$ enclose a bounded domain
within and we find a contradiction in the same way.

Thus $\Gamma_4$ comes to $(\Sigma_{-1} \cup \Sigma_0) \cap \mathbb C^+$
and continues into the outer domain of $\mathbb C \setminus \mathcal N_{\phi}$.
It cannot go to infinity because of \eqref{eq:NPhinearpoles}
and so it has to go to the real line at a point $p_4$ and
by symmetry it continues in the lower half plane where
it crosses $\Sigma_{-1} \cup \Sigma_0$ again and ends
at $\overline{s}$. 

As $\Re \Phi$ decreases along $(\Sigma_{-1}\cup \Sigma_0)
\cap \mathbb C^+$ from left to right, we find $\Re \Phi_{\alpha}(\sqrt{\alpha}) < \Re \Phi_{\alpha}(s)$.
Since $\Re \Phi_{\alpha}(z) \to +\infty$ as $z \to \infty$, the level
set $\mathcal N_{\Phi}$ intersects the real line at a point
$> \sqrt{\alpha}$. This can only be at $p_4$. 
Thus $\Gamma_4$ and $\Gamma_3$ enclose a domain where
$\Re \Phi_{\alpha} < \Re \Phi_{\alpha}(s)$ and that contains (part of) 
$\Sigma_0$ where $\Phi_{\alpha}$ has
its branch cut, and where $\Re \Phi_{\alpha}$ is not harmonic. Hence
there is no contradiction with the maximum principle.

To summarize, we have a situation as in   
Figure \ref{fig: level zero curve in L Phi ext high 2}
in case $p_1 < -1$, or as
in Figure \ref{fig: level zero curve in L Phi ext high}
if $-1 < p_1 < -\alpha$. In the latter case, there is 
also a separate part $\Gamma_5$ of $\mathcal N_{\Phi}$
that goes around $-1$.

Figures \ref{fig: level zero curve in L Phi ext high 2} and
\ref{fig: level zero curve in L Phi ext high}
are for the high temperature regime. In the low temperature
regime we have that $\Sigma_0$ is the full circle of radius $\sqrt{\alpha}$.
Then in the above discussion we can replace $\Sigma_{-1} \cup \Sigma_0$
by $\Sigma_0$. It follows that $\Gamma_1$, $\Gamma_2$, $\Gamma_3$
stay inside the disk of radius $\sqrt{\alpha}$, and so
$\Gamma_1$ does not go around $-1$. There is always a
part $\Gamma_5$ going around $-1$ in the low temperature
regime, as shown in Figure \ref{fig: level zero curve in L Phi ext low}.

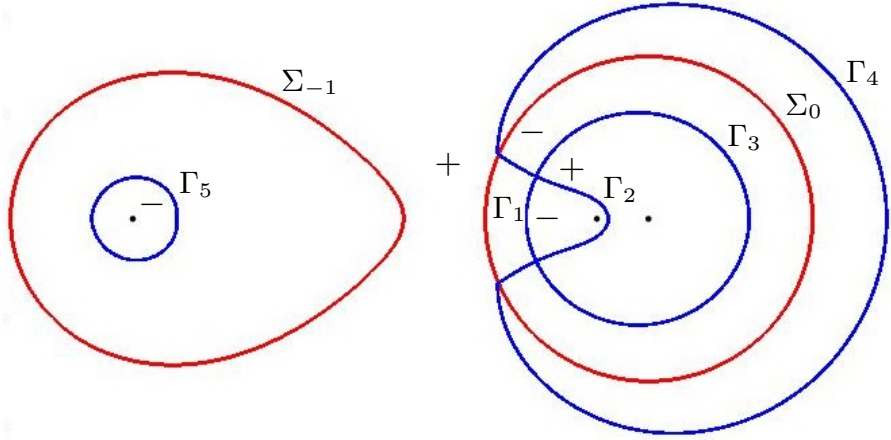
\begin{figure}[t]
	\begin{center}
		\begin{tikzpicture}[master,scale = 1.3,every node/.style={scale=1.3}]
		\node at (0,0) {\includegraphics[width=9cm]{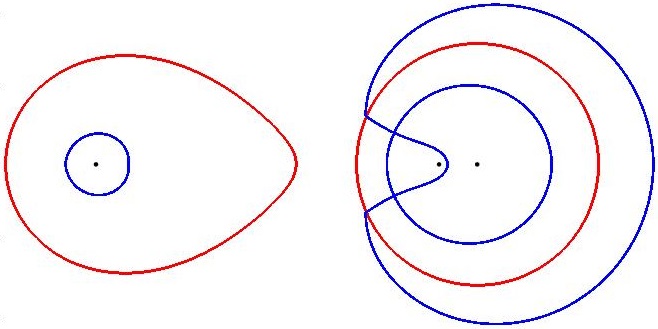}};
		\node at (0.85,0.88) {$-$};
		\node at (1,0) {$-$};
		\node at (1.25,0.5) {$+$};
		\node at (0,0.55) {$+$};
		\node at (-3,0.15) {$-$};
		\node at (.62,0.1) {\small $\Gamma_1$};
		\node at (1.74,0.3) {\small $\Gamma_2$};
		\node at (3,0.83) {\small $\Gamma_3$};
		\node at (4.2,1.49) {\small $\Gamma_4$};
		\node at (-2.55,0.35) {\small $\Gamma_5$};
		\node at (-1.4,1.38) {\small $\Sigma_{-1}$};
		\node at (3.6,1.15) {\small $\Sigma_0$};
		\end{tikzpicture}
	\end{center}
	\caption{\label{fig: level zero curve in L Phi ext low} 
		The level 	set $\mathcal{N}_{\Phi}$ (blue) and the
		contours $\Sigma_{-1}$ and $\Sigma_0$ in the low temperature regime (here $\alpha = \frac{1}{10}$).	The set $\mathcal{N}_{\Phi_{\alpha}}$ divides the plane into five regions and the sign of $\Re (\Phi_{\alpha}-\Phi_{\alpha}(s))$ is indicated in each of these five regions by $+$ or $-$.} 
\end{figure}

It is now clear that we can find contours as described next.
See also Figures
\ref{fig:gammawzlow} and \ref{fig:gammawzhigh} below.

\begin{corollary} \label{cor:contoursexist}
	Let $(\xi,\eta) \in \mathcal L_{\alpha}$ 	with $\eta \leq \frac{\xi}{2} < 0$. 
	\begin{enumerate}
		\item[\rm (a)] In the low temperature regime
		there are closed contours $\gamma_z$ and $\gamma_{w,in}$,
		$\gamma_{w,out}$ such that
		\begin{itemize}
			\item $\gamma_{w,out}$ lies  outside the circle $\gamma_0$,
			does not go around $-1$, and is such that
			\[ \Re \Phi_{\alpha}(w) > \Re \Phi_{\alpha}(s), \qquad w \in \gamma_{w,out}. \]
			\item $\gamma_{w,in}$ lies inside the circle $\gamma_0$,
			goes around $-\alpha$, and it passes through $s$ and $\overline{s}$
			in such a way that
			\[ \Re \Phi_{\alpha}(w) > \Re \Phi_{\alpha}(s), \qquad w \in \gamma_{w,in} 	\setminus \{s, \overline{s}\}. \]
			\item $\gamma_z$ lies inside the circle $\gamma_0$,
			goes around $0$, and it passes through $s$ and $\overline{s}$
			in such a way that
			\[ \Re \Phi_{\alpha}(z) < \Re \Phi_{\alpha}(s), \qquad z \in \gamma_{z} 	\setminus \{s, \overline{s}\}. \]
		\end{itemize}
		
		\item[\rm (b)]
		In the high temperature regime there exist contours $\gamma_z$ and $\gamma_w$ such that 
		\begin{itemize}
			\item $\gamma_{w}$ lies in the domain bounded
			by $\Sigma_0 \cup \Sigma_{-1}$,
			it goes  around $-1$, and it passes through $s$
			and $\overline{s}$ in such a way that
			\[ \Re \Phi_{\alpha}(w) > \Re \Phi_{\alpha}(s), \qquad w \in \gamma_{w} 	\setminus \{s, \overline{s}\}, \]
			\item $\gamma_z$ lies inside the circle $\gamma_0$,
			goes around $0$, and it passes through $s$ and $\overline{s}$
			in such a way that
			\[ \Re \Phi_{\alpha}(z) < \Re \Phi_{\alpha}(s), \qquad z \in \gamma_{z} 	\setminus \{s, \overline{s}\}. \]
		\end{itemize}
	\end{enumerate}
\end{corollary}
In the low temperature regime we will also use $\gamma_w = \gamma_{w,in} \cup \gamma_{w,out}$.

\section{Analysis of double contour integrals}
\label{sec:contour analysis}

\subsection{Lozenge probabilities}

In the final part of the analysis we are going to deform contours in the 
double contour integral to the ones from Corollary \ref{cor:contoursexist}, which leads
to the proof of Theorem \ref{thm:main}. We start by expressing the 
probabilities for the three types of lozenges as double contour integrals.

We use $F(z;x,y)$ as in \eqref{eq:Fzxy} and for a function
$(w,z) \mapsto H(w,z)$,
\begin{align} 
\mathcal I_N(x,y;H) = & \frac{1}{(2\pi i)^2}
\oint_{\gamma_0} \oint_{\gamma_0} R_N(w,z) \frac{(w+1)^N(w+\alpha)^N}{w^{2N}}
\frac{F(z;x,y)}{F(w;x,y)} H(w,z) dw dz \label{eq:IxyH}
\end{align}
We will use \eqref{eq:IxyH} only for functions $(w,z) \mapsto H(w,z)$ that
are products of a rational function in $w$ and a rational
function in $z$, both with poles at $-1$, $-\alpha$, and $0$ only. In addition,
the integrand in \eqref{eq:IxyH} will have singularities for $w=0$ and $z=0$
only, and the contour $\gamma_0$ can be deformed to an arbitrary closed contour around $0$,
and we can take different contours for the two integrals.

\begin{theorem} \label{thm:doubleintegrals_for_lozenge_densities}
	The following statements hold:
	\begin{align}\label{double contour formula for P1}
	\mathbb P\left(\tikz[scale=.3,baseline=(current bounding box.center)] {\draw (0,-1) \lozr; \filldraw (0,-1) circle(5pt); \draw (0,-1) node[below] {$(x,y)$}} \right) & = 
	\begin{cases} \mathcal I_{N}(x,y; H_{1,even}), & \text{ if $x$ is even}, \\[5pt]
	\mathcal I_N(x,y;H_{1,odd}), & \text{ if $x$ is odd}, \end{cases} \\
	\label{double contour formula for P2}
	\mathbb P\left(\tikz[scale=.3,baseline=(current bounding box.center)] {\draw (0,-1) \lozu; \filldraw (0,-1) circle(5pt); \draw (0,-1) node[below] {$(x,y)$}} \right) & = 
	\begin{cases} \mathcal I_N(x,y; H_{2,even}), & \text{ if $x$ is even}, \\[5pt]
	\mathcal I_{N}(x,y;H_{2,odd}), & \text{ if $x$ is odd}, \end{cases} \\
	\label{double contour formula for P3}
	\mathbb P\left(\tikz[scale=.3,baseline=(current bounding box.center)] {\draw (0,0) \lozd; \filldraw (1,0) circle(5pt); \draw (1,0) node[below] {$(x,y)$}} \right) 
	& = 1- \mathcal I_N(x,y; H_3) 
	\end{align}
	with $\mathcal I_N(x,y;H)$ as in \eqref{eq:IxyH}, and
	\begin{equation} \label{eq:Hfunctions}
	\begin{aligned} 
	H_{1,even}(w,z) & = \frac{w}{z(w+\alpha)}, \quad & H_{1,odd}(w,z) & = \frac{w}{z(w+1)}, \\
	H_{2,even}(w,z) & = \frac{\alpha}{z(w+\alpha)}, & H_{2,odd}(w,z) & = \frac{1}{z(w+1)}, \\
	H_3(w,z) & = \frac{1}{z}.
	\end{aligned}
	\end{equation}		 
\end{theorem}
The formula \eqref{double contour formula for P3} is immediate from the
formula \eqref{eq:kernel} for the correlation kernel, since $K_N(x,y,x,y)$
is the probability to have a path at $(x,y + \tfrac{1}{2})$ which is
the same as the probability to have either a type I or type II lozenge at the
location $(x,y)$. Hence $1- K_N(x,y,x,y)$ is the probability to have a type III
lozenge at location $(x,y)$ which is \eqref{eq:IxyH} with $H(w,z) = H_3(w,z) =  \frac{1}{z}$. 
The point of Theorem \ref{thm:doubleintegrals_for_lozenge_densities} is
that there exist similar double contour integrals for the other two
probabilities.

\medskip

The proof of Theorem \ref{thm:doubleintegrals_for_lozenge_densities} 
relies on two lemmas.  
We start by defining the height function $h:\{0,\ldots,2N\}\times \mathbb Z \to \mathbb N$ in terms of the paths $\pi_j:\{0,1,\ldots,2N\} \to \mathbb{Z}+\tfrac{1}{2}$,
for $j=1,\ldots,2N$, by
$$h(x,y)= \#\{ j \mid  \pi_j(x)< y \}.$$
The graph of $h$ is a stepped surface and the paths can be thought of as level curves of this random surface.  We can recover the tiling from the height function by using simple identities which relate the positions of the different lozenges to differences of the height function. 
\begin{lemma} \label{lem:height_to_lozenge} 
	The following identities hold:
	\begin{align*}
	h(x,y+1)-h(x+1,y+1) &= 
	\begin{cases}
	1, & \text{ there is a lozenge 	\tikz[scale=.3,baseline=(current bounding box.center)] { \draw (0,0) \lozr; \filldraw circle(5pt);  \node[below] (i) at (0,-.2) {\tiny{$(x,y)$}};} }\\[-10pt]
	0, & \text{ otherwise.}
	\end{cases} \\
	h(x+1,y+1)-h(x,y) &= 
	\begin{cases}
	1, & \text{ there is a lozenge 	\tikz[scale=.3,baseline=(current bounding box.center)] { \draw (0,0) \lozu; \filldraw circle(5pt);  \node[below] (i) at (0,-.2) {\tiny{$(x,y)$}};}  }\\[-5pt]
	0, & \text{ otherwise.}
	\end{cases} \\
	h(x,y+1)-h(x,y) & = 
	\begin{cases}
	0, & \text{ there is a lozenge 	\tikz[scale=.3,baseline=(current bounding box.center)] { \draw (0,0) \lozd; \filldraw (1,0) circle(5pt);  \node[below] (i) at (1,-.2) {\tiny{$(x,y)$}};} }\\[-5pt]
	1, & \text{ otherwise.}
	\end{cases}
	\end{align*}
\end{lemma}
\begin{proof}
	The proof is straightforward. 
\end{proof}
The next step is a double integral formula for the expectation value of the height function. 

\begin{lemma}\label{lem:doubleintegralheight}
	For $(x, y) \in \{0,1,\ldots, 2N\} \times \mathbb Z$,
	\begin{equation*}
	\mathbb E[h(x,y)] = \sum_{k < y} K_N(x,k,x,k)	=  \frac{1}{(2\pi i)^2} \oint_{\tilde{\gamma}}\oint_\gamma R_N(w,z)  \frac{(w+1)^N(w+\alpha)^N }{w^{2N}}  
	\frac{F(z;x,y)}{F(w;x,y)} \frac{dw dz}{w-z}. 
	\end{equation*} 
	where $\tilde \gamma$ is deformation of $\gamma$ such that $|z|<|w|$ whenever $z \in \tilde{\gamma}$ and $w \in \gamma$. 
\end{lemma}
\begin{proof}
	By the determinantal structure of the correlations (see Proposition~\ref{proposition1.1}) we have
	$$ \mathbb E[h(x,y)]= \sum_{k < y} K_N(x,k,x,k).
	$$
	After inserting the expression \eqref{eq:kernel} for the  kernel, 
	bringing the sum inside the integrals, and evaluating the geometric series 
	$ \ds \frac{1}{z}\sum_{k < y} \frac{w^k}{z^k} = \frac{w^y}{z^y} \frac{1}{w-z}$
	for $|z|<|w|$,   	we obtain the statement. 
\end{proof}
Now we are ready for the proof of Theorem \ref{thm:doubleintegrals_for_lozenge_densities}.

\begin{proof}[Proof of Theorem \ref{thm:doubleintegrals_for_lozenge_densities}]
	Lemma \ref{lem:height_to_lozenge} implies that  
	$$\mathbb P\left( 	\tikz[scale=.3,baseline=(current bounding box.center)] { \draw (0,0) \lozr; \filldraw circle(5pt);  \node[below] (i) at (0,-.2) {\tiny{$(x,y)$}};} \right)=
	\mathbb E[h(x,y+1)]-\mathbb E[h(x+1,y+1)].$$
	We insert the double contour integral formula of Lemma \ref{lem:doubleintegralheight} and combine the two integrals by subtracting the two integrands. Since
	\begin{equation}
	\left(\frac{F(z;x,y+1)}{F(w;x,y+1)} - \frac{F(z;x+1,y+1)}{F(w;x+1,y+1)}
	\right)\frac{1}{w-z} = 
	\frac{F(z;x,y)}{F(w;x,y)} \times \begin{cases} \frac{w}{z(w+\alpha)}, & \text{ if $x$ is even},\\
	\frac{w}{z(w+1)}, & \text{ if $x$ is odd}, 
	\end{cases}
	\end{equation}
	which we can check from \eqref{eq:Fzxy} separately for $x$ even or odd, 
	the formula \eqref{double contour formula for P1} follows.
	Note also that the pole at $z=w$ disappeared when we took the difference, 
	and therefore $\tilde{\gamma}$ can be moved back to $\gamma$ in
	\eqref{double contour formula for P1}. 
	
	The proof of \eqref{double contour formula for P2} is similar,
	and \eqref{double contour formula for P3} is immediate from the 
	structure of the determinantal point process, as already noted after
	the statement of Theorem \ref{thm:doubleintegrals_for_lozenge_densities}.
\end{proof}

\subsection{Symmetries}
We use symmetries in the double integral \eqref{eq:IxyH} to be able to restrict attention
to the lower left part of the hexagon.

\begin{proposition} \label{prop:symmetries} 
	The double integral \eqref{eq:IxyH} has symmetries under the
	mappings $(x,y) \mapsto (2N-x,2N-y)$ and $(x,y) \to (x,N+x-y)$ as follows.
	\begin{enumerate} 
		\item[\rm (a)] We have
		\begin{align} \label{eq:Isymmetry1}
		\mathcal I_N(2N-x,2N-y; H) & = \mathcal I_N(x,y; \widehat{H}), 
		\end{align}
		with
		\begin{equation} \label{eq:hatH} \widehat{H}(w,z) = H(z,w) \times
		\begin{cases} 1, & \text{ if $x$ is even}, \\
		\frac{w+\alpha}{w+1} \frac{z+1}{z+\alpha}, & \text{ if $x$ is odd}.
		\end{cases} \end{equation}
		\item[\rm (b)] We have
		\begin{align} 	\label{eq:Isymmetry2}
		\mathcal I_N(x,N+x-y; H) & = \mathcal I_N(x,y; \widetilde{H})
		\end{align}
		with
		\begin{equation} \label{eq:tildeH}
		\widetilde{H}(w,z) = \frac{\alpha}{wz} H\left(\frac{\alpha}{w},\frac{\alpha}{z}\right) \times
		\begin{cases} 1, & \text{ if $x$ is even}, \\
		\frac{w+\alpha}{w+1} \frac{z+1}{z+\alpha}, & \text{ if $x$ is odd}.
		\end{cases} \end{equation}
	\end{enumerate}	
\end{proposition}

\begin{proof} (a)
	From \eqref{eq:Fzxy} we deduce
	\[ F(z;2N-x,2N-y) = \frac{(z+1)^N (z+\alpha)^N}{z^{2N}} F(z;x,y)^{-1}
	\times \begin{cases} 1 & \text{ if $x$ is even}, \\
	\frac{z+\alpha}{z+1} & \text{ if $x$ is odd}.
	\end{cases} \]
	We insert this in the double integral  \eqref{eq:IxyH} with $(2N-x,2N-y)$ instead of
	$(x,y)$, and we interchange variables $(w,z) \mapsto (z,w)$. Since $R_N(w,z)$ is a symmetric
	expression in the two variables, the identity \eqref{eq:Isymmetry1} with
	$\widehat{H}$ given by \eqref{eq:hatH} follows.	
	
	(b) We now apply the change of variables $w \mapsto \frac{\alpha}{w}$, $z \mapsto \frac{\alpha}{z}$
	to the integral \eqref{eq:IxyH} with $(x,N+x-y)$ instead of $(x,y)$.
	Then $R_N(w,z)$ transforms as in \eqref{eq:RNsymmetry} which we will prove in a separate
	lemma below. The other factors in the integrand of \eqref{eq:IxyH} transform as
	\begin{align*}
	\frac{(w+1)^N (w+\alpha)^N}{w^{2N}} & \mapsto \alpha^{-N} (w+1)^N (w+\alpha)^N \\
	H(w,z) dwdz & \mapsto H \left( \frac{\alpha}{w}, \frac{\alpha}{z} \right)
	\frac{\alpha^2}{w^2 z^2} dwdz \\
	F(z; x, N+x-y) & \mapsto \alpha^{-N - \lfloor \frac{x}{2} \rfloor + y} 
	z^N F(z;x,y) \times \begin{cases} 1, & \text{if $x$ is even} \\
	\frac{z+1}{z+\alpha}, & \text{if $x$ is odd}. \end{cases}
	\end{align*}
	and similarly for $F(w;x,N+x-y)$. Combining all the factors we arrive at
	\eqref{eq:Isymmetry2} with $\widetilde{H}$ as in \eqref{eq:tildeH}. Finally, each transformation reverses the orientation of the respective contour. We change the orientation of each  contour back to the original one at the cost of a minus sign and since we do to this two times the minus signs  cancel against each other. 
\end{proof}

In the proof of part (b) of Proposition \ref{prop:symmetries} we needed an identity for $R_N$
that we prove in a separate lemma. It is related to a symmetry in the Riemann-Hilbert problem \ref{rhp:Y}.
\begin{lemma} 
	\begin{enumerate}
		\item[\rm (a)] Let $\gamma = \gamma_{0}$ be the circle centered at $0$ 
		of radius $\sqrt{\alpha}$. Then the following symmetry holds
		\begin{equation}\label{eq:symmetry} Y(z) = 
		\begin{pmatrix} \alpha^{\frac{N}{2}} & 0 \\ 0 & -\alpha^{-\frac{N}{2}} \end{pmatrix} 
		Y(0)^{-1} Y\left(\frac{\alpha}{z}\right)
		\begin{pmatrix} z^{N} \alpha^{-\frac{N}{2}} & 0 \\ 0 & - z^{-N} \alpha^{\frac{N}{2}}
		\end{pmatrix}.
		\end{equation}
		\item[\rm (b)] The Christoffel-Darboux kernel $R_N$ satisfies 
		\begin{equation} \label{eq:RNsymmetry} 
		R_{N}\left(\frac{\alpha}{w},\frac{\alpha}{z}\right)= 
		\frac{\alpha^{N-1}}{w^{N-1} z^{N-1}} R_{N}(w,z), 
		\qquad w,z \in \mathbb C \setminus \{0\}. \end{equation}
	\end{enumerate}
\end{lemma}

\begin{proof}
	Part (a) follows since the right-hand side of \eqref{eq:symmetry} satisfies
	the conditions of the RH problem \ref{rhp:Y}, as can be check by straightforward
	calculations, and the uniqueness of the solution of the RH problem.
	
	Part (b) follows after inserting \eqref{eq:symmetry} into \eqref{RnzwY},
	again with simple calculations.
\end{proof}

There are corresponding symmetries for the location of the saddle point.

\begin{proposition} \label{prop:saddlesymmetry}
	Let $(\xi,\eta) \in \mathcal L_{\alpha}$. Then
	also $(-\xi,-\eta) \in \mathcal L_{\alpha}$,
	$(\xi,\xi-\eta) \in \mathcal L_{\alpha}$ and
	\begin{align} \label{eq:saddlesymmetry1}
	s(-\xi,-\eta;\alpha) & = s(\xi,\eta;\alpha) \\
	s(\xi,\xi-\eta;\alpha) & =  \label{eq:saddlesymmetry2}
	\alpha \left( \overline{s(\xi,\eta;\alpha)} \right)^{-1}
	\end{align}	
\end{proposition}
\begin{proof}
	From \eqref{Psidef}, we have
	\[ \Psi_{\alpha}(z; \xi,\eta) = - \Phi_{\alpha}(z;-\xi,-\eta)
	\]
	and this implies \eqref{eq:saddlesymmetry1}.
	
	It can be readily verified from \eqref{eq:Qalphahigh} and
	\eqref{eq:Qalphalow} that $\frac{\alpha^2}{z^4} Q_{\alpha} \left(\frac{\alpha}{z} \right)	= 	Q_{\alpha} (z)$.
	Noting that $\phi'(z) = \pm Q_{\alpha}(z)^{1/2}$ by
	\eqref{phidefhigh} and \eqref{phideflow} and keeping track
	of the signs of the square roots, we obtain from this
	\[ - \frac{\alpha}{z^2} \phi'\left(\frac{\alpha}{z} \right)	
	= \phi'(z) \]
	Also, a straightforward computation shows that
	\[ - \frac{\alpha}{z^2} \left[ \frac{\xi}{2} \left( \frac{1}{z+1} + \frac{1}{z+\alpha}\right)
	- \frac{\eta}{z}  \right]_{z \mapsto \frac{\alpha}{z}} 	=  
	\frac{\xi}{2} \left( \frac{1}{z+1} + \frac{1}{z+\alpha}\right)
	- \frac{\xi-\eta}{z}.		
	\]
	From \eqref{Phidef} and \eqref{Psidef} and the last two equalities, we then find
	\[ -\frac{\alpha}{z^2} \Phi_{\alpha}'\left(\frac{\alpha}{z}; \xi,\eta\right) = 	\Phi_{\alpha}'(z; \xi, \xi-\eta) \]
	and similarly for $\Psi_{\alpha}$. This gives \eqref{eq:saddlesymmetry2}, since by definition 
	$s(\xi,\xi-\eta;\alpha)$ is the saddle that is in the
	upper half plane, and therefore the complex conjugation appears
	in \eqref{eq:saddlesymmetry2}.
\end{proof}

\subsection{Preliminaries to the asymptotic analysis}

Theorem \ref{thm:main} will follow from Theorem \ref{thm:doubleintegrals_for_lozenge_densities} 
and the following result.

\begin{proposition} \label{prop:doubleintegrallimit}
	Let $0 < \alpha \leq 1$.
	Suppose $x, y \in \mathbb N$ vary with $N$ such that
	\eqref{eq:scaled_variables} holds with $(\xi,\eta) \in \mathcal L_{\alpha}$.
	Let $(w,z) \mapsto H(w,z)$ satisfy the conditions stated after the definition \eqref{eq:IxyH}.
	Then $\mathcal I_N(x,y; H)$ from  \eqref{eq:IxyH} has the limit
	\begin{equation} \label{eq:Hintegral} 
	\lim_{N \to \infty} \mathcal I_N(x,y;H) = \frac{1}{2\pi i} \int_{\overline{s}}^s H(z,z) dz 
	\end{equation}
	where $s = s(\xi,\eta;\alpha)$ and the integration path from $\overline{s}$ to $s$ in \eqref{eq:Hintegral}
	is in $\mathbb C \setminus (-\infty,0]$.
\end{proposition}

The integrals \eqref{eq:Hintegral} are easy to calculate if $H$ is one of the functions from \eqref{eq:Hfunctions}.
For $H = H_{1,even}$, we obtain for example
\begin{align*} 
\frac{1}{2\pi i} \int_{\overline{s}}^s H_{1,even}(z,z)  dz 
& = \frac{1}{2\pi i} \int_{\overline{s}}^s \frac{dz}{z+\alpha}  = \frac{1}{2\pi i} \left[ \log(s+\alpha) - \log(\overline{s} + \alpha) \right] \\
& = \frac{1}{\pi} \arg (s+\alpha). 
\end{align*} 
Clearly, $\arg(s+\alpha)$ is equal to the angle $\psi_1$
in the triangle $T_{\alpha}$ of Figure \ref{fig:triangles}. 	 
Thus \eqref{prob paths up main thm} with $x$ even follows from
\eqref{double contour formula for P1} and 
Proposition \ref{prop:doubleintegrallimit}. The other
limits in Theorem \ref{thm:main} follow in a similar fashion.
Therefore we have reduced the proof of Theorem \ref{thm:main} to the
proof of Proposition \ref{prop:doubleintegrallimit}.

The symmetries from Proposition \ref{prop:symmetries} allow us 
to restrict our attention to $(\xi, \eta) \in \mathcal L_{\alpha}$ with $\eta \leq \frac{\xi}{2} \leq 0$.

Indeed, suppose that we can prove Proposition \ref{prop:doubleintegrallimit} for certain $(\xi,\eta) \in \mathcal L_{\alpha}$. 
Let $(x,y)$ vary with $N$ such that  \eqref{eq:scaled_variables}
hold but with limits  $(\xi,\xi-\eta) \in \mathcal L_{\alpha}$. 
Suppose $H$ satisfies the
conditions of Proposition \ref{prop:doubleintegrallimit}.
Then by \eqref{eq:Isymmetry2}
\begin{align} \lim_{N \to \infty} 
\mathcal I_N(x,y; H) & = \lim_{N \to \infty} \mathcal I_N(x,N+x-y; \widetilde{H}) \\
& =  \frac{1}{2\pi i} \int_{\overline{s}}^s
\widetilde{H}(z,z) dz, \qquad s = s(\xi,\eta;\alpha),
\end{align}
since  also $\widetilde{H}$
satisfies the conditions of Proposition \ref{prop:doubleintegrallimit},
and by assumption Proposition \ref{prop:doubleintegrallimit} 
holds for $(\xi,\eta)$.
Using \eqref{eq:tildeH} and after changing variables $\frac{\alpha}{z} \mapsto z$, we find 
\begin{align} 
\lim_{N \to \infty} \mathcal I_N(x,y; H) & =
\frac{1}{2\pi i} \int_{\overline{s}}^s
\frac{\alpha}{z^2} H\left(\frac{\alpha}{z}, \frac{\alpha}{z} \right) dz \\
& = \frac{1}{2\pi i} \int_{\alpha s^{-1}}
^{\alpha (\overline{s})^{-1} } H(z,z) dz,
\qquad s = s(\xi,\eta;\alpha). 
\end{align}
We finally use \eqref{eq:saddlesymmetry2} and
we find \eqref{eq:Hintegral} with $s = s(\xi,\xi-\eta;\alpha)$. 
Thus Proposition \ref{prop:doubleintegrallimit} holds for
$(\xi,\xi-\eta)$ if  it holds for $(\xi,\eta)$.

Similarly, but now using \eqref{eq:Isymmetry1}--\eqref{eq:hatH} and  \eqref{eq:saddlesymmetry1}, 
we find that  Proposition \ref{prop:doubleintegrallimit} holds for $(-\xi,-\eta)$ 
if it holds for $(\xi,\eta)$, and by combining the two arguments, 
it also holds for $(-\xi, -\xi+\eta)$.

Thus in order to prove Proposition \ref{prop:doubleintegrallimit}
it suffices to do it for $(\xi,\eta) \in \mathcal L_{\alpha}$
with $\eta \leq \frac{\xi}{2} \leq 0$. We focus on the case
$\eta \leq \frac{\xi}{2} < 0$ and give full arguments there. The case
$\xi = 0$ is special since it means that the saddle $s(\xi,\eta;\alpha)$
is on the branch cut $\Sigma_0$. It can be handled as a limiting case
with the help of additional contour deformations.

\subsection{Contour deformations}

\subsubsection{Contour deformation in the low temperature regime}

We start the analysis of the double integral \eqref{eq:IxyH} with a contour deformation. 
There are several ways to deform the contours, and the ones we are going to present 
will be useful for the lower left part of the liquid region, that is for 
$(\xi,\eta) \in \mathcal L_{\alpha}$ with $\eta \leq \xi/2 < 0$ as in Corollary \ref{cor:contoursexist}.
The deformations will be different for the low and high temperature regimes.

\begin{figure}
	\begin{center} 
		\begin{tikzpicture}[master,scale = 1.3,every node/.style={scale=1.3}]
		\node at (0,0) {\includegraphics[width=9cm]{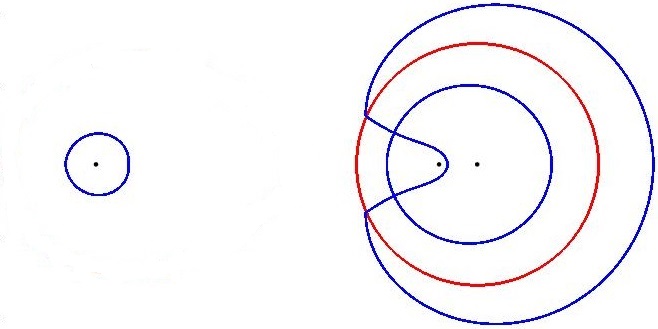}};

		\draw [line width=0.65 mm] (2.35,-0.01) circle(2.5);
		\draw[black,arrows={-Triangle[length=0.3cm,width=0.2cm]}]
		($(2.1,2.5)$) --  ++(-0.0001,0);
		
		\draw[line width=0.65 mm,->-=0.85] (0.9,0.425) to [out=200,in=90] (0.6,0) to [out=-90,in=160] (0.9,-0.425) to [out=-20,in=-90] (2.1,0) to [out=90,in=20] (0.9,0.425);
		
		\draw[line width=0.65 mm,->-=0.85,green] (0.9,0.425) to [out=-70,in=90] (1.1,0) to [out=-90,in=70] (0.9,-0.425) to [out=-110,in=180] (2.1,-1.35) to [out=0,in=-90] (3.3,0) to [out=90,in=0] (2.1,1.35) to [out=180,in=110] (0.9,0.425);
		
		\node at (2.5,1.4) {\footnotesize $\gamma_z$};
		%	\node at (2.3,1.95) {\footnotesize $\gamma_0$};
		\node at (0.2,2) {\footnotesize $\gamma_{w,out}$};
		\node at (2.5,0) {\footnotesize $\gamma_{w,in}$};
		
		%		\node at (-3.2,-0.2) {\tiny $-1$};
		%		\node at (1.3,-0.1) {\tiny $-\alpha$};
		%		\node at (1.9,-0) {\tiny $0$};
		%		\node at (1.05,0.55) {\tiny $s$};
		%		\node at (1.05,-0.6) {\tiny $\overline{s}$};
		\end{tikzpicture}
	\end{center}
	\caption{Contours $\gamma_z$ (green), $\gamma_{w,out}$ (black), and $\gamma_{w,in}$ (black)
		in the low temperature regime. The contours satisfy the
		conditions of Corollary \ref{cor:contoursexist} (a) and Proposition \ref{prop:deformationlow}. \label{fig:gammawzlow}} 
\end{figure}
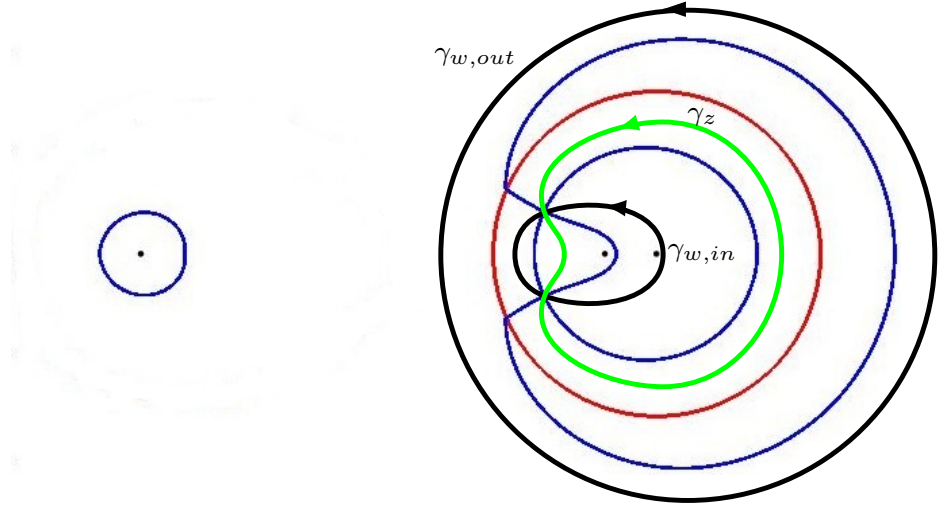

\begin{proposition}
	\label{prop:deformationlow} Let $0 < \alpha \leq \frac{1}{9}$ and
	$(\xi,\eta) \in \mathcal L_{\alpha}$ with $\eta < \frac{\xi}{2} <0$.
	Let $\gamma_z$, $\gamma_{w,in}$ and $\gamma_{w,out}$ 
	be closed contours as in Corollary \ref{cor:contoursexist} (a), 
	(see also Figure \ref{fig:gammawzlow}).
	Then \eqref{eq:IxyH} is equal to
	\begin{multline} \label{eq:deformationlow}
	\mathcal I_N(x,y;H) 	= \frac{1}{2\pi i} \int_{\overline{s}}^s H(z,z) dz + \frac{1}{(2\pi i)^2} \oint_{\gamma_z} dz
	\oint_{\gamma_{w,in}} \frac{dw}{w-z}
	\mathcal R_{N}(w,z) \frac{F(z; x,y)}{F(w;x,y)} H(w,z) \\
	- \frac{1}{(2\pi i)^2} \oint_{\gamma_z} dz
	\oint_{\gamma_{w,out}} \frac{dw}{w-z}
	\mathcal R_{N}(w,z) \frac{F(z; x,y)}{F(w;x,y)} H(w,z)
	\end{multline}
	where $\mathcal R_N$ is given by \eqref{eq:calRN} and
	$F$ is given by \eqref{eq:Fzxy}.
\end{proposition}

\begin{proof}
	In \eqref{eq:IxyH} we use $\gamma_z$ for the
	integral with respect to the $z$ variable, and $\gamma_0$
	(initially) for the $w$ variable. By the conditions in
	Corollary \ref{cor:contoursexist} (a), the contour $\gamma_z$
	lies inside $\gamma_0$.
	
	By Sokhotskii-Plemelj formula and \eqref{eq:calRN} 
	we have for $w \in \gamma_0$, 	
	\[  
	R_N(w,z) \frac{(w+1)^N (w+\alpha)^N}{w^{2N}} (w-z) \\
	= \mathcal R_{N,+}(w,z) - \mathcal R_{N,-}(w,z)  
	\] 
	where the $\pm$ boundary values are with respect to the $w$ variable. 
	This we substitute into the double integral \eqref{eq:IxyH} 
	to obtain 	the difference of two double integrals,
	\begin{multline} \frac{1}{(2\pi i)^2} \oint_{\gamma_z} dz
	\oint_{\gamma_0} \frac{dw}{w-z}
	\mathcal R_{N,+}(w,z) \frac{F(z;x,y)}{F(w;x,y)} H(w,z) \\
	- \frac{1}{(2\pi i)^2} \oint_{\gamma_z} dz
	\oint_{\gamma_0} \frac{dw}{w-z}
	\mathcal R_{N,-}(w,z) \frac{F(z;x,y)}{F(w;x,y)} H(w,z).
	\end{multline}
	
	We deform $\gamma_0$ inwards to $\gamma_{w,in}$ 
	in the first double integral 
	and outwards to $\gamma_{w,out}$ in the second double integral.
	(Recall that $+$-side refers to the interior of $\gamma_0$
	and $-$-side to its exterior.)
	
	We do not encounter any singularites of the integrand if we do
	the deformation into the exterior domain, since by assumption
	$\gamma_{w,out}$ does not go around $-1$. Thus by Cauchy's
	theorem  we obtain the last term in \eqref{eq:deformationlow}.
	
	In the deformation of the first integral we pick up residue
	contributions for those $z \in \gamma_z$ that are in
	the exterior of $\gamma_{w,in}$. This is due to the pole
	at $w=z$ that we encounter when deforming $\gamma_0$ into
	$\gamma_{w,in}$. Since $\mathcal R_N(z,z) = 1$,
	the contribution of the poles leads to the first term
	in \eqref{eq:deformationlow}. The remaining double 
	integral is the second term
	in \eqref{eq:deformationlow}.
\end{proof}

\subsubsection{Contour deformation in the high temperature regime}

In the second proposition (relevant for the high temperature case)
we modify the definition \eqref{eq:calRN}. We use a large circle $\gamma_{\rho}$
centered at the origin  of radius $\rho > 10$ and define  
\begin{equation} \label{eq:calRNhigh} 
\widetilde{\mathcal R}_{N}(w,z) =  \frac{1}{2\pi i} \oint_{\gamma_{\rho}} 
R_N(s,z) \frac{(s+1)^N(s+\alpha)^N}{s^{2N}} \frac{s-z}{s-w} ds. 
\end{equation}
Note that \eqref{eq:calRNhigh} coincides with \eqref{eq:calRN} for $w$
inside $\gamma_0$, and it is the analytic continuation (in the $w$ variable) 
of \eqref{eq:calRN} with $|w|<\alpha$  to the disk $|w| < \rho$. 
Because of \eqref{eq:calRNinT} and the jump \eqref{eq:Tjump2} of $T$, we have
\begin{equation} \label{eq:tildeRNinT}	
\widetilde{\mathcal R}_N(w,z)
= \begin{cases} \begin{pmatrix} 1 & 0 \end{pmatrix} T^{-1}(w) T(z) \begin{pmatrix} 1 \\ 0 \end{pmatrix} e^{N(g(z)-g(w))}, & \quad |w| < \sqrt{\alpha}, \\
\begin{pmatrix} 1 & - e^{2N \phi(z)}  \end{pmatrix} T^{-1}(w) T(z) \begin{pmatrix} 1 \\ 0 \end{pmatrix} e^{N(g(z)-g(w))}, & \quad \sqrt{\alpha} < |w| < \rho, \end{cases}
\end{equation}

\begin{figure}[t]
	\begin{center}
		\begin{tikzpicture}
		[master,scale = 1.3,every node/.style={scale=1.3}]
		\node at (0,0) {\includegraphics[width=9cm]{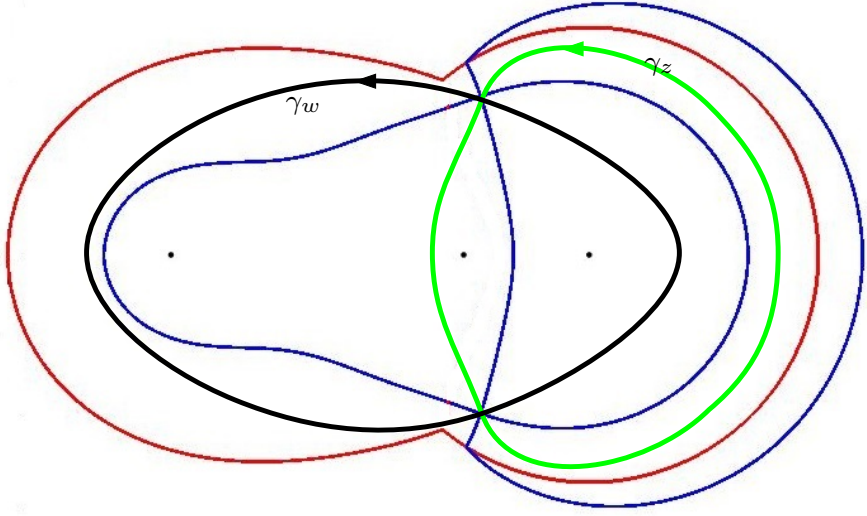}};

		\draw[green, line width=0.65 mm] (0.5,1.56) [out=-110,in=90] to (0,0) [out=-90,in=110] to (0.5,-1.64) [out=-70,in=-135] to (2.8,-1.6) [out=40,in=-90] to (3.5,0) [out=90,in=-45] to (2.8,1.5) [out=135,in=70] to (0.5,1.56);
		\draw[green,arrows={-Triangle[length=0.3cm,width=0.2cm]}]
		($(1.3,2.08)$) --  ++(-0.0001,0);
		\node at (2.3,1.9) {\footnotesize $\gamma_z$};
		\node at (-1.3,1.5) {\footnotesize $\gamma_w$};
		
		\draw [black, line width=0.65 mm] plot [smooth cycle, tension=0.7] coordinates {(0.5,1.56) (2.5,0) (0.5,-1.64) (-2,-1.5) (-3.5,0) (-2,1.5) };
		\draw[black,arrows={-Triangle[length=0.3cm,width=0.2cm]}]
		($(-0.8,1.75)$) --  ++(-0.0001,0);

		%		\node at (-2.6,-0.2) {\tiny $-1$};
		%		\node at (0.43,-0.2) {\tiny $-\alpha$};
		%		\node at (1.6,-0.2) {\tiny $0$};
		%		\node at (0.62,1.37) {\tiny $s$};
		%		\node at (0.64,-1.4) {\tiny $\overline{s}$};
		\end{tikzpicture}
	\end{center}
	\caption{\label{fig:gammawzhigh} 
		The  contours $\gamma_z$ (green) and $\gamma_w$ (black) in the high
		temperature regime. The contours satisfy the conditions 
		of Corollary \ref{cor:contoursexist} (b) and
		Proposition~\ref{prop:deformationhigh}.} 
\end{figure}

\begin{proposition} \label{prop:deformationhigh}
	Let $\frac{1}{9} < \alpha < 1$ and
	$(\xi,\eta) \in \mathcal L_{\alpha}$ with $\eta \leq \frac{\xi}{2} <0$.
	Suppose $\gamma_z$ and $\gamma_{w}$ 
	are closed contours as in Corollary \ref{cor:contoursexist} (b), 
	(see also Figure \ref{fig:gammawzhigh}).
	Let $(x,y)$ be coordinates inside the hexagon.
	Then the double contour integral \eqref{eq:IxyH} is equal to
	\begin{equation} \label{eq:deformationhigh}
	\mathcal I_N(x,y;H) =
	\frac{1}{2\pi i} \int_{\overline{s}}^s H(z,z) dz 	+ \frac{1}{(2\pi i)^2} \oint_{\gamma_z} dz
	\oint_{\gamma_{w}} \frac{dw}{w-z}
	\widetilde{\mathcal R}_{N}(w,z) \frac{F(z; x,y)}{F(w;x,y)} H(w,z),
	\end{equation}
	where $\widetilde{\mathcal R}_N$ is given by \eqref{eq:calRNhigh} and
	$F$ is given by \eqref{eq:Fzxy}.
\end{proposition}
\begin{proof}
	As in the proof of Proposition \ref{prop:deformationlow} we have	
	(but now we use \eqref{eq:calRNhigh})
	\[  
	R_N(w,z) \frac{(w+1)^N (w+\alpha)^N}{w^{2N}} (w-z) \\
	= \widetilde{\mathcal R}_{N,+}(w,z) - 
	\widetilde{\mathcal R}_{N,-}(w,z)  
	\] 
	with $w \in \gamma_\rho$, and the $\pm$ boundary values are
	for $w \in \gamma_\rho$.
	
	We choose $\gamma_\rho$ for the contour in the $w$ integral
	in \eqref{eq:IxyH} and $\gamma_z$ for the $z$ integral.
	Then the double contour integral is a difference of two double integrals
	\begin{multline} \label{eq:integraldifferencehigh} 
	\frac{1}{(2\pi i)^2} \oint_{\gamma_z} dz
	\oint_{\gamma_\rho} \frac{dw}{w-z}
	\widetilde{\mathcal R}_{N,+}(w,z) \frac{F(z;x,y)}{F(w;x,y)} H(w,z) \\
	- \frac{1}{(2\pi i)^2} \oint_{\gamma_z} dz 
	\oint_{\gamma_\rho} \frac{dw}{w-z}
	\widetilde{\mathcal R}_{N,-}(w,z) \frac{F(z;x,y)}{F(w;x,y)} H(w,z)
	\end{multline}
	with $\gamma_z$ inside $\gamma_\rho$.
	
	The integrand in the second double integral has no singularities for $|w| > \rho$, since
	the poles are at $w=z$, $w=-1$, $w=-\alpha$, and they are
	all inside. For $|w| > \rho$ we have $\widetilde{\mathcal R}_N(w,z) = \mathcal{R}(w,z)$.
	From the asymptotic behavior in the RH problem \ref{rhp:Y} for $Y$ we get
	\[ \begin{pmatrix} 1 & 0 \end{pmatrix} Y^{-1}(w)
	= \begin{pmatrix} 1 & 0 \end{pmatrix} \begin{pmatrix}
	w^{-N} & 0 \\ 0 & w^N \end{pmatrix} \left(I + \bigO(w^{-1})\right) = 
	\bigO\left(w^{-N}\right) \]
	as $w \to \infty$, and thus by \eqref{eq:calRN}  
	\[ \widetilde{\mathcal R}_N(w,z) = \bigO\left(w^{-N}\right) \quad \text{ as }
	w \to \infty. \]
	Also 	by the definition  of $F$, see \eqref{eq:Fzxy}, we have 
	$\left(F(w;x_2,y_2) \right)^{-1} = \bigO(w^{y_2 - x_2})$ as $w \to \infty$.
	By combining with \eqref{eq:Hfunctions}, we see that the full integrand  in \eqref{eq:integraldifferencehigh} is 
	therefore $O\left(w^{-N+y_2-x_2 -1}\right)$ as $w \to \infty$.
	Since $(x, y)$ is a point inside the hexagon, we have
	inequalities $-N < y_2 - x_2 < N$. Thus, since we are dealing
	with integers, the integrand is $O(w^{-2})$ as $w \to \infty$.
	Therefore the second double integral in
	\eqref{eq:integraldifferencehigh} vanishes identically.
	
	In the first double integral we deform $\gamma_\rho$ to $\gamma_w$
	as in the statement of the proposition. We pick up a
	residue contribution at the pole $w=z$ for those $z \in \gamma_z$ that
	lie in the exterior of $\gamma_w$. This gives the
	first term in \eqref{eq:deformationhigh}. The remaining
	double integral is the second term in \eqref{eq:deformationhigh}.
\end{proof}

\subsection{Proof of Proposition \ref{prop:doubleintegrallimit}}

We are now ready for the proof of Proposition \ref{prop:doubleintegrallimit}
which, as already noted leads to the proof of Theorem \ref{thm:main}.
We also noted that it suffices to prove the proposition for
$(\xi, \eta) \in \mathcal L_{\alpha}$ with  $\eta \leq \frac{\xi}{2} \leq 0$.

We first assume $\xi < 0$ and later deal with the modifications
that are necessary for $\xi = 0$.

We write $x = x_N = (1 + \xi_N) N$, $y = y_N = (1+\eta_N) N$, and we
are in the situation where
\[ (\xi_N, \eta_N) \to (\xi,\eta) \in \mathcal L_{\alpha} \]
with $\eta \leq \frac{\xi}{2} < 0$. For $N$ large enough,
we then also have $(\xi_N, \eta_N) \in \mathcal L_{\alpha}$
with $\frac{\xi_N}{2} < 0$. We may also assume that $\eta_N \leq \frac{\xi_N}{2} < 0$,
because of symmetries as in Proposition \ref{prop:symmetries} (b) 
and Proposition \ref{prop:saddlesymmetry}.

Then also $\Phi_N(z) := \Phi_{\alpha}(z;\xi_N,\eta_N) $
and the saddle  $s_N := s(\xi_N,\eta_N;\alpha)$ vary with $N$,
but in a controlled way. As $N \to \infty$ they tend
to their limiting values $\Phi_{\alpha}(z;\xi,\eta)$
and $s := s(\xi,\eta;\alpha)$. 

In particular
\begin{equation} \label{eq:HNtendstoH} 
\frac{1}{2\pi i} \int_{\overline{s}_N}^{s_N} H(z,z) dz
\to \frac{1}{2\pi i} \int_{\overline{s}}^{s} H(z,z) dz
\end{equation}
as $N \to \infty$.

\subsubsection{Low temperature regime with $\eta < \frac{\xi}{2} < 0$}
Let $\gamma_{z}^{(N)}$ and $\gamma_{w,in}^{(N)}$, $\gamma_{w,out}^{(N)}$
be contours as in Corollary \ref{cor:contoursexist} (a) and  
Proposition~\ref{prop:deformationlow}  but corresponding to the parameters 
$(\xi_N,\eta_N)$ and $s = s_N$.
Then by \eqref{eq:deformationlow} 
\begin{multline} \label{eq:deformationlow2}
\mathcal I_N(x_N,y_N;H) =	\frac{1}{2\pi i} \int_{\overline{s}_N}^{s_N} H(z,z) dz + \frac{1}{(2\pi i)^2} \oint_{\gamma_z^{(N)}} dz
\oint_{\gamma_{w,in}^{(N)}} \frac{dw}{w-z}
\mathcal R_N(w,z) \frac{F(z;x_N,y_N)}{F(w;x_N,y_N)} H(w,z) \\
- \frac{1}{(2\pi i)^2} \oint_{\gamma_z^{(N)}} dz
\oint_{\gamma_{w,out}^{(N)}} \frac{dw}{w-z}
\mathcal R_N(w,z) \frac{F(z;x_N,y_N)}{F(w;x_N,y_N)} H(w,z) 
\end{multline}
and in view of \eqref{eq:HNtendstoH} it is enough to show that
the two double integrals in \eqref{eq:deformationlow2} tend
to $0$ as $N \to \infty$.

By Corollary \ref{cor:TandTinvsmall} (a) there exists a constant 
$C_1 > 0$ such that
\begin{equation} \label{eq:RNbound} 
\left| \mathcal R_N(w,z) \right| 
\leq C_1 \left| e^{N(g(z) - g(w))} \right|. \end{equation}
Also by definitions \eqref{Phidef} and \eqref{eq:Fzxy}
\[ e^{N g(z)} F(z;x_N,y_N) e^{N \frac{\ell}{2}} =
e^{N \Phi_N(z)} \times
\begin{cases} 1, & \text{ if $x_N$ is even}, \\
\left(\frac{z+\alpha}{z+1} \right)^{1/2}, & 
\text{ if $x_N$ is odd}. \end{cases} \]
The contours stay away from $-\alpha$ and $-1$, therefore the extra factor 
in case $x_N$ is odd remains bounded and  bounded away from $0$.
Combining this with \eqref{eq:RNbound} we obtain for some constant $C_2>0$,
\begin{equation} \label{eq:RNtimesFbound} 
\left| \mathcal R_N(w,z) \frac{F(z;x_N,y_N)}{F(w;x_N,y_N)} \right| 
\leq C_2 \left| e^{N(\Phi_N(z)- \Phi_N(w))} \right|,
\end{equation}
for $w \in \gamma_w^{(N)} := 
\gamma_{w,out}^{(N)} \cup \gamma_{w,in}^{(N)}$, and $z \in \gamma_z^{(N)}$.

By Corollary \ref{cor:contoursexist} (a) the contours are in regions
where $\Re \Phi_N(z) < \Re \Phi_N(s_N) < \Re \Phi_N(w)$,
except for $\{ w, z \} \subset \{s_N, \overline{s}_N\}$, when there is equality. 
We can actually estimate (since the saddles are simple, and locally near
the saddles we can follow steepest/ascent paths)
\begin{equation} \label{eq:PhinearsNbound}
\begin{aligned} 
\Re \left( \Phi_N(w) - \Phi_N(s_N) \right)
\geq  C_3 |w-s_N|^2, & \quad \text{ for } w \in \gamma_w^{(N)} \cap \mathbb C^+, \\
\Re \left( \Phi_N(z) - \Phi_N(s_N) \right)
\leq -  C_3 |z-s_N|^2, & \quad \text{ for } z \in \gamma_z^{(N)} \cap \mathbb C^+,
\end{aligned}
\end{equation} 
with a constant $C_3 > 0$ that is independent of $N$. By symmetry of
the contours in the real axis, there are similar estimates for
$w$ and $z$ in the lower half plane.
Then it follows from \eqref{eq:RNtimesFbound}
that the second double integral in \eqref{eq:deformationlow2} is exponentially small as $N \to \infty$
since $\gamma_{w,out}^{(N)}$ stays away from the saddle $s_N$.

The first double integral in \eqref{eq:deformationlow2} is not exponentially small, 
since the contours intersect at
the saddles $s_N$ and $\overline{s}_N$. The dominant contribution
comes from both $w$ and $z$ close to the saddle points. 
For a small enough $\delta > 0$, we may assume that $\gamma_{w,in}^{(N)}
\cap D_{\delta}(s_N)$ and $\gamma_z^{(N)} \cap D_{\delta}(s_N)$ are straight line segments
that meet at right angles. Then there are parametrizations 
with $-\delta < x < \delta$ and $-\delta < y < \delta$ such
that $|z-s_N| = |x|$, $|w-s_N| = |y|$ and $|w-z| = \sqrt{x^2+y^2}$
for $z, w$ on the contours in the $\delta$-neighborhood of $s_N$.

From  estimates \eqref{eq:RNtimesFbound} and \eqref{eq:PhinearsNbound} we then 
easily get for some $C_4 > 0$,
\begin{multline}
\left|
\frac{1}{(2\pi i)^2} \oint_{\gamma_z^{(N)} \cap D_{\delta}(s_N)} dz
\oint_{\gamma_{w,in}^{(N)} \cap D_{\delta}(s_N)} \frac{dw}{w-z}
\mathcal R_N(w,z) \frac{F(z;x_N,y_N)}{F(w;x_N,y_N)} H(w,z) \right| \\
\leq C_4 \iint_{|x|^{2}+|y|^{2}\leq \delta^{2}}
e^{-2C_3 N (x^2+y^2)} \frac{dxdy}{\sqrt{x^2 + y^2}}  =  2\pi C_4 \int_0^{\delta}  e^{-2C_3 N r^2} dr
\end{multline}
which tends to zero as $N \to \infty$. The same estimates hold
for $w$ and $z$ near $\overline{s}_N$, or for $w$ near $s_N$
and $z$ near $\overline{s}_N$ or vice versa, and it follows that the
first double integral in \eqref{eq:deformationlow2} tends to zero as $N \to \infty$.

Thus both double integrals tend to zero as $N \to \infty$.
Because of \eqref{eq:HNtendstoH} we then conclude that \eqref{eq:Hintegral} holds.

%\begin{figure}
%\begin{center} 
%	\begin{tikzpicture}[master,scale = 1.3,every node/.style={scale=1.3}]
%	\node at (0,0) {\includegraphics[width=9cm]{High_PhiXi_Liquid_m025_x050_thick2}};
%	\draw [green, line width=0.65 mm] plot [smooth cycle, tension=0.7] coordinates {(1,0.2) (0.93,0.65) (0.97,1) (1.4,1.4) (2.2,1.4) (3,0.9) (3.35,0) (3,-0.9) (2.2,-1.4) (1.4,-1.4) (0.95,-1) (0.93,-0.65) (1,-0.2)};
%	\draw[green,arrows={-Triangle[length=0.3cm,width=0.2cm]}]
%	($(1.7,1.46)$) --  ++(-0.0001,0);
%	
%	\draw [line width=0.65 mm] (1.95,-0.01) circle(1.78);
%	%	\draw [line width=0.65 mm] (-3.05,0) circle(0.8);
%	\draw[black,arrows={-Triangle[length=0.3cm,width=0.2cm]}]
%	($(1.9,1.77)$) --  ++(-0.0001,0);
%	%	\draw[black,arrows={-Triangle[length=0.3cm,width=0.2cm]}]
%	%	($(-3.2,0.8)$) --  ++(-0.0001,0);
%	
%	\node at (1.82,1.28) {\footnotesize $\gamma_z$};
%	\node at (2.3,1.95) {\footnotesize $\gamma_0 = \gamma_w$};
%	%	\node at (-2.95,1.03) {\footnotesize $\gamma_*$};
%	
%	\end{tikzpicture}
%\end{center}
%\caption{\label{fig:gammaz high1} 
%The contours $\gamma_z$ (green) and $\gamma_w = \gamma_0$
%(black) for the 
%case where the part $\Gamma_1$ of the level line $\Re \Phi = \Re \Phi(s)$
%(shown in blue) stays inside $\gamma_0$. 
%%
%The figure shows such a case in the high temperature regime,
%but the same choice of contours applies to the low temperature
%regime.}
%\end{figure}

\subsubsection{High temperature regime with $\eta \leq \frac{\xi}{2} < 0$ }

The proof in the high temperature regime is similar.
We again use $N$ dependent contours $\gamma_w^{(N)}$ and $\gamma_z^{(N)}$
satisfying the conditions of Corolarry \ref{cor:contoursexist} (b) and Proposition \ref{prop:deformationhigh}.
Due to \eqref{eq:deformationhigh} and \eqref{eq:HNtendstoH} we have to show that 
\begin{equation} \label{eq:deformationhigh2} 
\frac{1}{(2\pi i)^2} 
\oint_{\gamma_z^{(N)}} dz \oint_{\gamma_w^{(N)}}
\frac{dw}{w-z} \widetilde{\mathcal R}_N(w,z) \frac{F(z;x_N,y_N)}{F(w;x_N,y_N)}
H(w,z) \end{equation}
tends to $0$ as $N \to \infty$.

We recall that $w \mapsto \widetilde{\mathcal R}_N(w,z)$
is the analytic continuation of $w \mapsto \mathcal R_N(w,z)$
from the disk $|w| < \sqrt{\alpha}$ into the large disk $|w| < \rho$.
It then follows from Corollary \ref{cor:TandTinvsmall} (b) and (c) that
\begin{equation} \label{eq:calRNbound} 
\widetilde{\mathcal R}_N(w,z) \leq	
C_1 \left| e^{N(g(z) - g(w))} \right| 
\end{equation}
whenever $w$ is in the domain bounded
by $\Sigma_0 \cup \Sigma_{-1}$ and $z \in \mathbb C$ with $w, z$ bounded
away from the branch points $z_{\pm}$.
This is the estimate that is analogous to \eqref{eq:RNbound} in the
low temperature regime.

By Corollary \ref{cor:contoursexist} (b) the contour $\gamma_w^{(N)}$ is 
inside $\Sigma_0 \cup  \Sigma_{-1}$, and we can apply \eqref{eq:calRNbound}
in the estimation of \eqref{eq:deformationhigh2}. The rest of the
proof is the same as in the low temperature regime with $\xi < 0$.

\subsubsection{Case $\xi=0$ and $\eta < 0$}

\begin{figure}[t]
	\begin{center}
		\hspace{-1cm}\begin{tikzpicture}
		[master,scale = 1.3,every node/.style={scale=1.3}]
		\node at (0,0) {\includegraphics[width=5cm]{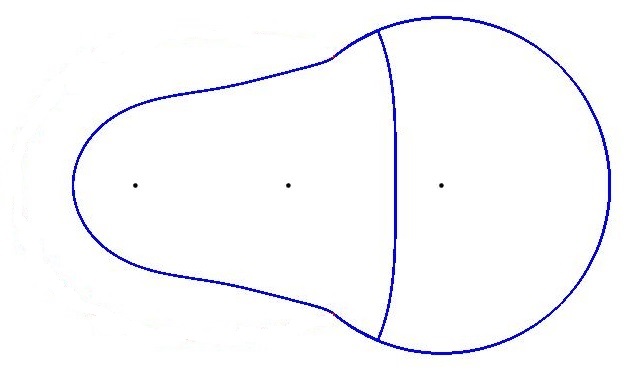}};
		\node at (-1.45,-0.2) {\tiny $-1$};
		\node at (-1.4,-0.02) {\tiny $\bullet$};
		\node at (-0.15,-0.2) {\tiny $-\alpha$};
		\node at (-0.16,-0.02) {\tiny $\bullet$};
		\node at (1,-0.2) {\tiny $0$};
		\node at (1,-0.02) {\tiny $\bullet$};
		\node at (0.47,1.05) {\tiny $s$};
		\node at (0.53,1.22) {\tiny $\bullet$};
		\node at (0.47,-1.07) {\tiny $\overline{s}$};
		\node at (0.53,-1.25) {\tiny $\bullet$};
		\node at (0.17,0.77) {\tiny $z_{+}$};
		\node at (0.165,0.985) {\tiny $\bullet$};
		\node at (0.17,-0.85) {\tiny $z_{-}$};
		\node at (0.165,-1.03) {\tiny $\bullet$};
		
		\node at (-0.8,1.5) {$+$};
		\node at (1.4,0.3) {$+$};
		\node at (-0.8,0.3) {$-$};
		\end{tikzpicture}\begin{tikzpicture}
		[slave,scale = 1.3,every node/.style={scale=1.3}]
		\node at (0,0) {\includegraphics[width=4.5cm]{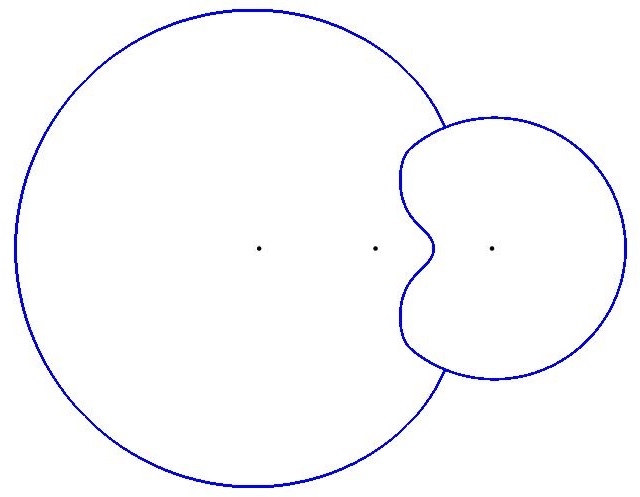}};
		\node at (-0.4,-0.2) {\tiny $-1$};
		\node at (-0.4,-0.01) {\tiny $\bullet$};
		\node at (0.4,-0.2) {\tiny $-\alpha$};
		\node at (0.4,-0.01) {\tiny $\bullet$};
		\node at (1.25,-0.2) {\tiny $0$};
		\node at (1.25,-0.01) {\tiny $\bullet$};
		\node at (0.88,0.7) {\tiny $s$};
		\node at (0.88,0.83) {\tiny $\bullet$};
		\node at (0.88,-0.72) {\tiny $\overline{s}$};
		\node at (0.88,-0.85) {\tiny $\bullet$};
		\node at (0.45,0.635) {\tiny $z_{+}$};
		\node at (0.6,0.65) {\tiny $\bullet$};
		\node at (0.49,-0.76) {\tiny $z_{-}$};
		\node at (0.6,-0.69) {\tiny $\bullet$};
		
		\node at (-0.7,0.3) {$+$};
		\node at (1.5,1.3) {$-$};
		\node at (1.4,0.3) {$-$};
		
		\draw[dashed,black,line width=0.65 mm] (-2.5,-1.8) to [out=90, in=-90] (-2.5,1.8);
		\end{tikzpicture}
	\end{center}
	\caption{\label{fig:xi=0 and eta<0} 
		The  sets $\mathcal{N}_{\Phi}$ (left) and $\mathcal{N}_{\Psi}$ (right) in the high
		temperature regime for $\xi = 0$ and $\eta < 0$. The signs of $\Re (\Phi_{\alpha} - \Phi_{\alpha}(s))$ (left) and $\Re (\Psi_{\alpha} - \Psi_{\alpha}(s))$ (right) are indicated with $\pm$.} 
\end{figure}
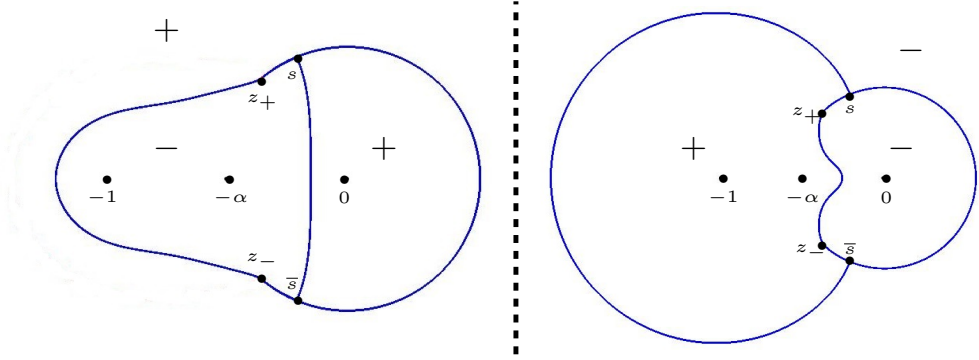

For $\xi = 0$, the saddle is on the branch cut $\Sigma_0$ for the functions
$\phi$ and $\Phi_{\alpha}$. We need additional deformation of contours to handle
this case. For definiteness we focus on the high temperature
regime, but the low temperature regime can be done similarly.

Note that $\Phi_{\alpha}(z) = \phi(z) - \eta \log z$ since $\xi = 0$, see \eqref{Phidef}. Since $\Re \phi(z) = 0$ for $z \in \Sigma_{0}$, and since $s \in \Sigma_{0}$, we have
$\Re \Phi_{\alpha}(s) = - \eta \log \sqrt{\alpha}$, and furthermore the set $\mathcal{N}_{\Phi}$ (defined in \eqref{eq:NPhi}) is such that
\begin{align*}
\Sigma_{0} \subset \mathcal{N}_{\Phi},
\end{align*}
see Figure \ref{fig:xi=0 and eta<0}, left. To deal with this case we also need information about the set $\mathcal{N}_{\Psi} = \{z \in \mathbb{C} | \Re \Psi_{\alpha}(z) = \Psi_{\alpha}(s)\}$, see Figure \ref{fig:xi=0 and eta<0}, right. For $\xi = 0$, we also have $\Sigma_{0} \subset \mathcal{N}_{\Psi}$. 

We treat the case $(0,\eta) \in \mathcal L_{\alpha}$ with $\eta < 0$ as a limit
of the case $(\xi,\eta)$ with $\eta < \frac{\xi}{2} < 0$  that we considered before.
In this limit the contours from 
Corollary \ref{cor:contoursexist} (b) can be chosen in such a way that they
tend to  contours $\gamma_z$ and $\gamma_w$  that partly overlap with $\Sigma_0$,
such that the following hold (see Figure \ref{fig:contours for xi=0 and eta<0} together with Figure \ref{fig:xi=0 and eta<0}, left)
\begin{itemize}
	\item $\gamma_w$ contains the subarcs 
	\[ \gamma_w \cap \Sigma_0 : \quad |w| = \sqrt{\alpha}, \,  \arg s \leq |\arg w| \leq \arg z_+(\alpha) \]
	of $\Sigma_0$ and lies otherwise  inside the (open) domain bounded
	by $\Sigma_0 \cup \Sigma_{-1}$, it goes around $-1$, and 
	\begin{equation} \label{eq:Phiwforxi0}
	\begin{aligned} 
	\Re \Phi_{\alpha}(w) > \Re \Phi_{\alpha}(s), & \quad  w \in \gamma_w \setminus \Sigma_0, \\
	\Re \Phi_{\alpha,+}(w) = \Re \Phi_{\alpha}(s), & \quad  w \in \gamma_w \cap \Sigma_0,
	\end{aligned} \end{equation}
	\item $\gamma_z$ contains the subarc
	\[ \gamma_z \cap \Sigma_0 : \quad |z| = \sqrt{\alpha}, \, -\arg s \leq \arg z \leq \arg s \]
	of $\Sigma_0$ and lies otherwise inside the domain bounded by $\Sigma_0 \cup \Sigma_{-1}$,
	it goes around $0$, and
	\begin{equation} \label{eq:Phizforxi0}
	\begin{aligned} 
	\Re \Phi_{\alpha}(z) < \Re \Phi_{\alpha}(s), & \quad  z \in \gamma_z \setminus \Sigma_0, \\
	\Re \Phi_{\alpha,+}(z) = \Re \Phi_{\alpha}(s), & \quad  z \in \gamma_z \cap \Sigma_0.
	\end{aligned} \end{equation}
\end{itemize}

We want to estimate the double integral in \eqref{eq:deformationhigh} with
$x = x_N = (1+o(1)) N $ and $y = y_N = (1+\eta +o(1)) N$ as $N \to \infty$. To avoid the
use of $N$ dependent contours as in the proofs above (which can be handled but would
obscure the exposition) we assume $x_N = N + O(1)$ and $y_N = (1 + \eta) N + O(1)$ as
$N \to \infty$. 
Then by combining \eqref{eq:Fzxy}, \eqref{Phidef} with
\eqref{eq:tildeRNinT} we find that $\widetilde{R}_N(w,z) \frac{F(z;x_N,y_N)}{F(w;x_N,y_N)}$
(which is the main part of the integrand in \eqref{eq:deformationhigh})
is equal to 
\begin{equation} \label{eq:PhizminPhiwforxi0}
e^{N (\Phi_{\alpha}(z) - \Phi_{\alpha}(w))} \times \begin{cases} 
\begin{pmatrix} 1 & 0 \end{pmatrix} T^{-1}(w) T(z) \begin{pmatrix} 1 \\ 0 \end{pmatrix},
& \quad w \in \gamma_w, |w| < \sqrt{\alpha}, \\
\begin{pmatrix} 1 & -e^{2N \phi(w)} \end{pmatrix}	
T^{-1}(w) T(z) \begin{pmatrix} 1 \\ 0 \end{pmatrix}, &
\quad w \in \gamma_w, |w| > \sqrt{\alpha}
\end{cases}
\end{equation}
times a factor that remains bounded as $N \to \infty$. 
In \eqref{eq:PhizminPhiwforxi0} we take $+$ boundary values for $\Phi_{\alpha}$
and $T$ whenever $w$ and/or $z$ are on $\Sigma_0$.

\begin{figure}[t]
	\begin{center}
		\begin{tikzpicture}
		[master,scale = 1.3,every node/.style={scale=1.3}]
		\node at (0,0) {\includegraphics[width=9cm]{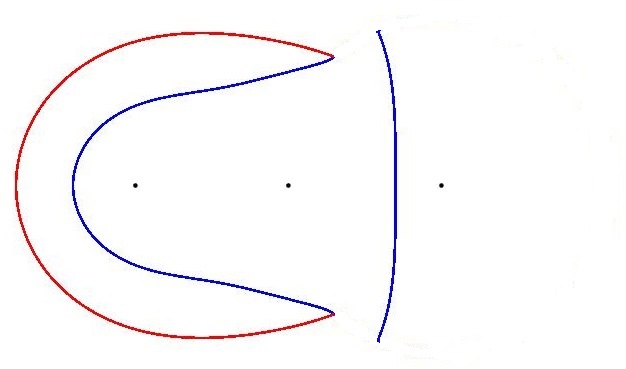}};

		\draw[green, line width=0.75 mm,->-=0.85] (0.92,2.2)
		[out=-110,in=90] to (0,0)
		[out=-90,in=110] to (0.92,-2.25)
		[out=-20,in=-135] to (3.6,-1.7)
		[out=48,in=-90] to (4.27,0)
		[out=90,in=-48] to (3.6,1.7)
		[out=135,in=20] to (0.92,2.2);
		
		\draw[black, line width=0.75 mm,->-=0.85] (0.92,2.2)
		[out=-160,in=40] to (0.28,1.825)
		[out=175,in=90] to (-3.7,0)
		[out=-90,in=180] to (0.28,-1.85)
		[out=-40,in=160] to (0.92,-2.25)
		[out=30,in=-30] to (0.92,2.2);
		
			\node at (-2.55,-0.2) {\tiny $-1$};
		\node at (-2.55,-0.02) {\tiny $\bullet$};
		\node at (-0.35,-0.2) {\tiny $-\alpha$};
		\node at (-0.35,-0.02) {\tiny $\bullet$};
		\node at (1.85,-0.2) {\tiny $0$};
		\node at (1.85,-0.02) {\tiny $\bullet$};
		\node at (0.9,2.35) {\tiny $s$};
		\node at (0.9,2.17) {$\bullet$};
		\node at (0.9,-2.45) {\tiny $\overline{s}$};
		\node at (0.9,-2.25) { $\bullet$};
		\node at (0.25,1.6) {\tiny $z_{+}$};
		\node at (0.25,1.8) { $\bullet$};
		\node at (0.25,-1.65) {\tiny $z_{-}$};
		\node at (0.25,-1.85) {$\bullet$};
		
		\end{tikzpicture}
	\end{center}
	\caption{\label{fig:contours for xi=0 and eta<0} 
		The  contours $\gamma_{z}$ (green) and $\gamma_{w}$ (black) for $\xi = 0$ and $\eta < 0$ in the high temperature regime. They are drawn on top of $\mathcal{N}_{\Phi} \cup \Gamma_{-1}$} 
\end{figure}
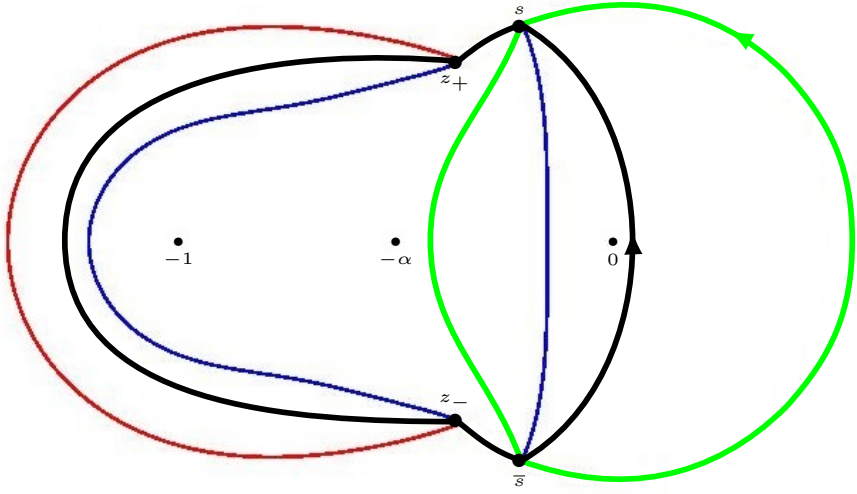

Because of \eqref{eq:Phiwforxi0} and \eqref{eq:Phizforxi0} we see that
\eqref{eq:PhizminPhiwforxi0} becomes exponentially small as $N \to \infty$
unless $w \in \gamma_w \cap \Sigma_0$ and $z \in \gamma_z \cap \Sigma_0$.
Here we also use that $\Re \phi(w) < 0$ for $w \in \gamma_w$, $|w| > \sqrt{\alpha}$,
and that $T$ and $T^{-1}$ remain bounded as $N \to \infty$ if we stay away from 
the branch points, see  Proposition \ref{prop:TandTinvsmall} (b).

On $\gamma_z \cap \Sigma_0$ we use the identity
\begin{equation} \label{eq:Tjumpsplit} 
T_+(z) \begin{pmatrix} 1 \\ 0 \end{pmatrix} 
= e^{-2N \phi_+(z)} T_+(z) \begin{pmatrix} 0 \\ 1 \end{pmatrix}
- T_-(z) \begin{pmatrix} 0 \\ 1 \end{pmatrix}, \quad z \in \Sigma_0, \end{equation}
which follows from the jump \eqref{eq:Tjump1} of $T$ across $\Sigma_0$.
Using \eqref{eq:Tjumpsplit} in \eqref{eq:PhizminPhiwforxi0} we split the integral
over $\gamma_z \cap \Sigma_0$ into a sum of two integrals, and deform both
integrals away from $\Sigma_0$. 

The integral with the first term of the right-hand side of \eqref{eq:Tjumpsplit} is deformed to the interior,
that is to a contour from $\overline{s}$ to $s$ lying inside the disk $|z| = \sqrt{\alpha}$.
The dominant part of the integrand is $e^{N(\Phi_{\alpha}(z) - 2 \phi(z))}$
and $\Re \Phi_{\alpha}(z) > \Re \Phi_{\alpha}(s)$  and $\Re \phi(z) > 0$ for $z$ on the deformed contour.
Fortunately, $\Re (\Phi_{\alpha}(z) - 2 \phi(z)) < \Re \Phi_{\alpha}(s)$, and this can be seen as follows.
By  \eqref{Phidef} and \eqref{Psidef} we have
$\Phi_{\alpha} - 2 \phi = \Psi_{\alpha}$. Since $\xi = 0$ we also find from \eqref{Phidef}
and \eqref{Psidef} that $\Phi_{\alpha} + \Psi_{\alpha} = - 2\eta \log z$.
Thus indeed
\begin{equation*}
\Re \Psi_{\alpha}(z) = - \Re \Phi_{\alpha}(z) - 2 \eta \log |z| < - \Re \Phi_{\alpha}(s) - 2 \eta \log |z| < \Re \Phi_{\alpha}(s) = -\eta \log \sqrt{\alpha} 
\end{equation*}
for $z$ on the deformed contour, since $\Re \Phi_{\alpha}(z) > \Re \Phi_{\alpha}(s)$\and $|z| < \sqrt{\alpha} < 1$
there. We also use $\eta < 0$. Thus the deformed integral coming from the first
term of \eqref{eq:Tjumpsplit} becomes small as $N \to \infty$.

The integral with the second term is moved outwards, again to a contour from $\overline{s}$
to $s$ but now lying in $|z| > \sqrt{\alpha}$. Since $\Phi_{\alpha,+} = \Psi_{\alpha,-}$
the deformed integral has the exponentially varying factor $e^{N \Psi_{\alpha}}$.
The contour can be taken such that $\Re \Psi_{\alpha}(z) < 0$ on the contour (see Figure \ref{fig:xi=0 and eta<0}, right),
and again the contribution becomes small as $N \to \infty$.

The integral (in the $w$-variable) over  $\gamma_w \cap \Sigma_0$ can be dealt with analytic continuation only. We note that by \eqref{eq:Tjump1} 
\[ \begin{pmatrix} 1 & 0 \end{pmatrix} T_{+}^{-1}(w)
= \begin{pmatrix} e^{-2N \phi_-(w)} & -1 \end{pmatrix} T_{-1}^{-1}(w) \]
which remains bounded if we analytically continue it  to the exterior of $\Sigma_0$.
We deform $\gamma_w \cap \Sigma_0$ to a contour from $s$ to $z_+(\alpha)$ lying in
the exterior of $\gamma_0$ together with its mirror image in the real, which is
a contour from $z_-(\alpha)$ to $\overline{s}$. 
Since $\Phi_{\alpha,+}(w) = \Psi_{\alpha,-}(w)$ on $\Sigma_0$, the main term in
the analytic continuation of \eqref{eq:PhizminPhiwforxi0} across $\gamma_w \cap \Sigma_0$
becomes $e^{-N \Psi_{\alpha}(w)}$. We are
able to deform contours such that $\Re \Psi_{\alpha}(w) > 0$ on the deformed
contour (from Figure \ref{fig:xi=0 and eta<0}, right), where we also take note of the local behavior near the saddle points 
$s$ and $\overline{s}$. The result is that the integral over the deformed contour
becomes small as $N \to \infty$.

What remains are local contributions near the saddles $s$ and $\overline{s}$ and
also near the branch points $z_{\pm}(\alpha)$, since we cannot move $\gamma_w$
away from the branch points.
The contributions from the saddles can be estimated as was done in detail for
the low temperature regime with $\eta < \frac{\xi}{2} < 0$. The contributions 
from the branch points are estimated similarly, but we have to note that 
$T^{-1}(w) = \bigO(N^{1/6})$ for $w$ close to the branch points, see 
Proposition \ref{prop:TandTinvsmall} (b). This slight increase however
still leads to a decay in the estimate and the conclusion is that  all
contributions vanish as $N \to \infty$.

\subsubsection{Case $\xi = \eta  =0$}

For $\xi = \eta =0$ we are at the center of the hexagon. The center belongs
to the liquid region only in the high temperature regime, and so this is
what we assume.  For $\xi = \eta =0$ the saddle coalesces with the branch point 
and the analysis requires additional deformation of contours. Note that by \eqref{Phidef}
we have
\[ \Phi_{\alpha}(z) = \phi(z) \qquad \text{ for } \xi = \eta = 0, \]
and $\Re \Phi_{\alpha}(s) = 0$ where $s = s(0,0;\alpha) = z_+(\alpha)$.

We approach this case as a limit of $(\xi, \eta) \in \mathcal L_{\alpha}$ with
$\eta \leq \frac{\xi}{2} < 0$. In this limit the contours from Corollary 
\ref{cor:contoursexist} (b) tend to contours $\gamma_w$ and $\gamma_z$ that
we may take as follows
\begin{itemize}
	\item $\gamma_w$ contains $\Sigma_{-1}$ and its analytic continuation (which is
	a critical orthogonal  trajectory, see Figure \ref{fig: crit traj alpha 03}) 
	such that
	\begin{align} 
	\Re \Phi_{\alpha}(w) > 0,  & \quad w \in \gamma_w \setminus \Sigma_{-1}. \\
	\Re \Phi_{\alpha}(w) = 0, & \quad w \in \Sigma_{-1}. 
	\end{align}
	\item $\gamma_z = \gamma_0$ and
	\begin{align} 
	\Re \Phi_{\alpha}(z) < 0,  & \quad z \in \gamma_z \setminus \Sigma_{0}. \\
	\Re \Phi_{\alpha}(z) = 0, & \quad z \in \Sigma_0.
	\end{align} 
\end{itemize}

The integrand of the double integral in \eqref{eq:deformationhigh} 
behaves like \eqref{eq:PhizminPhiwforxi0} as $N \to \infty$. With the above
choice of contours the integrand is exponentially small unless $w \in \Sigma_{-1}$
and $z \in \Sigma_0$. The case $z \in \Sigma_0$ is handled using the identity
\eqref{eq:Tjumpsplit} that we also used in the case $\xi = 0$, $\eta < 0$. It
allows us to split the integral into two integrals, deform one of them outwards
and the other one inwards, and both deformed integrals have exponentially decaying
integrands.

For  $w \in \Sigma_{-1}$ we use the second line of \eqref{eq:PhizminPhiwforxi0}
which tells us that the main $w$-dependent part is 
\[ e^{-N \Phi_{\alpha}(w)} \begin{pmatrix} 1 & - e^{2N \phi(w)} \end{pmatrix} T^{-1}(w) \]
which naturally splits into a sum (recall also $\Phi_{\alpha} = \phi$)
\begin{equation} 
e^{-N \phi(w)} \begin{pmatrix} 1 & 0 \end{pmatrix} T^{-1}(w) 
- e^{N \phi(w)} \begin{pmatrix} 0 & 1 \end{pmatrix} T^{-1}(w)
\end{equation}
and a corresponding splitting and deformation of the $w$-integral.
Namely the integral with the first term is deformed from $\Sigma_{-1}$ to a contour
from $z_+(\alpha)$ to $z_-(\alpha)$ lying outside $\Sigma_{-1}$ (where $\Re \phi > 0$) 
and the integral with the second term is deformed inwards (where $\Re \phi < 0$).

Then there is exponentially decay on the deformed contours as $N \to \infty$,
except for $w$ and $z$ near the branch points $z_{\pm}(\alpha)$. $T$ and $T^{-1}$
have moderate growth there, both of $\bigO(N^{1/6})$. They combine
to give an increase in $T^{-1}(w) T(z)$ of $\bigO(N^{1/3})$. Local
estimates still lead to a decay in the integrals, as required.

This completes the proof of Proposition \ref{prop:doubleintegrallimit} in all cases.

\subsection{Proof of Theorem \ref{thm:microsopiclimit}}

\begin{proof}
	
	With the coordinates in \eqref{eq:microsopic_variables}  (and the fact that $N \xi_N$ is assumed to be even) we can rewrite the kernel $K_{N}$ in \eqref{eq:kernel} as 
	\begin{align} \label{eq:microscopickernel}
K_N(x_1,y_1,x_2,y_2)= -\frac{\chi_{u_1>v_2}}{2 \pi i} \oint_\gamma H_K(z,z;u_1,v_1,u_2,v_2) dz  
+ \mathcal I_N(N \xi_N,N \eta_N; H_{K})
	\end{align}
	where $\mathcal I_N$ is as in \eqref{eq:IxyH} with
	$$ H_K(w,z;u_1,v_1,u_2,v_2)= \frac{(z+1)^{\lfloor \frac{u_1}{2}\rfloor}(z+\alpha)^{\lfloor \frac{u_1+1}{2}\rfloor }}{(w+1)^{\lfloor \frac{u_2}{2} \rfloor}(w+\alpha)^{\lfloor \frac{u_2+1}{2}\rfloor}} \frac{w^{v_2}}{z^{v_1+1}}.$$
	
	The first integral in \eqref{eq:microscopickernel} is independent of $N$. The asymptotic behavior of $\mathcal I_N(N \xi_N,N \eta_N; H_{K})$ as $N\to \infty$ is already  computed in 
	Proposition \ref{prop:doubleintegrallimit}. The first integral and the limit from Proposition \ref{prop:doubleintegrallimit} can be combined naturally into one single integral, which is the right-hand side of \eqref{eq:HKintegral}. This finishes the proof.\end{proof}
	\appendix

\section{Proof of Proposition \ref{proposition1.1}}
\begin{proof}[Proof of Proposition \ref{proposition1.1}]
	This is a special case of Theorem 4.7 in \cite{DK}.
	To identify the formula in \cite{DK} with \eqref{eq:kernel},
	we first of all note that $p=1$ and $K_N$ is a scalar kernel. 
	We have to identify $(m,x,m',y)$  and $(N,M,L)$ in \cite{DK} 
	with $(x_1,y_1,x_2,y_2)$ and $(N,N,2N)$ in the setting of our paper.
	
	Furthermore, for $0 \leq i < j \leq 2N$, the notation
	$A_{i,j}(z)$ in \cite{DK} stands for $A_{i,j}(z) = \ds \prod_{m=i}^{j-1} a_m(z)$
	where $a_m(z) = z+\alpha$ if $m$ is even, 
	and $a_m(z) = z+1$ if $m$ is odd. This gives
	\[ A_{x_2, x_1}(z) =  (z+1)^{\lfloor \frac{x_1}{2} \rfloor -\lfloor \frac{x_2}{2} \rfloor} (z+\alpha)^{\lfloor \frac{x_1+1}{2} \rfloor -\lfloor 
		\frac{x_2+1}{2}\rfloor}
	\] 
	which appears in the single integral in \eqref{eq:kernel},
	and similarly
	\begin{align*}
	A_{x_2,2N}(w) & = 
	(w+1)^{N -\lfloor \frac{x_2}{2} \rfloor} (w+\alpha)^{N - \lfloor \frac{x_2+1}{2}\rfloor}  \\
	A_{0,x_1}(z) & =  (z+1)^{\lfloor \frac{x_1}{2} \rfloor} (z+\alpha)^{\lfloor \frac{x_1+1}{2} \rfloor}
	\end{align*}
	which is part of the double integral in \eqref{eq:kernel}.
	
	Finally, according to \cite[Theorem 4.7]{DK}, $R_N$ is the reproducing kernel for polynomials of degree $\leq N-1$ 	with weight 
	$ \ds \frac{A_{0,L}(z)}{z^{M+N}} = \frac{(z+1)^N(z+\alpha)^N}{z^{2N}}$ on
	$\gamma$, as $M=N$ and $L=2N$.
	It means that $R_N(w,z)$ is a bivariate polynomial of degree $\leq N-1$
	in both variables that is uniquely characterized by the property  that
	\begin{equation} \label{eq:RNreproducing} 
	\frac{1}{2 \pi i} \oint_{\gamma} R_N(w,z) \frac{(z+1)^N(z+\alpha)^N}{z^{2N}}
	q(z) dz = q(w) \end{equation}
	for every polynomial $q$ of degree $\leq N-1$, see Lemma 4.6 in \cite{DK}.
	Since all orthogonal polynomials $p_n$ of degrees $n \leq 2N$ exist (we prove
	this in Proposition \ref{prop:prop51}),	the sum in \eqref{eq:CDkernel} is well-defined, and 
	by orthogonality using \eqref{eq:kappan} it defines a kernel 
	with the required reproducing property \eqref{eq:RNreproducing}.
	
	The expression in the second line of \eqref{eq:CDkernel} is 
	known as the Christoffel-Darboux formula,
	and it continues to hold for non-Hermitian orthogonality on a contour, with
	the same proof as for usual orthogonal polynomials on the real line.
\end{proof}

\end{document}